\documentclass[DIV12]{scrartcl}
\usepackage{etex}
\usepackage[utf8]{inputenc}
\usepackage[T1]{fontenc}
\usepackage{microtype}
\usepackage{mathptmx}
\DeclareMathAlphabet{\mathcal}{OMS}{cmsy}{m}{n}
\usepackage{caption}
\usepackage{tikz}
\usepackage{forest}
\usepackage{color}

\author{
\begin{tabular}[t]{c}
Antoine Amarilli \\
{\normalfont T\'el\'ecom ParisTech; Institut Mines--T\'el\'ecom; CNRS LTCI} \\
{\normalfont antoine.amarilli@telecom-paristech.fr} \\[0.5em]
Michael Benedikt \\
{\normalfont University of Oxford} \\
{\normalfont michael.benedikt@cs.ox.ac.uk}
\end{tabular}
}

\usepackage{amssymb, amsmath,  graphicx}
\usepackage{url}
\usepackage{xspace}
\usepackage{paralist}
\usepackage{graphicx}
\usepackage[hidelinks]{hyperref}
\usepackage{tabularx}
\usepackage{booktabs}
\usepackage{amsthm}

\newcommand{\kw}[1]{{\mathsf{#1}}\xspace}
\newcommand{\danger}{\kw{Dng}}

\newcommand{\nondanger}{\kw{NDng}}
\newcommand{\positions}{\kw{Pos}}
\newcommand{\pos}{\positions}

\newcommand{\appelem}[2]{\kw{Wants}(#1, #2)}
\newcommand{\chase}[2]{\kw{Chase}(#1, #2)}
\newcommand{\cov} {\kw{sim}}

\newcommand{\bsim}{\leq}
\newcommand{\bbsim}{\simeq}

\renewcommand{\a}{\mathrm{a}}

\newcommand{\f}{\mathrm{f}}
\renewcommand{\l}{\mathrm{l}}
\newcommand{\m}{\mathrm{m}}
\newcommand{\n}{\mathrm{n}}
\newcommand{\p}{\mathrm{p}}
\renewcommand{\r}{\mathrm{r}}
\renewcommand{\t}{\mathrm{t}}

\newcommand{\w}{\mathrm{w}}

\newcommand{\eqfun}{\leftrightarrow_{\mathrm{FUN}}}

\newcommand{\eqids}{\sim_{\mathrm{\incd}}}
\newcommand{\eqidsclass}[1]{[#1]_{\mathrm{\incd}}}
\newcommand{\eqrev}{\eqids} % not distinguished from \eqids

\newcommand{\ui}[2]{#1 \subseteq #2}

\newcommand{\ovl}{\kw{OVL}}
\newcommand{\overlap}{\ovl}

\newcommand{\myparagraph}[1]{\paragraph{#1.}}
\newcommand{\myeat}[1]{}

\newcommand{\calA}{\mathcal{A}}

\newcommand{\calE}{\mathcal{E}}
\newcommand{\calF}{\mathcal{F}}

\newcommand{\calH}{\mathcal{H}}

\newcommand{\calN}{\mathcal{N}}

\newcommand{\idgraph}{\Gamma}

\newcommand{\NN}{\mathbb{N}}

\newcommand{\fd}{\kw{FD}}
\newcommand{\uid}{\kw{UID}}

\newcommand{\ufd}{\kw{UFD}}
\newcommand{\incd}{\kw{ID}}
\newcommand{\cq}{\kw{CQ}}
\newcommand{\acq}{\kw{ACQ}}

\DeclareMathOperator{\id}{id}

\DeclareMathOperator{\pr}{\mathfrak{pr}}
\DeclareMathOperator{\parts}{\mathfrak{P}}

\DeclareMathOperator{\dom}{dom}

\newcommand{\entfin}{\models_\mathrm{fin}}
\newcommand{\entunr}{\models_\mathrm{unr}}

\newcommand{\complexity}[1]{\textsf{\mdseries\upshape #1}}

\newcommand\np{\textup{NP}\xspace}

\newcommand{\fds}{\Sigma_{\mathrm{FD}}}
\newcommand{\ufds}{\Sigma_{\mathrm{UFD}}}
\newcommand{\ids}{\Sigma_{\mathrm{UID}}}
\newcommand{\con}{\Sigma}
\newcommand{\ucon}{\Sigma_{\mathrm{U}}}

\newcommand{\idprec}{\rightarrowtail}
	
\newcommand\restr[2]{{% we make the whole thing an ordinary symbol
  \kern-\nulldelimiterspace % automatically resize the bar with \right
  #1 % the function
  _{|#2} % this is the delimiter
  }}

\newcommand{\fdrestr}[2]{#1^{#2}}
\newcommand{\relrestr}[2]{\restr{#1}{#2}}

\newcommand{\quot}[2]{#1/{#2}}

\newcommand{\card}[1]{\left|#1\right|}

\newcommand{\mapb}{\lambda}

\newcommand{\deft}[1]{\textbf{#1}}

\newcommand{\defo}[1]{\emph{#1}}

\newcommand{\defp}[1]{\emph{#1}}

\newcommand{\arity}[1]{\card{#1}}

\newcommand{\assm}[1]{{\selectfont\textsf{#1}}}
\newcommand{\invar}[1]{{\selectfont\textsf{#1}}}

\newcommand{\neqfunc}{n}
\newcommand{\neqidsc}{n}

\newcommand{\rfcl}{\rdfcl}
\newcommand{\rdfcl}{\kw{AFactCl}}

\newcommand{\lab}[1]{\Lambda(#1)}

\newcommand{\sprod}{\otimes} % simple product
\newcommand{\mprod}{\otimes} % mixed product

\newcommand{\pire}{PI} % piecewise realization

\newcommand{\mysubsection}[1]{\smallskip\subsection{\textbf{#1}}\smallskip}

\newcommand{\sccgraph}{G}

\newcommand{\mybf}[1]{\text{\textit{\textbf{#1}}}}

\mathchardef\breakingcomma\mathcode`\,
{\catcode`,=\active
    \gdef,{\breakingcomma\discretionary{}{}{}}
}

\usepackage{thm-restate}
\usepackage{thm-restate-fix}
\declaretheorem[name=Lemma,numberwithin=section]{lemma}

\declaretheorem[name=Theorem,sibling=lemma]{theorem}
\declaretheorem[name=Example,sibling=lemma]{example}
\declaretheorem[name=Corollary,sibling=lemma]{corollary}
\declaretheorem[name=Definition,sibling=lemma]{definition}

\makeatletter
\newcommand*{\defeq}{\mathrel{\rlap{\raisebox{0.3ex}{$\m@th\cdot$}}\raisebox{-0.3ex}{$\m@th\cdot$}}=}
\makeatother

\begin{document}

\title{Finite Open-World Query Answering\\with Number Restrictions\\(Extended Version)}

\maketitle

\begin{abstract}
Open-world query answering  is the problem of deciding, given a set of facts,
conjunction of constraints, and query, whether the 
facts and constraints imply the  query.
This amounts to 
reasoning over all instances that include the facts and satisfy the constraints.
We study \emph{finite open-world query answering} (FQA), which assumes that the underlying world
is finite and thus only considers the \emph{finite}
completions of the instance.
The major known decidable cases of FQA
  derive from the following: the guarded fragment of first-order logic, which can express referential constraints (data in one
place points to data in another) but cannot express number restrictions such
as functional dependencies; and the guarded fragment with number restrictions but
on a signature of arity only two.
  In this paper, we give the first decidability results for FQA that
 combine both referential constraints and number
restrictions
 for arbitrary
signatures: we show that, for unary inclusion dependencies and functional
dependencies, the finiteness assumption of FQA can be lifted up to taking the finite implication closure of
the dependencies~\cite{cosm}.
Our result relies on new techniques to construct finite universal models of such
constraints, for any bound on the maximal query size.
% 
%   Keywords: query answering, open world, decidability, finite controllability,
%   functional dependencies, inclusion dependencies, guarded fragment, chase.
% 
%   Topics: incompleteness inconsistency uncertainty, constraints, model theory
%   logics algebras computational complexity.
\end{abstract}

\section{Introduction}
\label{sec:intro}
A longstanding goal in computational logic is to design logical languages that are both decidable and expressive.
One approach is to distinguish integrity constraints and queries, and have
separate languages for them.
We would then seek  
decidability of the \emph{query answering with constraints} problem:
given a query $q$,  a conjunction of constraints $\con$, and a finite instance~$I$,
determine which answers
to~$q$ are certain to hold over any instance~$I'$ that extends $I$ and
satisfies $\con$.
This problem is often called \emph{open-world query answering}. It is 
fundamental for deciding
query
containment under constraints, querying in the presence of ontologies, or
reformulating queries with
constraints. Thus
 it has been the subject
of intense study within several communities for decades (e.g. \cite{johnsonklug, cali2003decidability, bgo,
pratt2009data, ibanezgarcia2014finite}).

In many cases (e.g., in databases)
 the instances $I'$ of interest  are the finite ones, and hence we can define
 \emph{finite open-world query answering} (denoted here as FQA),
which restricts the quantification
to \emph{finite} extensions  $I'$ of~$I$. In contrast, by \emph{unrestricted open-world query answering} (UQA) we refer to
the problem where $I'$ can be either finite or infinite.
Generally the class of queries is taken to be the conjunctive queries ($\cq$s) --- queries built up
from relational atoms via existential quantification and conjunction. We will restrict to $\cq$s here, and
thus omit explicit mention of the query language, focusing on the constraint language.

A first constraint class known to have tractable open-world query answering problems
are \emph{inclusion dependencies} ($\incd$s) --- constraints of the form, e.g.,
$\forall x y z ~ R(x, y, z) \rightarrow \exists v w ~ S(z, v, w, y)$.
The fundamental results of Johnson and Klug~\cite{johnsonklug} and  Rosati~\cite{rosatifinitecontrol}
show that both FQA and UQA are decidable for~$\incd$ and that, in fact, they coincide. When this happens, the constraints
are said to be \emph{finitely controllable}.
These results have been generalized by B{\'a}r{\'a}ny et al. \cite{bgo}
to a much richer class of constraints,  the guarded fragment of first-order logic.

However, those results do not cover a second 
important kind of constraints, namely \emph{number restrictions}, which express, e.g.,
uniqueness. We represent them by the class of
\emph{functional dependencies} ($\fd$s) --- 
of the form
$\forall \mybf{x} \mybf{y} ~ (R(x_1, \ldots, x_n) \wedge R(y_1, \ldots, y_n) \wedge
\bigwedge_{i\in L} x_i = y_i) \rightarrow x_r = y_r$.
The implication problem (does one $\fd$ follow from a set of others) is
decidable for $\fd$s, and coincides with implication
restricted to finite instances~\cite{AHV}. Trivially, the FQA and UQA problems
are also decidable for $\fd$s alone, and coincide.

Trying to combine $\incd$s and $\fd$s makes both UQA and FQA undecidable in
general~\cite{cali2003decidability}. However, UQA 
is known to be decidable when the $\fd$s and the $\incd$s are
\emph{non-conflicting}
\cite{johnsonklug,cali2003decidability}. 
Intuitively, this condition guarantees that the $\fd$s can be ignored, as long as they
hold on the initial instance~$I$, and one can then solve the query answering
problem by considering the $\incd$s alone.
But the non-conflicting condition 
only 
applies to UQA and not to FQA.
In fact it is known that even for very simple classes of~$\incd$s and $\fd$s,  including
non-conflicting classes, 
FQA and UQA do not coincide.
Rosati~\cite{rosatifinitecontrol} showed that FQA is  undecidable 
 for non-conflicting $\incd$s and $\fd$s (indeed, for $\incd$s and keys, which are less rich
than $\fd$s).

Thus a  general question is to what extent these classes,
$\fd$s and $\incd$s,  can be combined while retaining decidable FQA.
The only  decidable cases impose very
severe requirements. For example,
the constraint class of ``single KDs and FKs'' introduced
in~\cite{rosatifinitecontrol} has decidable FQA, but
such constraints cannot model, e.g., $\fd$s which are not keys.
Further, in contrast with the general case of $\fd$s and $\incd$s,
single KDs and FKs are always finitely controllable, which limits their
expressiveness. 
Indeed, we know of no tools to deal with FQA for non-finitely-controllable constraints  on relations of arbitrary arity.

A second decidable case is  where all relation symbols
and all subformulas of the constraints have arity at most two. In this context,
results of Pratt-Hartmann \cite{pratt2009data} imply the
decidability of both FQA and UQA
for a very rich non-finitely-controllable sublogic of first-order logic.
For some fragments of this arity-two logic, the complexity of FQA has recently
been isolated by Ib{\'a\~n}ez-Garc{\'i}a et al.
 \cite{ibanezgarcia2014finite}. Yet these results do not apply to arbitrary
 arity signatures.

 \emph{The contribution of this paper is to provide the first result
about
finite query
answering for non-finitely-controllable $\incd$s and $\fd$s over relations of
arbitrary arity.}
As the problem is undecidable in general, we must naturally make some restriction. Our choice is to limit to 
\emph{Unary $\incd$s} ($\uid$s), which export
only one variable: for instance, $\forall x y z ~ R(x, y, z) \rightarrow \exists w ~ S(w, x)$.
$\uid$s and $\fd$s are an interesting class to study because they are not
finitely controllable, and allow the
modeling, e.g., of single-attribute foreign keys, a common use case in database
systems.
The decidability of UQA for $\uid$s and $\fd$s is known because they
are always non-conflicting. In this paper, we show that
finite query answering is decidable for $\uid$s and $\fd$s, and obtain tight bounds on
its complexity.

The idea is to \emph{reduce the finite
case to the unrestricted case}, but in a more complex way than by finite
controllability.
We make use of a technique
originating in Cosmadakis et al. \cite{cosm} to study finite implication on $\uid$s and $\fd$s:
the 
\emph{finite closure} operation
which takes
 a conjunction of~$\uid$s and $\fd$s and determines exactly which additional
 $\uid$s and $\fd$s are
implied
over finite instances. 
Rosati~\cite{rosatieswc} and Ib\'a\~nez-Garc\'ia~\cite{ibanezgarcia2014finite}  make use of the closure
operation in their study of constraint classes over schemas of arity two. They
show that
finite query answering for a query $q$, instance~$I$, and constraints $\con$ reduces to 
unrestricted query answering for~$I$, $q$, and the finite closure~$\con'$ of
$\con$. In other words, the closure construction which is sound for implication is also sound for query answering.

We show that the same general approach applies to arbitrary arity signatures,
with
constraints being $\uid$s and $\fd$s. Our main result thus reduces finite query answering to unrestricted query answering,
for $\uid$s and $\fd$s in arbitrary arity:

\begin{theorem}
  \label{thm:main2}
  For any finite instance~$I$,
  conjunctive query~$q$,  
  and constraints $\con$ consisting of~$\uid$s and $\fd$s,
 the finite open-world query answering problem for $I, q$ under $\con$ has the same answer
as the unrestricted open-world query answering problem for $I, q$ under the
finite closure of~$\con$.
\end{theorem}

Using the known results about the complexity of UQA for $\uid$s,
we isolate the precise complexity
of finite query answering  with respect to $\uid$s and $\fd$s, showing that it matches
that of UQA:

\begin{corollary} \label{cor:comp} The combined complexity of the finite open-world query answering problem for
  $\uid$s and $\fd$s is $\np$-complete, and it is PTIME in data complexity (that
  is, when the constraints and query are fixed).
\end{corollary}

Our proof of Theorem~\ref{thm:main2} is quite involved, since dealing with
arbitrary arity models
introduces many new difficulties
that do not arise in
the arity-two case or in the case of $\incd$s in isolation.
We borrow and adapt a variety of techniques from prior work: using $k$-bounded
simulations to preserve small acyclic $\cq$s~\cite{ibanezgarcia2014finite}, dealing
with $\uid$s following a topological sort~\cite{cosm,ibanezgarcia2014finite},
performing a chase that reuses sufficiently similar
elements~\cite{rosatifinitecontrol}, and taking the product with groups of large
girth to blow up cycles~\cite{otto2002modal}.
However, we must also develop some  new infrastructure to deal with number restrictions
in an arbitrary arity setting: distinguishing between so-called \defo{dangerous}
and \defo{non-dangerous} positions when chasing, constructing realizations for
relations in a \emph{piecewise} manner following the $\fd$s, reusing elements
in a \emph{combinatorial} way that shuffles them to avoid violating the
higher-arity $\fd$s,
and a new notion of \defo{mixed product} to blow cycles
up while preserving fact overlaps to avoid violating the higher-arity $\fd$s.

\myparagraph{Paper structure}
The general scheme, presented in Section~\ref{sec:overall}, is to construct models of~$\uid$s and $\fd$s that are
universal up to a certain query size~$k$, which we call \emph{$k$-universal models}.
We start with only \emph{unary} $\fd$s
($\ufd$s) and \emph{acyclic} $\cq$s ($\acq$s), and by assuming that the $\uid$s
and $\ufd$s are \emph{reversible}, a condition inspired by the 
finite closure construction.

As a warm-up, Section~\ref{sec:wsound}
proves the weakened result for a much weaker notion than $k$-universality,
starting with binary signatures and generalizing to arbitrary arity.
We extend the result to $k$-universality in Section~\ref{sec:ksound}, 
maintaining a $k$-bounded simulation to the chase, and performing \emph{thrifty} chase steps that reuse
sufficiently similar elements without violating $\ufd$s. We also
rely on a structural observation about the chase under $\uid$s
(Theorem~\ref{thm:locality}).
Section~\ref{sec:manyscc} eliminates the assumption that dependencies are
reversible, by
partitioning the $\uid$s into classes that are either reversible or trivial,
and satisfying successively each class following a certain ordering.

We then generalize our result to higher-arity (non-unary) $\fd$s in
Section~\ref{sec:hfds}.
This requires us to define a new notion of thrifty chase steps
that apply to instances with many ways to reuse elements; the
existence of these instances relies on a combinatorial construction of models
of $\fd$s with a high number of facts but a small domain (Theorem~\ref{thm:combinatorial}).
Last, in Section~\ref{sec:cycles}, we apply a cycle blowup process to the result
of the previous constructions, to go from acyclic to arbitrary $\cq$s through a
product with acyclic groups.
The technique is inspired by Otto~\cite{otto2002modal} but
must be adapted to respect $\fd$s. 

Complete proofs of our results are provided in the appendix.

\section{Background}
\label{sec:prelim}
\myparagraph{Instances}
We assume an infinite countable set of \defp{elements} (or \defp{values}) $a,
b, c, \ldots$ and \defp{variable names} $x, y, z, \ldots$. 
A 
\defp{schema}~$\sigma$ consists of \defp{relation names} (e.g., $R$) with an 
\defp{arity} (e.g., $\arity{R}$)
which we assume is $\geq 1$. Following the unnamed perspective, the set of
\defp{positions} of~$R$ is
$\positions(R) \defeq \{R^i \mid 1 \leq i \leq \arity{R}\}$, and we define
$\positions(\sigma) \defeq \bigsqcup_{R \in \sigma} \pos(R)$.
We identify $R^i$ and $i$ when no confusion can result.

A relational \defp{instance} (or \defp{model}) $I$ of~$\sigma$ 
is a set of \defp{ground facts} of the form $R(\mybf{a})$ where $R$ is a
relation name and $\mybf{a}$ an $\arity{R}$-tuple of values.
The \defp{size} $\card{I}$ of an instance $I$ is its number of facts.
The
\defp{active domain} $\dom(I)$ of~$I$ is
the set of the elements which appear in $I$. For any position $R^i \in
\positions(\sigma)$, we define the \defp{projection} $\pi_{R^i}(I)$ of~$I$ to
$R^i$ as the
set of the elements of~$\dom(I)$ that occur at position $R^i$ in $I$.
For $L \subseteq \positions(R)$,
the projection~$\pi_L(I)$ is a set of~$\card{L}$-tuples defined analogously;
for convenience, departing from the unnamed perspective,
we index those tuples by the positions of~$L$.
A \defp{superinstance} 
 of~$I$ is a (not necessarily finite) instance $I'$ such
that $I \subseteq I'$. 

A \defp{homomorphism} from an instance $I$ to an instance $I'$ is a mapping
$h : \dom(I) \rightarrow \dom(I')$ such that, for every fact $F = R(\mybf{a})$
of $I$, the fact $h(F) \defeq R(h(a_1), \ldots, h(a_{\card{R}}))$ is in~$I'$.

\myparagraph{Constraints}
We consider integrity constraints (or \defp{dependencies}) which are special sentences of first-order
logic.
As usual in the relational setting, we do not allow function symbols.
The definition of an instance $I$ satisfying a constraint~$\con$, written $I \models \con$, is standard.

An \defp{inclusion dependency} $\incd$ is a sentence of the form
$\tau: \forall \mybf{x} \, R(x_1, \ldots, x_n) \rightarrow \exists
\mybf{y} \, S(z_1, \ldots, z_m)$,
where $\mybf{z} \subseteq \mybf{x} \cup \mybf{y}$ and no variable occurs
twice in $\mybf{z}$.
The \defp{exported variables} are the variables of $\mybf{x}$ that occur in
$\mybf{z}$,
and
the \defp{arity} of the dependency is the number of such variables.
This work only studies 
\defp{unary inclusion dependencies} ($\uid$s) which are the $\incd$s
with arity~$1$.
If $\tau$ is a $\uid$, we write $\tau$ as $\ui{R^p}{S^q}$, where $R^p$ and $S^q$ are the
positions of~$R(\mybf{x})$ and $S(\mybf{z})$ where the exported variable occurs.
For instance, the $\uid$ $\forall x y \, R(x, y) \rightarrow \exists z \, S(y, z)$ is
written $\ui{R^2}{S^1}$.
We assume without loss of generality that there are no \defp{trivial} $\uid$s of the form $\ui{R^p}{R^p}$.

We say that a conjunction $\ids$ of~$\uid$s is \defp{transitively closed} if it is
closed under implication by the \defp{transitivity rule}: if
$\ui{R^p}{S^q}$ and $\ui{S^q}{T^r}$ are in $\ids$, then so is $\ui{R^p}{T^r}$
unless it is trivial. The
transitive closure of~$\ids$ can clearly be computed in PTIME in $\ids$, and it
contains all non-trivial $\uid$s implied by $\ids$ over finite or
unrestricted instances~\cite{casanova}.
We say a $\uid$ $\tau : \ui{R^p}{S^q}$ is \defp{reversible} relative to $\ids$ if both 
$\tau$ and its \defp{reverse} $\tau^{-1} \defeq \ui{S^q}{R^p}$ are in~$\ids$.

 A \defp{functional dependency} $\fd$
is a  sentence of the form 
$\phi: \forall \mybf{x} \mybf{y} \, (R(x_1, \ldots, x_n) \wedge
R(y_1, \ldots, y_n) \wedge \bigwedge_{R^l \in L} x_l = y_l) \rightarrow x_r = y_r$,
where $L \subseteq \positions(R)$ and $R^r \in \positions(R)$.
For brevity, we write~$\phi$ as $R^L \rightarrow R^r$.
We call $\phi$ a \defp{unary functional dependency} $\ufd$ if
$\card{L} = 1$; otherwise it is \defp{higher-arity}.
For instance, $\forall x x' y y' \, R(x, x') \wedge R(y, y') \wedge
x' = y' \rightarrow x = y$ is a $\ufd$, and we write it $R^2 \rightarrow R^1$.
We assume that
$\card{L} > 0$, i.e., we do not allow nonstandard or degenerate
$\fd$s.
We
call $\phi$ \defp{trivial} if $R^r \in R^L$, in which case $\phi$ always
holds.
Two facts $R(\mybf{a})$ and $R(\mybf{b})$ \defp{violate} a non-trivial $\fd$ $\phi$ if
$\pi_{L}(\mybf{a}) = \pi_{L}(\mybf{b})$
but $a_r \neq b_r$.

The \defp{key dependency} $\kappa: R^L \rightarrow R$, for $L \subseteq \pos(R)$, is
the conjunction of $\fd$s $R^L \rightarrow R^r$
for all $R^r \in \pos(R)$; it is \defp{unary} if $\card{L} = 1$. If $\kappa$
holds, we call $L$ a \defp{key} (or \defp{unary key}) of~$R$.

\myparagraph{Queries}
An \defp{atom} $A = R(\mybf{t})$ consists of a relation
name $R$ and a $\arity{R}$-tuple $\mybf{t}$ of variables or constants.
A \defp{conjunctive query} $\cq$ is 
an existentially quantified conjunction of atoms.
In this paper we focus for simplicity on Boolean queries (queries without
free variables),
but all our results hold for non-Boolean queries as
well, by the standard method of enumerating the assignments.
The \defp{size} $\card{q}$ of a $\cq$ $q$ is its number of atoms.

A \defp{Berge cycle} in a Boolean $\cq$~$q$ is a sequence 
$A_1, \allowbreak x_1, \allowbreak A_2, \allowbreak x_2, \allowbreak \ldots,
\allowbreak A_n, \allowbreak x_n$ with $n \geq 2$,
where the $A_i$ are pairwise distinct atoms of~$q$, the
$x_i$ are pairwise distinct variables of $q$,
and $x_i$ occurs in $A_i$ and $A_{i+1}$
for $1 \leq i \leq n$ (with addition modulo~$n$, so
$x_n$ occurs in~$A_1$).
We call~$q$ \defp{acyclic} if $q$ has no Berge cycle and if no variable
of~$q$ occurs more than once in the same atom. We write $\acq$ for the
class of acyclic $\cq$s.

A Boolean $\cq$ $q$ \defp{holds}
in an instance $I$ exactly when there is a homomorphism~$h$ from 
the atoms of~$q$ to
$I$ such that $h$ is the identity on the constants of~$q$
(we call this a \defp{homomorphism from~$q$ to~$I$}).
The image of~$h$ is called a \defp{match} of~$q$ in $I$.

\myparagraph{QA problems}
We define the \defp{unrestricted open-world query answering} problem (UQA) as
follows: given a finite instance~$I$, a conjunction of constraints~$\con$, and
a Boolean $\cq$~$q$,
decide whether there is a superinstance of~$I$ that satisfies $\con$ and
violates~$q$. If there is none, we say that $I$ and $\con$ \defp{entail} $q$
and write $(I, \con) \entunr q$.

This work focuses on the
\defp{finite query answering problem} (FQA), which is the variant of open-world query answering where we require the
counterexample superinstance to be finite; if none exists, we write $(I,
\con) \entfin q$. Of course $(I, \con) \entunr q$ implies $(I, \con)
\entfin q$.
We say a conjunction of constraints $\con$ is \defp{finitely controllable} if FQA and UQA   coincide:
for every finite instance~$I$ and every Boolean
$\cq$ $q$,
 $(I, \con) \entunr q$ iff $(I, \con)
\entfin q$.

The \defp{combined complexity} of the UQA and FQA problems, for a fixed class of
constraints, is the complexity of
deciding it when all of~$I$, $\con$ (in the constraint class) and $q$ are given
as input. The \defp{data complexity} is defined by assuming 
that $\con$ and~$q$ are fixed, and only $I$ is given as input.

\myparagraph{Chase}
We say that a superinstance $I'$ of an instance $I$ is \defp{universal} for
constraints $\con$ if $I' \models \con$ and if for any $\cq$ $q$, $I'
\models q$ iff $(I, \con) \entunr q$. We now recall the definition of the
\defp{chase}~\cite{AHV,onet2013chase}, a standard
construction of (generally infinite) universal superinstances. We assume that we
have fixed an infinite set~$\calN$ of \defp{nulls} which 
is disjoint from $\dom(I)$.
We only define the chase for transitively closed $\uid$s, which
we call the \defp{$\uid$ chase}.

We say that a fact $F_\a = R(\mybf{a})$ of an instance $I$ is an \defp{active fact}
for a $\uid$ $\tau: \ui{R^p}{S^q}$ if, writing $\tau: \forall\mybf{x} \,
R(\mybf{x}) \rightarrow \exists\mybf{y} S(\mybf{z})$,
there is a homomorphism from $R(\mybf{x})$ to $F_\a$ but no such
homomorphism can be extended to a homomorphism from
$\{R(\mybf{x}), S(\mybf{z})\}$ to~$I$.
In this case we say that $a_p$ \defp{wants} to occur at position $S^q$ in~$I$, 
written $a_p \in \appelem{I}{S^q}$,
and that we \defp{want} to apply the $\uid$~$\tau$ to $a_p$,
written $a_p \in \appelem{I}{\tau}$. Note that
$\appelem{I}{\tau} = \pi_{R^p}(I)
\backslash \pi_{S^q}(I)$.

The result of a \defp{chase step} on the active fact $F_\a$ for~$\tau$
in~$I$ (we call this \defp{applying} $\tau$ to $F_\a$) is the superinstance~$I'$
of~$I$ obtained by adding a new fact $F_\n = S(\mybf{b})$ defined as follows:
we set $b_q \defeq a_p$, which we call the \defp{exported element} (and $S^q$
the \defp{exported position} of~$F_\n$), and use fresh nulls from $\calN$ to instantiate the
existentially quantified variables of~$\tau$ and complete $F_\n$; we say
the corresponding elements are \defp{introduced} at $F_\n$.
This ensures that $F_\a$ is no longer an active fact in~$I'$ for~$\tau$.

A \defp{chase round} of a conjunction $\ids$ of~$\uid$s on $I$
is the result of applying simultaneous chase steps on
all active facts for all $\uid$s of~$\ids$, using distinct fresh
elements.
The \defp{$\uid$ chase} $\chase{I}{\ids}$ of~$I$ by $\ids$ is the (generally infinite)
fixpoint of applying chase rounds. It is a universal superinstance
for~$\ids$~\cite{fagin2003data}.

As we are chasing by transitively closed $\uid$s,  
if we perform the \defp{core chase}~\cite{onet2013chase} rather than the
$\uid$ chase defined above,
we can ensure the following
\defp{Unique Witness Property}: for any element $a \in \dom(\chase{I}{\ids})$ and position $R^p$ of~$\sigma$,
if two different facts of~$\chase{I}{\ids}$ contain~$a$ at position~$R^p$,
then they are both facts of~$I$.
In our context, however, the core chase matches the $\uid$ chase defined above, except at the
first round. Thus, modulo the first round, by $\chase{I}{\ids}$ we refer
to the $\uid$ chase, which has the Unique Witness Property.
See Appendix~\ref{apx:chase} for details.

\myparagraph{Finite closure}
Rosati~\cite{rosatifinitecontrolpods,rosatifinitecontrol} showed that, while conjunctions of~$\incd$s are
finitely controllable, even conjunctions of~$\uid$s and $\fd$s may not be.
However, Cosmadakis et al.~\cite{cosm}  showed how to decide in PTIME
the \defp{finite implication} problem for $\uid$s and $\fd$s: given a
conjunction $\con$ of such dependencies, decide whether a $\uid$ or $\fd$ is
implied by $\con$ over finite instances.
The \defp{finite closure} of~$\con$ is the set of the $\uid$s and
$\fd$s thus implied by $\con$ in the finite.

Rosati~\cite{rosatieswc} later showed that
the finite closure could be used to reduce UQA to FQA for some constraints
on relations of arity at most two. Following the same idea,
we say that a conjunction of constraints $\con$ 
is \defp{finitely controllable up to finite closure} if for every finite
instance $I$, and Boolean $\cq$~$q$,
 $(I, \con) \entfin q$
 iff $(I, \con') \entunr q$, where $\con'$ is the finite closure of~$\con$.
This implies that we can reduce FQA to UQA, even if finite controllability does not hold.

\section{Main Result and Overall Approach}
\label{sec:overall}
We study open-world query answering for $\fd$s and $\uid$s. For unrestricted
query answering (UQA), the following is already known, from bounds on UQA for
$\uid$s:

\begin{restatable}{proposition}{complexity}
  \label{prp:complexity}
  UQA for $\fd$s and $\uid$s has PTIME data complexity and NP-complete 
  combined complexity.
\end{restatable}

However, for the \emph{finite case}, even the decidability of FQA for $\fd$s and
$\uid$s is not known. Here is our main result, which is proved in the rest of
this paper:

\begin{theorem}[Main theorem]
  \label{thm:main}
Conjunctions of~$\fd$s and $\uid$s are finitely controllable up to finite closure.
\end{theorem}

From these two results, and an efficient computation of the closure, we
deduce that the complexity of FQA matches that of UQA (see
Appendix~\ref{apx:prf_complexityfqa}):

\begin{restatable}{corollary}{complexityfqa}
  \label{cor:complexityfqa}
 FQA for $\fd$s and $\uid$s has PTIME data complexity and NP-complete
  combined complexity.
\end{restatable}

\subsection{Rephrasing with universal models}

We prove the main theorem via the notions of \defo{$k$-sound} and \defo{$k$-universal instances}.

\begin{definition}
  For $k \in \mathbb{N}$, we say that a superinstance $I$ of an instance~$I_0$ is
  \deft{$k$-sound} for constraints $\con$ (and for~$I_0$) if 
  for every constant-free $\cq$~$q$ of size $\leq k$ such that $I \models q$,
  we have $(I_0,
  \con) \entunr q$. We say it is \deft{$k$-universal} if the converse also
  holds: $I \models
  q$ whenever $(I_0, \con) \entunr q$.
\end{definition}

The assumption that $q$ is constant-free is without loss of
generality: we can always assume that, for each constant $c \in \dom(I_0)$, a
fact $P_c(c)$ has been added to $I_0$ for a fresh unary relation $P_c$, and
$c$ was replaced in $q$ by a existentially quantified variable~$x_c$ with the atom $P_c(x_c)$ added to~$q$. So for simplicity \emph{we
assume from now on that queries are constant-free}.

Theorem~\ref{thm:main} is implied by the following (see Appendix~\ref{apx:prf_main}):

\begin{restatable}[Universal models]{theorem}{univexists}
\label{thm:univexists}
For every conjunction $\con$ of~$\fd$s $\fds$ and $\uid$s $\ids$ closed under
finite implication,
for every finite instance $I_0$ that satisfies $\fds$,
for any $k \in \mathbb{N}$,
there exists a finite superinstance $I$ of~$I_0$
that is $k$-sound for $\con$ and satisfies~$\con$ (and hence is $k$-universal).
\end{restatable}

The fact that such an $I$ is $k$-universal is because 
any superinstance of~$I_0$ that
satisfies $\con$ must satisfy all $\cq$s~$q$ such that $(I_0, \con) \entunr q$,
by definition of~$\entunr$.

We now fix the conjunction $\con$ of $\fd$s~$\fds$ and $\uid$s~$\ids$.
We assume that $\con$ is closed under finite implication;
in particular,
$\fds$ and $\ids$ in isolation are closed under implication, which implies that
$\ids$ is transitively closed.
We also fix the instance~$I_0$ such that $I_0 \models \fds$, and
the maximal query size $k \in \NN$.

Our goal in the rest of this paper is to
construct the finite $k$-sound superinstance of~$I_0$ that satisfies $\con$, thus
proving the Universal Models Theorem and hence the Main Theorem.

\subsection{Restricting to $\acq$s, $\ufd$s, and reversible constraints}

We first prove the Universal Models Theorem for a restricted class of
queries and dependencies, which we now define. We will lift these restrictions later.

First, we define $\ufds$ to be the \emph{unary} $\fd$s of~$\fds$, and write $\ucon \defeq \ufds
\wedge \ids$. Note that, as we assumed that $\con$ is closed under finite implication for
$\ufd$s and $\uid$s, the characterization of~\cite{cosm} implies that~$\ucon$
also is.
We will first construct a $k$-sound superinstance that only satisfies~$\ucon$; in
Section~\ref{sec:hfds} we will show how to adapt the process to also satisfy
$\con$.

Second, we will first construct a superinstance that is \mbox{$k$-sound} only for
acyclic Boolean queries; in Section~\ref{sec:cycles} we will show how to make the resulting
superinstance sufficiently acyclic to be sound for cyclic queries as well.

Hence, in Sections~\ref{sec:wsound},~\ref{sec:ksound} and~\ref{sec:manyscc}, we prove the following weakening
of the Universal Models Theorem. The restrictions will be lifted in Sections~\ref{sec:hfds} and~\ref{sec:cycles}.

\begin{restatable}[Acyclic unary universal models]{theorem}{myuniv}
\label{thm:myuniv}
There exists a finite superinstance of~$I_0$ that satisfies~$\ucon$ and is $k$-sound
for~$\ucon$ and~$\acq$ (and hence $k$-universal for~$\ucon$ and~$\acq$).
\end{restatable}

To prove the Acyclic Unary Universal Models Theorem, in Sections~\ref{sec:wsound} and~\ref{sec:ksound}, we
will assume the following condition on the structure of the dependencies:

\medskip
\begin{samepage}
\begin{compactdesc}
\item[\assm{reversible}:] The following holds about $\ucon$:
  \begin{compactitem}
  \item all $\uid$s in $\ids$ are \deft{reversible} (remember this means that the reverse $\tau^{-1}$ of any
    $\tau \in \ids$ is also in $\ids$);
  \item for any positions $R^p$ and $R^q$ occurring in $\uid$s of $\ids$,
  if $R^p \rightarrow R^q$ is in~$\ufds$ 
  then so is $R^q \rightarrow R^p$.\end{compactitem}
\end{compactdesc}
\end{samepage}
\medskip

Intuitively, assumption \assm{reversible} is
connected to the finite closure characterization of~\cite{cosm}, which adds to
$\ucon$ the reverses of any $\uid$s and $\ufd$s that form a certain cyclic
pattern.

Working under assumption \assm{reversible},
Section~\ref{sec:wsound} proves an even weaker version of
the Acyclic Unary Universal Models Theorem, which replaces $k$-soundness by
weak-soundness; Section~\ref{sec:ksound} proves the actual theorem. Assumption
\assm{reversible} is lifted in Section~\ref{sec:manyscc} to conclude the proof.

\section{Weak-Soundness and Reversible $\uid$s}
\label{sec:wsound}
The goal of this section is to prove the Acyclic Unary Universal Models Theorem
(Theorem~\ref{thm:myuniv}) under assumption \assm{reversible}, replacing
$k$-soundness by \defp{weak-soundness}.

\begin{restatable}{definition}{wssinstance}
  \label{def:wssinstance}
  A superinstance $I'$ of an instance $I$ is
  \deft{weakly-sound} if the following holds:
  \begin{compactitem}
  \item for any $a \in \dom(I)$ and $R^p \in \pos(\sigma)$, if $a \in
    \pi_{R^p}(I')$, then either $a \in \pi_{R^p}(I)$ or $a \in
    \appelem{I}{R^p}$;
  \item for any $a \in \dom(I') \backslash \dom(I)$ and $R^p, S^q \in \pos(\sigma)$,
    if $\a \in \pi_{R^p}(I')$ and $a \in \pi_{S^q}(I')$ then $R^p = S^q$ or $\ui{R^p}{S^q}$
    is in~$\ids$.
  \end{compactitem}
\end{restatable}

Intuitively, a superinstance is weakly-sound if existing elements were only
added to positions where they wanted to appear, and new elements only occur at
positions which are connected in~$\ids$.
This section shows the following:

\begin{restatable}[Acyclic unary weakly-sound models]{proposition}{mywsound}
  \label{prp:auwsm}
  Under assumption \assm{reversible},
  there exists a finite superinstance of~$I_0$ that satisfies $\ucon$ and is
  weakly-sound.
\end{restatable}

The proposition itself will not be reused in the sequel, but the proof
introduces some useful concepts 
to prove the actual Acyclic Unary Universal Models Theorem in
Section~\ref{sec:ksound}.

\subsection{Binary signatures and balanced instances}

For simplicity, we first focus on a simplified case with a binary signature,
making the following assumption that will be lifted later in this section:

\smallskip
\begin{compactdesc}
\item[\assm{binary}:] all relations have arity~$2$ and $\ufds$ contains the $\ufd$s $R^1
    \rightarrow R^2$ and $R^2 \rightarrow R^1$ for any
    relation~$R$.
\end{compactdesc}
\smallskip

Our approach to construct a weakly-sound superinstance $I'$ of~$I_0$ that satisfies
$\ucon$ is then to perform a \emph{completion process}
that adds new (binary) facts  to connect together elements. As all
possible $\ufd$s hold, $I'$ can only
contain a new fact $R(a_1, b_2)$ if, for $i \in \{1, 2\}$,
$a_i \notin \pi_{R^i}(I_0)$, so that if $a_i \in \dom(I_0)$ then $a_i \in
\appelem{I_0}{R^i}$ by weak soundness.

One easy situation is when $I_0$ is \defo{balanced}: for
every relation~$R$, we can construct a bijection between the elements that want
to be in~$R^1$ and those that want to be in~$R^2$:

\begin{definition}
  \label{def:balanced}
  An instance $I$ is \deft{balanced} if, for every two
  positions $R^p$ and $R^q$ such that $R^p \rightarrow R^q$ and $R^q \rightarrow
  R^p$ are in~$\ufds$, we have $\card{\appelem{I}{R^p}} = \card{\appelem{I}{R^q}}$.
\end{definition}

If $I_0$ is balanced, we can show the Acyclic Unary Weakly-Sound
Models Proposition under assumption \assm{binary}, simply by pairing together
elements, without adding any new ones:

\begin{restatable}{proposition}{balwsnd}
  \label{prp:balwsnd}
  Assuming \assm{binary} and \assm{reversible},
  any balanced finite instance $I$ satisfying $\ufds$ has a finite weakly-sound
  superinstance $I'$ that satisfies $\ucon$, with $\dom(I') = \dom(I)$.
\end{restatable}

However, our instance $I_0$ may not be balanced. The idea is then to balance it by
adding ``helper'' elements and assigning them to positions, as the following
example shows:

\begin{example}
  \label{exa:helper}
  Consider three binary relations $R$, $S$, $T$, with the $\uid$s
  $\ui{R^2}{S^1}$, $\ui{S^2}{T^1}$, $\ui{T^2}{R^1}$
  and their reverses, and the
  $\fd$s prescribed by assumption \assm{binary}.
  Consider $I_0 \defeq \{R(a, b)\}$.
  We have $a \in \appelem{I_0}{T^2}$ and $b \in \appelem{I_0}{S^1}$;
  however $\appelem{I_0}{S^2} = \appelem{I_0}{T^1} = \emptyset$, so $I_0$ is not balanced.

  Still, we can construct the weakly-sound superinstance $I \defeq \{R(a, b), S(b, c),
  T(c, a)\}$ that satisfies the constraints. Intuitively, we have added a
  ``helper'' element $c$ and ``assigned'' it to the positions $S^1$
  and $T^2$, which are connected by the $\uid$s.
\end{example}

We now formalize this idea of constructing weakly-sound superinstances where the
domain is augmented with \defo{helper elements}. We first need to
understand at which positions the helpers can appear to avoid violating
weak-soundness:

\begin{definition}
  \label{def:eqids}
  For any two positions $R^p$ and $S^q$, we write $R^p \eqids S^q$ when
  $R^p = S^q$ or when
  $\ui{R^p}{S^q}$, and hence $\ui{S^q}{R^p}$ by assumption \assm{reversible},
  are in~$\ids$.
\end{definition}

As $\ids$ is transitively closed, $\eqids$ is an equivalence relation.
Our idea to construct weakly-sound superinstances is thus to first decide on the
helpers that we want to add, and the $\eqids$-class to which we want to assign
them, following the definition of weak-soundness.
We represent this choice as a \defo{partially-specified superinstance}, or
\defo{pssinstance}:

\begin{definition}
  \label{def:completion}
  A \deft{pssinstance}
  of an instance $I$ is a triple $P = (I, \calH,
  \mapb)$ where $\calH$ is a finite set of \deft{helpers} and $\mapb$ maps
  each $h \in \calH$ to an
  $\eqids$-class $\mapb(h)$.

  We define $\appelem{P}{R^p} \defeq \appelem{I}{R^p}
  \sqcup \{h \in \calH \mid R^p \in \mapb(h)\}$. This allows us
  to talk of~$P$ being \deft{balanced}
  following Definition~\ref{def:balanced}.

  A superinstance $I'$ of~$I$ is a \deft{realization} of~$P$ if $\dom(I') =
  \dom(I) \sqcup \calH$, and, for any fact $R(\mybf{a})$ of~$I' \backslash I$ and
  $R^p \in \pos(R)$, we have $a_p \in \appelem{P}{R^p}$.
\end{definition}

\begin{example}
  In Example~\ref{exa:helper}, a pssinstance of $I_0$ is $P \defeq (I_0, \{c\}, \mapb)$
  where $\mapb(c) \defeq \{S^1, T^2\}$, and $I$ is a realization of~$P$.
\end{example}

It is always possible to balance an instance by adding helpers:

\begin{restatable}[Balancing]{lemma}{hascompletion}
  \label{lem:hascompletion}
  For any finite instance~$I$, if $I$ satisfies~$\ufds$ then it has a
  balanced pssinstance.
\end{restatable}

From there, we can construct realizations like we constructed superinstances in Lemma~\ref{prp:balwsnd}.

\begin{restatable}[Binary realizations]{lemma}{balwsndcb}
  \label{lem:balwsndcb}
  For any balanced pssinstance $P$ of an instance $I$ that satisfies
  $\ufds$, we can construct a realization of
  $P$ that satisfies $\ucon$.
\end{restatable}

We then observe that realizations are weakly-sound superinstances of $I_0$.

\begin{restatable}[Binary realizations are completions]{lemma}{wsrwss}
  \label{lem:wsrwss}
  If $I'$ is a realization of a pssinstance of~$I$ then it is a
  weakly-sound superinstance of~$I$.
\end{restatable}

We have thus proved the Acyclic Unary Weakly-Sound
Models Proposition under assumptions \assm{binary} and \assm{reversible}, using
the completion process formed by combining the three above lemmas.

\subsection{Arbitrary arity and piecewise realizations}

We now lift assumption \assm{binary} (but retain assumption \assm{reversible}).
We show how to generalize the previous constructions to the arbitrary arity
case. Contrary to the \assm{binary} situation, we will see later that the
resulting completion process needs to assume that a certain \emph{saturation}
process has been applied to~$I_0$ beforehand.

The definition of balanced instances (Definition~\ref{def:balanced}) generalizes
to arbitrary arity, and we can show that the Balancing Lemma
(Lemma~\ref{lem:hascompletion}) still holds.
We keep the definition of pssinstance (Definition~\ref{def:completion})
but need to change the notion of realization. We replace it
by \defo{piecewise realizations}, which are defined on subsets of positions
that are connected in~$\ufds$.

\begin{definition}
  \label{def:eqfun}
  For any two positions $R^p$ and $R^q$, we write $R^p \eqfun R^q$ whenever $R^p
  \rightarrow R^q$ and $R^q \rightarrow R^p$ are in~$\ufds$.
\end{definition}

By transitivity of~$\ufds$, $\eqfun$ is clearly an equivalence relation.
We number the $\eqfun$-classes of
$\pos(\sigma)$ as $\Pi_1, \ldots, \Pi_{\neqfunc}$
and define \defp{piecewise instances} by their projections to the~$\Pi_i$:

\begin{definition}
  \label{def:piecewise}
  A \deft{piecewise instance}
  is an $\neqfunc$-tuple $\pire = (K_1, \ldots, K_{\neqfunc})$,
  where each $K_i$ is a set of $\card{\Pi_i}$-tuples, indexed by~$\Pi_i$ for
  convenience.
  The \deft{domain} of~$\pire$ is $\dom(\pire) \defeq \bigcup_i \dom(K_i)$.
  For $1 \leq i \leq \neqfunc$ and $R^p \in \Pi_i$, we write
  $\pi_{R^p}(\pire) \defeq \pi_{R^p}(K_i)$.
\end{definition}

We use this to define \defo{piecewise realizations} of pssinstances:

\begin{definition}
  \label{def:superbvpw}
  A piecewise instance $\pire = (K_1, \ldots, K_{\neqfunc})$ is a
  \deft{piecewise realization} of the pssinstance $P = (I, \calH, \mapb)$ if:
  \begin{compactitem}
    \item $\pi_{\Pi_i}(I) \subseteq K_i$ for all $1 \leq i \leq \neqfunc$,
    \item $\dom(\pire) = \dom(I) \sqcup \calH$, 
    \item for all $1 \leq i \leq \neqfunc$, for all $R^p
  \in \Pi_i$, for every tuple $\mybf{a} \in K_i \backslash \pi_{\Pi_i}(I)$, 
  we have $a_p \in \appelem{P}{R^p}$.
\end{compactitem}
\end{definition}

In order to generalize the Binary Realizations Lemma
(Lemma~\ref{lem:balwsndcb}),
we need to talk of a piecewise
instance $\pire$ ``satisfying'' $\ucon$. For $\ufds$, we require that $\pire$ respects the
$\ufd$s within each $\eqfun$-class. For
$\ids$, we define it directly from the projections of~$\pire$.

\begin{definition}
  A piecewise instance $\pire$ is \deft{$\ufds$-compliant} if, for all $1 \leq i \leq n$, there
  are no two tuples $\mybf{a} \neq \mybf{b}$ in~$K_i$
  such that $a_p = b_p$ for some
  $R^p \in \Pi_i$.

  $\pire$ is \deft{$\ids$-compliant} if $\appelem{\pire}{\tau} \defeq \pi_{R^p}(\pire) \backslash
  \pi_{S^q}(\pire)$ is empty for all $\tau \in \ids$.

  $\pire$ is \deft{$\ucon$-compliant} if it is $\ufds$- and $\ids$-compliant.
\end{definition}

We can then generalize the Binary Realizations Lemma:

\begin{restatable}[Realizations]{lemma}{balwsndc}
  \label{lem:balwsndc}
  For any balanced pssinstance~$P$ of an instance $I$ that satisfies
  $\ufds$, we can construct a $\ucon$-compliant piecewise
  realization of $P$.
\end{restatable}

\begin{example}
  Consider a $4$-ary relation $R$ and the $\uid$s $\tau: \ui{R^1}{R^2}$,
  $\tau': \ui{R^3}{R^4}$ and their reverses, and the $\ufd$s $\phi: R^1 \rightarrow
  R^2$, $\phi' : R^3 \rightarrow R^4$ and their reverses. We have
  $\Pi_1 = \{R^1, R^2\}$ and $\Pi_2 = \{R^3, R^4\}$.
  Consider $I_0 \defeq \{R(a,
  b, c, d)\}$, which is balanced, and the balanced pssinstance
  $P \defeq (I_0, \emptyset, \mapb)$,
  where $\mapb$ is the empty function. A $\ucon$-compliant piecewise realization of~$P$
  is $\pire \defeq (\{(a, b), (b, a)\}, \{(c, d), (d, c)\})$.
\end{example}

We now transform the $\ucon$-compliant piecewise
realization~$\pire$ into a weakly-sound superinstance, generalizing the ``Binary
Realizations Are Completions'' Lemma (Lemma~\ref{lem:wsrwss}),
and completing the description of our completion process.
The idea is to expand each tuple $\mybf{t}$ of each $K_i$ to an entire fact
$F_{\mybf{t}}$ of the
corresponding relation.

However, to fill the other positions of $F_{\mybf{t}}$, we will need to
reuse existing elements of~$I_0$.
For this, we want $I_0$ to contain some $R$-fact for
every relation $R$ that occurs in $\chase{I_0}{\ids}$.

\begin{restatable}{definition}{wsaturated}
  \label{def:wsaturated}
  A relation $R$ is \deft{achieved} (by $I$ and $\ids$) if there is some $R$-fact in
  $\chase{I}{\ids}$.
 
  A superinstance $I'$ of an instance $I$ is \deft{relation-saturated} (for~$\ids$)
  if every achieved relation (by $I$ and $\ids$) occurs in~$I'$.
\end{restatable}

\begin{example}
  Consider two binary relations $R$ and $T$ and a unary relation $S$, the
  $\uid$s $\tau: \ui{S^1}{R^1}$, $\tau': \ui{R^2}{T^1}$ and their reverses, no $\ufd$s, and the
  non-relation-saturated instance $I_0 \defeq \{S(a)\}$ which is trivially balanced.

  $P \defeq (I_0, \emptyset, \mapb)$, with $\mapb$ the empty function, is a
  pssinstance of~$I$, and
  $\pire \defeq (\{(a)\}, \emptyset, \{(a)\}, \emptyset, \emptyset)$,
  where $\Pi_1$ and $\Pi_3$ are
  the $\eqfun$-classes of $R^1$ and $S^1$, is a $\ucon$-compliant piecewise realization of~$P$.
  However, we cannot easily complete $\pire$ to a superinstance of $I_0$
  satisfying~$\tau$ and~$\tau'$, because,
  to create the fact $R(a, \bullet)$, we need to create an element to fill
  position $R^2$, and this would introduce a violation of~$\tau'$. Intuitively,
  this is because $I_0$ is not relation-saturated.

  Consider instead the instance $I_1 \defeq I_0 \sqcup \{S(c), R(c, d), T(d)\}$.
  We can complete~$I_1$ to satisfy $\tau$ and $\tau'$ by adding the
  fact $R(a, d)$, reusing the element $d$ to fill position $R^2$.
\end{example}

Clearly, initial chasing on~$I_0$ ensures relation-saturation:

\begin{restatable}[Relation-saturated solutions]{lemma}{wtsol}
  \label{lem:wtsol}
  The result of performing sufficiently many chase rounds on any instance~$I$ is relation-saturated.
\end{restatable}

Relation-saturation ensures that we can reuse existing elements when
completing $\pire$. This allows us to perform the last step of the completion
process:

\begin{restatable}[Using realizations to get completions]{lemma}{sreal}
  \label{lem:sreal}
  For any finite relation-saturated instance $I$ that satisfies $\ufds$,
  from a $\ucon$-compliant piecewise realization $\pire$ of a pssinstance of
  $I$, we can construct a finite weakly-sound
  superinstance of~$I$ that satisfies $\ucon$.
\end{restatable}

We can now prove the Acyclic Unary Weakly-Sound Models Proposition.
Consider our initial finite instance~$I_0$, that satisfies~$\ufds$, and chase it
to a finite relation-saturated superinstance $I_0'$ using
the Relation-Saturated Solutions Lemma. By the Unique Witness Property, $I_0'$
still satisfies $\ufds$, and it is clearly a weakly-sound superinstance of
$I_0$.

Now, perform the completion process: construct a balanced pssinstance~$P$
of~$I_0'$ using the Balancing Lemma (Lemma~\ref{lem:hascompletion}),
and a finite $\ucon$-compliant piecewise realization $\pire$ of~$P$
by the Realizations Lemma (Lemma~\ref{lem:balwsndc}).
Then, use the realization~$\pire$ 
with Lemma~\ref{lem:sreal}
to construct the finite weakly-sound superinstance~$I$ of~$I_0'$ that satisfies
$\ucon$. $I$ is clearly also a weakly-sound superinstance
of~$I_0$, so the result is proven.

\section{$k$-Soundness and Reversible $\uid$s}
\label{sec:ksound}
We now move from weak-soundness to~$k$-soundness, to prove the Acyclic Unary
Universal Models Theorem (Theorem~\ref{thm:myuniv}), still making assumption \assm{reversible}.

We first introduce the notion of \emph{aligned superinstances} that we use
to maintain $k$-soundness, and give the saturation process that
generalizes relation-saturation.
We then define a notion of \emph{thrifty chase steps},
and a completion process that
uses these chase steps to repair $\uid$ violations in the instance.

\subsection{Aligned superinstances and fact-saturation}

We ensure $k$-soundness by maintaining a
\emph{$k$-bounded simulation} from our superinstance of $I_0$ to the chase $\chase{I_0}{\ids}$.
Indeed, $\chase{I_0}{\ids}$ is a universal model for $\ids$, and it satisfies
$\fds$ (by the Unique Witness Property, and because $I_0$ does). Hence, it is in
particular $k$-sound for~$\con$. Now, as acyclic queries of size $\leq k$ are
preserved through $k$-bounded simulations, superinstances of $I_0$ with a
$k$-bounded simulation to $\chase{I_0}{\ids}$ are indeed $k$-sound for $\acq$.

\begin{definition}
  \label{def:bbsim}
  For $I$, $I'$ two instances, $a \in \dom(I)$, $b \in \dom(I')$, and $n \in
  \NN$,
  we write $(I, a) \bsim_n (I', b)$ if,
  for any fact $R(\mybf{a})$ of~$I$ with $a_p = a$ for some~$R^p \in \pos(R)$,
  there exists a fact $R(\mybf{b})$ of~$I'$ such that $b_p = b$,
  and $(I, a_q) \bsim_{n-1} (I', b_q)$ for all~$R^q \in \pos(R)$.
  The base case $(I, a) \bsim_0 (I', b)$ always holds.

  An \deft{$n$-bounded simulation} from $I$ to $I'$ is a mapping $\cov$ 
  such that for all $a \in \dom(I)$, $(I, a) \bsim_n (I', \cov(a))$.

  We write $a \bbsim_n b$ for $a, b \in \dom(I)$ if both $(I, a) \bsim_n (I, b)$ and
  $(I, b) \bsim_n (I, a)$; this is an equivalence relation on $\dom(I)$.

\end{definition}

\begin{restatable}{lemma}{ksimacq}
  \label{lem:ksimacq}
For any instance $I$ and $\acq$ $q$ of size $\leq n$ such that $I \models
q$, if there is an $n$-bounded simulation from $I$ to $I'$, then $I' \models q$.
\end{restatable}

We accordingly give a name to superinstances of $I_0$ that have a $k$-bounded
simulation to the chase. For convenience, we also require them to be finite and satisfy $\ufds$.
For technical reasons we require that the simulation is the identity on $I_0$,
that it does not map other elements to $I_0$, and that elements occur in the
superinstance
at least at the position where their $\cov$-image was introduced in the chase:

\begin{definition}
  An \deft{aligned superinstance} $J = (I, \cov)$ of~$I_0$ is a finite superinstance
  $I$ of~$I_0$
  that satisfies $\ufds$,
  and a $k$-bounded simulation
  $\cov$ from $I$ to $\chase{I_0}{\ids}$ such that $\restr{\cov}{I_0}$ is the
  identity and $\restr{\cov}{(I\backslash I_0)}$ maps to $\chase{I_0}{\ids}
  \backslash I_0$.
  
  Further, for any $a \in \dom(I) \backslash \dom(I_0)$,
  letting $R^p$ be the position where $\cov(a)$ was introduced in
  $\chase{I_0}{\ids}$, we require that $a \in \pi_{R^p}(I)$.
\end{definition}

Before we perform the \emph{completion process} that allows us to satisfy
$\ids$, we need to perform a \emph{saturation process}, like relation-saturation
in the previous section. Instead of achieving all relations, we want the
aligned superinstance to achieve all \defo{fact classes}:

\begin{definition}
  \label{def:factclass}
  A \deft{fact class} is a pair $(R^p, \mybf{C})$ of a position $R^p
  \in \pos(\sigma)$
  and a $\arity{R}$-tuple of 
  $\bbsim_k$-classes of elements of
  $\chase{I_0}{\ids}$. The dependency on~$k$ is omitted for brevity.

  The \deft{fact class} of a fact $F = R(\mybf{a})$ of~$\chase{I_0}{\ids}
  \backslash I_0$ is $(R^p, \mybf{C})$, where $a_p$ is the exported element
  of~$F$ and $C_i$ is the $\bbsim_k$-class of~$a_i$ in $\chase{I_0}{\ids}$ for
  all $R^i \in \pos(R)$.
  
  A fact class $(R^p, \mybf{C})$ is
  \deft{achieved} if it is the fact class of some fact of
  $\chase{I_0}{\ids} \backslash I_0$.
  We write $\rfcl$ for the set of all achieved 
  fact classes
  (for brevity, the dependence on~$I_0$, $\ids$, and $k$ is omitted from notation).

  An aligned superinstance $J = (I, \cov)$ is \deft{fact-saturated} if, for any achieved fact
  class $D = (R^p, \mybf{C})$ in $\rfcl$, there is a fact $F_D =
  R(\mybf{a})$ of~$I \backslash I_0$
  such that $\cov(a_i) \in C_i$ for all $R^i \in \pos(R)$.
  We say that $F_D$ \deft{achieves}~$D$ in~$J$.
\end{definition}

\begin{restatable}{lemma}{clasfino}
  \label{lem:clasfino}
  For any initial instance $I_0$, set $\ids$ of~$\uid$s, and $k \in \mathbb{N}$, $\rfcl$ is finite.
\end{restatable}

We now define our saturation process: chase $I_0$ until all fact classes are
achieved, which is possible in finitely many rounds thanks to the above lemma. The
result is easily seen to be a fact-saturated aligned superinstance:

\begin{restatable}[Fact-saturated solutions]{lemma}{nondetexhaust}
  \label{lem:nondetexhaust}
  The result $I$ of performing sufficiently many chase
  rounds on $I_0$ is such that $J_0 = (I, \id)$ is a fact-saturated aligned
  superinstance of~$I_0$.
\end{restatable}

We thus obtain a fact-saturated aligned superinstance $J_0$ of~$I_0$,
which we now want to complete to one that satisfies $\ids$.

\subsection{Fact-thrifty completion}

Our general method to repair $\uid$ violations in $J_0$ is to apply a form of
chase step on aligned superinstances, which may reuse elements:
\defo{thrifty chase steps}. To define them, we first distinguish \defo{dangerous}
and \defo{non-dangerous} positions, which determine how we may reuse elements
when chasing.

\begin{restatable}{definition}{dangerdef}
  \label{def:dangerdef}
  We say a position $S^r \in \pos(\sigma)$ is \deft{dangerous} for a position
  $S^q \neq S^r$ if $S^r \rightarrow S^q$ is in~$\ufds$, and write $S^r \in \danger(S^q)$.
  Otherwise, $S^r$ is
  \deft{non-dangerous}, written $S^r \in \nondanger(S^q)$. Note that $\{S^q\}
  \sqcup \danger(S^q) \sqcup \nondanger(S^q) = \pos(S)$.
\end{restatable}

\begin{definition}[Thrifty chase steps]
  \label{def:thrifty}
  Let $J = (I, \cov)$ be an aligned superinstance of~$I_0$, let $\tau : \ui{R^p}{S^q}$
  be a $\uid$ of~$\ids$, and let $F_\a = R(\mybf{a})$ be an active fact for
  $\tau$ in $I$.
  
  Because $\cov$ is a $1$-bounded simulation, $\cov(a_p) \in
  \pi_{R^p}(\chase{I_0}{\ids})$, so, because the chase satisfies~$\tau$, there
  is a fact $F_\w = S(\mybf{b}')$ in $\chase{I_0}{\ids}$
  with $b'_q = \cov(a_p)$; we call $F_\w$ the \deft{chase witness}.

  Applying a \deft{thrifty chase step} on $F_\a$ for~$\tau$ yields an aligned
  superinstance $J' = (I', \cov')$. We define $I'$ as $I$ plus a new fact $F_\n =
  S(\mybf{b})$, where $b_q = a_p$ and the $b_r$ for $S^r \neq S^q$ may be elements of $\dom(J)$ or fresh elements.
  We require that:
  \begin{compactitem}
  \item for $S^r \in \nondanger(S^q)$, $b_r \in \pi_{S^r}(J)$ (so they are not
    fresh)
  \item for $S^r \in \danger(S^q)$, $b_r \notin \pi_{S^r}(J)$ (so they may be
    fresh)
  \item for $S^r \neq S^q$, if $b_r$ is not fresh then $\cov(b_r) \bbsim_k
    b'_r$.
  \end{compactitem}

  We define $\cov'$ by extending $\cov$ to $\dom(J')$: we set $\cov'(b_r)
  \defeq b_r'$ whenever $b_r$ is fresh.

  A \deft{fact-thrifty chase step} is a thrifty chase step where we choose one fact
  $F_\r = S(\mybf{c})$ of $J \backslash I_0$ that achieves the fact class of
  $F_\w$ (that is, $\cov(c_i) \bbsim_k b'_i$ for all
  $i$), and use $F_\r$ to define $b_r \defeq c_r$ for all $S^r \in \nondanger(S^q)$.

  The chase step is \deft{fresh} if $b_r$ is fresh for all $S^r \in \danger(S^q)$.
\end{definition}

Thrifty chase steps may in general violate $\ufds$, but fact-thrifty chase steps
never do. For this reason, we will only use fact-thrifty chase steps in this
section. The point of working with fact-saturated aligned superinstances is that we can ensure
that a suitable $F_\r$ always exists. We thus claim:

\begin{restatable}[Fact-thrifty chase steps]{lemma}{ftok}
  \label{lem:ftok}
  For any fact-saturated aligned superinstance $J$, the result $J'$ 
  of a fact-thrifty chase step on $J$ is indeed a well-defined aligned
  superinstance where the former active fact $F_\a$ is no longer active.
\end{restatable}

We now claim that we can expand fact-saturated superinstances to satisfy $\ids$,
using fact-thrifty chase steps:

\begin{restatable}[Fact-thrifty completion]{proposition}{ftcomp}
  \label{prp:ftcomp}
  Under assumption \assm{reversible},
  for any fact-saturated aligned
  superinstance $J$ of~$I_0$,  we can expand $J$ by fact-thrifty chase steps to
  a fact-saturated aligned superinstance $J'$ of~$I_0$ that satisfies $\ids$.
\end{restatable}

This proposition allows us to prove
the Acyclic Unary Universal Models Theorem (Theorem~\ref{thm:myuniv}) under assumption
\assm{reversible}. Indeed, consider the
fact-saturated aligned superinstance~$J_0$ produced by the
Fact-Saturated Solutions Lemma (Lemma~\ref{lem:nondetexhaust}). 
Applying the Fact-Thrifty Completion Proposition to~$J_0$
yields a fact-saturated aligned superinstance~$J'$, which is a finite $k$-sound
superinstance of~$I_0$ that satisfies $\ufds$ and satisfies $\ids$.

The rest of this
section sketches the proof of the Proposition (see Appendix~\ref{apx:prf_ftcomp}
for the full proof).
The idea is to construct, as in Section~\ref{sec:wsound}, a balanced pssinstance $P$ of
the input aligned superinstance~$J$, and a $\ucon$-compliant piecewise realization $\pire$ of~$P$.
Now, instead of completing the facts of $\pire$ to add them directly to~$J$, we
add them one by one, using fact-thrifty chase steps, to ensure that alignedness
is preserved.

The only problematic point is that $\pire$ could connect together elements that
have dissimilar $\cov$-images, violating alignedness.
However, we show that, up to chasing for $k+1$ rounds on the initial $J$ with
fresh fact-thrifty chase steps before constructing~$P$,
we can ensure what we call \emph{$k$-reversibility}: all
elements that want to be at some position $R^p$ in~$J$ have a $\cov$-image whose
$\bbsim_k$-class only depends on~$R^p$. Once we have ensured this, we can
essentially stop worrying about $\cov$-images, because respecting
weak-soundness, as $\pire$ does, is sufficient.

The reason why $k+1$ chasing rounds suffice to ensure this
is by a general structural
observation on the $\uid$ chase:
when the last $k$ $\uid$s applied to an element $a$ of $\chase{I_0}{\ids}$ are
reversible (as is the case here, by assumption \assm{reversible}),
the $\bbsim_k$-class of $a$ only depends on the $\eqids$-class of the
position where it was introduced, and not on its exact history. Formally:

\begin{restatable}[Chase locality theorem]{theorem}{locality}
  \label{thm:locality}
  For any instance~$I_0$, transitively closed set of $\uid$s $\ids$,
  and $n \in \mathbb{N}$,
  for any two
  elements $a$ and $b$
  respectively introduced at positions~$R^p$ and~$S^q$ in $\chase{I_0}{\ids}$
  such that $R^p \eqids S^q$,
  if the last $n$ $\uid$s applied to create $a$ and~$b$ are reversible,
  then $a \bbsim_n b$.
\end{restatable}

\section{Arbitrary $\uid$s: Lifting Assumption \assm{reversible}}
\label{sec:manyscc}
This section concludes
the proof of the Acyclic Unary Universal Models Theorem (Theorem~\ref{thm:myuniv})
by removing
assumption \assm{reversible}. We do so by splitting $\ids$ in subsets that can
be satisfied sequentially:

\begin{definition}
  \label{def:opartition}
  For any $\tau, \tau' \in \ids$,
  we write $\tau \idprec \tau'$ when we can write $\tau = \ui{R^p}{S^q}$ and $\tau' =
  \ui{S^r}{T^u}$ with $S^q \neq S^r$, and the $\ufd$ $S^r \rightarrow S^q$ is in
  $\ufds$.
  An \deft{ordered partition} $(P_1, \ldots, P_{\neqidsc})$ of~$\ids$ is a
  partition of $\ids$ (i.e., $\ids = \bigsqcup_i P_i$)
  such that for any $\tau \in P_i$, $\tau' \in P_j$, if
  $\tau \idprec \tau'$ then $i \leq j$.
\end{definition}

The notion of ordered partition is useful because thrifty chase steps can only
cause new $\uid$ violations at the dangerous positions of the new fact.
This implies the following:

\begin{restatable}{lemma}{sccbyscc}
  \label{lem:sccbyscc}
  Let $J$ be an aligned superinstance of~$I_0$ and $J'$ be the result of
  applying a thrifty chase step on $J$ for a $\uid$ $\tau$ of~$\ids$.
  Assume that a $\uid$ $\tau'$ of~$\ids$ was satisfied by $J$
  but is not satisfied by $J'$. Then $\tau \idprec \tau'$.
\end{restatable}

Hence, given an ordered partition of~$\ids$, once we have satisfied the $\uid$s
of the first $i$ classes $P_1, \ldots, P_i$, then this property is preserved
while we do thrifty chasing with $P_j$, $j > i$. So if we can satisfy each $P_i$
individually with thrifty chase steps, then we can satisfy $\ids$ by satisfying
$P_1, \ldots, P_{\neqidsc}$.

Of course, the point of partitioning $\ids$ is to be able to control the
structure of the $\uid$s in each class:

\begin{definition}
  \label{def:manageable}
  We call $P \subseteq \ids$ \deft{reversible} if it is transitively
  closed
  (as $\ids$ is) and satisfies assumption \assm{reversible}.
  
  We say $P
  \subseteq \ids$ is \deft{trivial} if we have $P = \{\tau\}$ for some $\tau
  \in \ids$ such that $\tau \not\idprec \tau$. An ordered partition is
  \deft{manageable} if all of its classes are either reversible or trivial.
\end{definition}

If $P \subseteq \ids$ is reversible,
then the previous section describes
how to complete with thrifty chase steps
any fact-saturated aligned superinstance
of $I_0$ to one that
satisfies~$P$. If $P$ is trivial, it follows directly from
Lemma~\ref{lem:sccbyscc} that we can satisfy it:

\begin{restatable}{corollary}{trivscc}
  \label{cor:trivscc}
  For any trivial class $\{\tau\}$, 
  performing one chase round on an aligned fact-saturated superinstance $J$ of~$I_0$ by
  fresh fact-thrifty chase steps
  for $\tau$ yields an aligned superinstance $J'$ of~$I_0$ that satisfies~$\tau$.
\end{restatable}

We now claim that we can construct a manageable partition of~$\ids$. We build it
as a topological sort of the strongly connected components
(SCCs) of the directed graph on $\ids$ defined by~$\idprec$, with the technical
complication that SCCs must be closed under $\uid$ reversal.
The construction relies on the fact
that $\ids$ is closed under finite implication, as characterized by Cosmadakis
et al.~\cite{cosm}.

\begin{restatable}{lemma}{nontrivconv}
  \label{lem:nontrivconv}
  Any conjunction $\ids$ of~$\uid$s closed under finite implication has a
  manageable partition.
\end{restatable}

\begin{example}
  Consider the $\uid$s $\tau_R: \ui{R^1}{R^2}$,
  $\tau_S: \ui{S^1}{S^2}$, $\tau: \ui{R^3}{S^3}$, and the $\ufd$s
  $\phi_R: R^1 \rightarrow R^2$,
  $\phi_S: S^1 \rightarrow S^2$,
  $\phi_R': R^3 \rightarrow R^1$, and
  $\phi_S' : S^1 \rightarrow S^3$.
  The $\uid$s $\tau_R^{-1}$ and $\tau_S^{-1}$, and $\ufd$s
  $\phi_R^{-1}$, $\phi_S^{-1}$, and $R^3 \rightarrow R^2$, $S^2 \rightarrow
  S^3$, are finitely implied. A
  manageable partition is $(\{\tau_R, \tau_R^{-1}\}, \{\tau\}, \{\tau_S,
  \tau_S^{-1}\})$, where the first and third classes are reversible and the
  second is trivial.
\end{example}

We can now conclude the proof of the Acyclic Unary Universal Models Theorem
(Theorem~\ref{thm:myuniv}).
We first note that the Fact-Saturated Solutions Lemma
(Lemma~\ref{lem:nondetexhaust})
does not use assumption
\assm{reversible}, so we apply it (with $\ids$) to obtain from $I_0$ an
aligned fact-saturated superinstance $J_1$ of~$I_0$.
This is the \deft{saturation process}.

We now satisfy $\ids$ by a \deft{completion process}.
Build a manageable partition $(P_1, \ldots, P_{\neqidsc})$ of~$\ids$,
by Lemma~\ref{lem:nontrivconv}.
Now, for $1 \leq i \leq \neqidsc$, use fact-thrifty chase steps by $\uid$s of
$P_i$ to
extend the
fact-saturated aligned superinstance $J_i$ to a larger one~$J_{i+1}$
that satisfies~$P_i$.
If $P_i$ is trivial, use Corollary~\ref{cor:trivscc}. If $P_i$ is
reversible, apply the Fact-Thrifty Completion Proposition
(Proposition~\ref{prp:ftcomp}), taking $\ids$ to be $P_i$.
By Lemma~\ref{lem:sccbyscc}, the result $J_{i+1}$ satisfies $\bigcup_{j
\leq i} P_j$.

Hence the result $J_{n+1}$ of the completion process is an aligned superinstance of
$I_0$ that satisfies $\ids$; as an aligned superinstance, it
is also finite, satisfies $\ufds$, and is $k$-sound for $\acq$; so it is
$k$-universal for $\ucon$ and $\acq$. This concludes the proof of the Acyclic
Unary Universal Models Theorem.

\section{Higher-Arity $\fd$s}
\label{sec:hfds}
We now bootstrap the Acyclic Unary Universal Models Theorem
(Theorem~\ref{thm:myuniv}) to the Universal Models Theorem
(Theorem~\ref{thm:univexists}).
The first step is to change our construction to avoid violating higher-arity
$\fd$s, namely, show the following, which applies to
$\con = \ids \wedge \fds$ rather than $\ucon = \ids \wedge \ufds$:

\begin{restatable}[Acyclic universal models]{theorem}{myuniv2}
\label{thm:myuniv2}
There is a finite superinstance of~$I_0$ that is $k$-universal
for~$\con$
and $\acq$ queries.
\end{restatable}

The problem to address is that our completion process to satisfy $\ids$ was defined with fact-thrifty
chase steps, which reuse elements from the same facts at the same positions
multiple times. This may violate $\fds$, and we can show that is the only point
where we do so in the construction.

The goal of this section is to define a new version of thrifty chase steps that
preserves $\fds$ rather than just $\ufds$; we call them \defo{envelope-thrifty
chase steps}. We first describe the new saturation process designed for
them. Second, we define how they work, redefine the completion
process of the previous section to use them, and use this new completion process
to prove the Acyclic Universal Models Theorem above.

\subsection{Envelopes and saturation}

We start by defining a new notion of saturated instances.
Recall the notions of fact classes (Definition~\ref{def:factclass}) and thrifty 
chase steps (Definition~\ref{def:thrifty}).
When a thrifty chase step wants to create a fact $F_\n$ whose chase witness
$F_\w$ has
fact class $(R^p, \mybf{C})$, it needs elements to reuse in $F_\n$ at positions of
$\nondanger(R^p)$. They must have the right $\cov$-image and
must already occur at the positions where they are reused.

Fact-thrifty chase steps reuse a tuple of elements from one fact $F_\r$,
and thus apply to \emph{fact-saturated instances} with one fact for each class. Our new notion
of envelope-thrifty chase steps will need saturated instances that have
\emph{multiple}
reusable tuples. A set of such tuples is called an \defo{envelope} for $(R^p,
\mybf{C})$:

\begin{restatable}{definition}{envelope}
  \label{def:envelope}
  Consider $D = (R^p, \mybf{C})$ in $\rdfcl$, and write $O \defeq \nondanger(R^p)$.
  An \deft{envelope}~$E$ for~$D$ and for an aligned superinstance $J = (I, \cov)$ of~$I_0$
  is a non-empty set of $\card{O}$-tuples indexed by $O$, with domain
  $\dom(I)$, such that:
  \begin{compactitem}
  \item for every $\fd$ $\phi : R^L \rightarrow R^r$ of $\fds$ with $R^L \subseteq O$
    and $R^r \in O$, $E$ satisfies $\phi$ (seeing its tuples as facts on~$O$);
  \item for every $\fd$ $\phi : R^L \rightarrow R^r$ of $\fds$ with $R^L
    \subseteq O$ and $R^r \notin O$, for all $\mybf{t}, \mybf{t}' \in E$,
    $\pi_{R^L}(\mybf{t}) = \pi_{R^L}(\mybf{t}')$ implies $\mybf{t} = \mybf{t}'$;
    \item for every $a \in \dom(E)$, there is exactly one position $R^q \in
      O$ such
      that $a \in \pi_{R^q}(E)$; and then we also have $a \in \pi_{R^q}(J)$;
    \item for any fact $F = R(\mybf{a})$ of~$J$ and $R^q \in
      O$, if
      $a_q \in \pi_{R^q}(E)$, then $F$ achieves $D$ in $J$
      and $\pi_{O}(\mybf{a}) \in E$.
  \end{compactitem}
\end{restatable}

Intuitively, the tuples in the envelope~$E$ satisfy
the $\ufd$s of~$\ufds$ within~$\nondanger(R^p)$, and never overlap on positions
that determine a position out of $\nondanger(R^p)$. Further, their elements
already occur at the positions where they will be reused, and have the
right $\cov$-image for the fact class~$D$. To simplify the reasoning, we
also impose that each
element of~$E$ is used at only one position, and occurs at that position only in facts
which achieve~$D$ and whose projection to $\nondanger(R^p)$ is in~$E$.

Depending on~$O$, it may be possible to use a singleton tuple as the envelope,
like fact-thrifty chase steps, and not violate $\fds$. The class is then \defo{safe}.
Otherwise, we focus on the envelope tuples which do not appear in the
instance yet.

\begin{definition}  
  We call $(R^p, \mybf{C})$ in $\rdfcl$ \deft{safe} 
  if there is no $\fd$ $R^L
  \rightarrow R^r$ in~$\fds$ with $R^L \subseteq \nondanger(R^p)$ and $R^r
  \notin \nondanger(R^p)$.

  Letting $E$ be an envelope for~$(R^p, \mybf{C})$ and $J$ be an aligned
  superinstance, the \deft{remaining tuples} of~$E$ are $E \backslash \pi_{\nondanger(R^p)}(J)$ if
  $(R^p, \mybf{C})$ is unsafe, and $E$ if it is safe.
\end{definition}

We now introduce the notion of \defo{global envelopes}, that give us one
envelope per class of~$\rdfcl$. This leads to our new notion of saturation: a
saturated instance has a global envelope with many remaining tuples in the
unsafe classes. Note that this implies fact-saturation.

\begin{definition}
  A \deft{global envelope} $\calE$ for an aligned superinstance $J = (I, \cov)$
  of~$I_0$ is a mapping from each $D \in \rdfcl$
  to an envelope $\calE(D)$ for~$D$ and~$J$, such that the envelopes have
  pairwise disjoint domains.

  We call $J$ \deft{$n$-envelope-saturated} if it has a global envelope~$\calE$
  such that $\calE(D)$ has $\geq n$ remaining tuples for all unsafe $D \in
  \rdfcl$. $J$ is
  \deft{envelope-saturated} if it is $n$-envelope-saturated for $n > 0$, and
  \deft{envelope-exhausted} otherwise.
\end{definition}

We now justify that we can make arbitrarily saturated superinstances of~$I_0$
(the switch to~$I_0'$ is a technicality):

\begin{restatable}[Sufficiently envelope-saturated solutions]{proposition}{preproc}
  \label{prp:preproc}
  For any $K \in \mathbb{N}$ and instance~$I_0$,
  we can build a superinstance $I_0'$ of~$I_0$ that is $k$-sound for $\cq$,
  and an aligned superinstance $J$ of
  $I_0'$ that satisfies~$\fds$ and is $(K\card{J})$-envelope-saturated.
\end{restatable}

\begin{example}
  For simplicity, we work with instances rather than aligned superinstances.
  Consider $I_0 \defeq \{S(a), T(z)\}$, the $\uid$s $\tau: \ui{S^1}{R^1}$ and $\tau':
  \ui{T^1}{R^1}$ for a $3$-ary relation $R$,
  and the $\fd$ $\phi: R^2 R^3 \rightarrow R^1$.
  Consider $I \defeq I_0 \sqcup \{R(a, b, c)\}$ obtained by one chase step of~$\tau$
  on $S(a)$. 
  It would violate $\phi$ to perform a fact-thrifty chase step
  of $\tau'$ on $z$ to create $R(z, b, c)$, reusing $(b, c)$ at $\nondanger(R^1)
  = \{R^2, R^3\}$.

  Now, consider the $k$-sound $I_0' \defeq \{S(a), T(z), S(a'), S(z')\}$,
  and $I' \defeq I_0' \sqcup
  \{R(a, b, c), \allowbreak R(a', b', c')\}$ obtained by two chase steps.
  The two facts $R(a, b, c)$ and $R(a', b', c')$ would be mapped 
  to the same fact class~$D$, so we can define
  $E(D) \defeq \{(b, c), (b', c'), (b', c), (b, c')\}$. We can now satisfy $\ids$
  on~$I'$ without violating~$\phi$,
  with two envelope-thrifty chase steps that reuse the remaining tuples $(b', c)$ and
  $(b, c')$ of~$E(D)$.
\end{example}

The crucial result needed for the Sufficiently Envelope-Saturated Proposition is
the following, which may be of independent interest, and is proved in
Appendix~\ref{apx:prf_combinatorial} using a combinatorial construction. The
fact that unary keys are problematic is the reason why we handle safe classes
differently.

\begin{restatable}[Dense interpretations]{theorem}{combinatorial}
  \label{thm:combinatorial}
  For any set $\fds$ of~$\fd$s over a relation $R$ with no unary key, and $K \in
  \mathbb{N}$, there exists a non-empty instance $I$ of~$R$ that satisfies
  $\fds$ and has at least $K \card{\dom(I)}$ facts.
\end{restatable}

Hence, we have defined the new notion of $n$-envelope-saturation, and a
saturation process to achieve it: the Sufficiently Envelope-Saturated Solutions
Proposition. Unlike the Fact-Saturated Solutions Lemma, where one
fact of each class was enough, we have shown that envelope-saturated superinstances
may have an arbitrarily high saturation relative to the instance size.

\subsection{Envelope-thrifty chase steps}

We can now introduce \defo{envelope-thrifty chase steps}:

\begin{definition}
  \deft{Envelope-thrifty chase steps} are thrifty chase steps
  (Definition~\ref{def:thrifty}) applicable to
  envelope-saturated aligned superinstances.
  Let $S^q$ be the exported position of the new fact $F_\n$,
  let $F_\w = S(\mybf{b}')$ be the chase witness,
  and let $D = (S^q, \mybf{C}) \in \rdfcl$ be the fact class of~$F_\w$.
  We choose some remaining tuple $\mybf{t}$ of $\calE(D)$ and
  define $b_r \defeq t_r$ for all $S^r \in \nondanger(S^q)$.
\end{definition}

Recall from 
Lemma~\ref{lem:ftok}
that fact-thrifty chase steps apply to
fact-saturated aligned superinstances,
and never violate~$\ufds$. Similarly, envelope-thrifty chase
steps apply to envelope-saturated aligned superinstances,
and never violate $\fds$:

\begin{restatable}{lemma}{envfds}
  \label{lem:envfds}
  For $n > 0$,
  for any $n$-envelope-saturated aligned superinstance $J$ that satisfies
  $\fds$,
  the result $J'$ of an envelope-thrifty chase step on~$J$
  is an $(n-1)$-envelope-saturated superinstance that satisfies~$\fds$.
\end{restatable}

We now modify the Fact-Thrifty Completion Proposition
(Proposition~\ref{prp:ftcomp}),
generalized without assumption \assm{reversible} as in the previous section, 
to use envelope-thrifty
chase steps instead of fact-thrifty chase steps. This is possible because the
choice of reused elements at non-dangerous positions makes no difference in
terms of applicable $\uid$s, as they already occur at the position where they
are reused. Hence, we can
perform the exact same process as before (except the non-dangerous reuses), using Lemma~\ref{lem:envfds} to justify that
$\fds$ is preserved;
but we must abort if we reach an
envelope-exhausted instance:

\begin{restatable}[Envelope-thrifty completion]{proposition}{etcomp}
  \label{prp:etcomp}
  For any envelope-saturated aligned superinstance $J$ of~$I_0$ that satisfies~$\fds$, we
  can obtain by envelope-thrifty chase steps an aligned superinstance $J'$ of
  $I_0$, such that $J'$ is either envelope-exhausted or satisfies $\con$.
\end{restatable}

The last problem to address is exhaustion.
Unlike fact-saturation, envelope-saturation ``runs out'';
whenever we use a remaining tuple $\mybf{t}$ in a chase step to create $F_\n$ and
obtain a new aligned superinstance~$J'$, then we cannot use $\mybf{t}$ again in~$J'$.
So we must start with a sufficiently envelope-saturated superinstance,
and we must
control how many chase steps are applied in the envelope-thrifty completion
process. From the details of our construction, we can show the following:

\begin{restatable}[Envelope blowup]{lemma}{eblowup}
  \label{lem:eblowup}
  There exists $B \in \mathbb{N}$ depending only on $k$ and $\ucon$ such that,
  for any aligned superinstance $J = (I, \cov)$ of~$I_0$, and global envelope $\calE$,
  letting $J' = (I', \cov')$ be the result of the
  envelope-thrifty completion process, we have $\card{I'} < B \card{I}$.
\end{restatable}

We can now conclude the proof of the 
Acyclic Universal Models Theorem (Theorem~\ref{thm:myuniv}) that we stated at the beginning of this section.
Start by applying the saturation process of the
Sufficiently Envelope-Saturated Solutions Proposition
to obtain an aligned superinstance $J = (I, \cov)$ of some $k$-sound $I_0'$,
such that $J$ satisfies $\fds$ and is $(B \card{I})$-envelope-saturated.
Now, apply the Envelope-Thrifty Completion Proposition to obtain an aligned
superinstance $J'$ of $I_0$.
By the Envelope Blowup Lemma, $J'$ contains $< B \card{I}$ new facts, so,
by Lemma~\ref{lem:envfds}, $J'$ must still be
$1$-envelope-saturated. Hence, $J'$ satisfies~$\con$. This concludes the proof, as
$J'$ is an aligned superinstance of~$I_0$.

\section{Cyclic Queries}
\label{sec:cycles}
We now finally complete our proof
of the Universal Models Theorem (Theorem~\ref{thm:univexists})
by moving from
acyclic Boolean $\cq$s to arbitrary Boolean $\cq$s.
We do so by a generic process which is essentially independent from our previous
construction.

Intuitively, the only cyclic $\cq$s that hold in $\chase{I_0}{\ids}$ either
have an acyclic self-homomorphic match (so they are implied by an acyclic $\cq$ that
also holds) or have all cycles matched to elements of~$I_0$.
Hence, in a $k$-sound instance for $\cq$, no other cyclic queries must be true.
We ensure this by a cycle blowup process that takes
the product of our~$I$ with a group of high girth,
following Otto~\cite{otto2002modal}. However, we 
need to adjust this construction to avoid creating $\fd$ violations.

We let $J_{\f} = (I_{\f}, \cov)$ be the aligned superinstance
obtained from
the Acyclic Universal Models Theorem (Theorem~\ref{thm:myuniv2}).
Its underlying instance $I_{\f}$ is a finite superinstance of
$I_0$ that satisfies $\con$, and the $k$-bounded simulation
$\cov$ guarantees that $I_{\f}$ is $k$-sound for $\acq$. Our goal in this section is to make
$I_{\f}$ $k$-sound for $\cq$ while still satisfying $\con$, so that it is
$k$-universal. This will conclude the proof of the Universal Models Theorem
(Theorem~\ref{thm:univexists}).

\subsection{Simple product}

Let us first introduce preliminary notions:

\begin{definition}
  A \deft{group} $G = (S, \cdot)$ over a finite set~$S$ consists of an
  associative \deft{product law} $\cdot : S^2 \rightarrow S$, a \deft{neutral
  element} $e \in S$, and an \deft{inverse law} $\cdot^{-1} : S \rightarrow S$
  such that $x \cdot x^{-1} = x^{-1} \cdot x = e$ for all $x \in S$. We say that $G$ is
  \deft{generated} by $X \subseteq S$ if all elements of~$S$ can be written as a
  product of elements of~$X$ and $X^{-1} \defeq \{x^{-1} \mid x \in X\}$.

  Given a group $G$ generated by $X$, the
  \deft{girth} of~$G$ under $X$ is the length of the shortest non-empty word
  $\mybf{w}$ of elements of~$X$ and $X^{-1}$
  such that $w_1 \cdots w_n = e$ and $w_i \neq w_{i+1}^{-1}$ for all $1 \leq i <
  n$. (If $X =
  \{g\}$ with $g = g^{-1}$, the girth is~$1$.)
  % Note that this definition is equivalent to the one using cycles of the Cayley graph.
\end{definition}

\begin{restatable}[\cite{margulis1982explicit}]{lemma}{groups}
  \label{lem:groups}
  % via ALON , N. 1995. Tools from higher algebra. In Handbook of Combinatorics,
  % R. Graham, M. Grötschel, and Lovasz, Eds., Vol. II., North-Holland, 1749–1783.
  % via Otto, Highly acyclic groups, hypergraph covers, and the guarded
  % fragment, section 2.1 (2012)
  % note that if |X| = 1 we need a different construction, e.g., Z/nZ
  For all $n \in \mathbb{N}$ and finite non-empty set~$X$, there is a finite group $G = (S,
  \cdot)$ generated by $X$
  with girth $\geq n$ under $X$.
  We call $G$ an \deft{$n$-acyclic group generated by $X$}.
\end{restatable}

In other words, in an $n$-acyclic group generated by $X$, there is no short
product of elements of~$X$ and their inverses which evaluates to~$e$, except
those that include a factor $x x^{-1}$.

We now take the product of~$I_{\f}$ with such a finite group~$G$. This 
ensures that any cycles in the product instance are large, because they project to
cycles in~$G$. We use a specific generator:

\begin{definition}
  The \deft{fact labels} of a superinstance~$I$ of~$I_0$ are $\lab{I} \defeq
  \{\l^F_i \mid F \in I \backslash I_0, 1 \leq i \leq \card{F}\}$.
\end{definition}

Now, we define the product of a superinstance $I$ of $I_0$ with a group generated by
$\lab{I}$. We make sure not to blow up cycles in~$I_0$, so the result remains a
superinstance of~$I_0$:

\begin{definition}
  \label{def:prod}
  Let $I$ be a finite superinstance of~$I_0$ and 
  $G$ be a finite group generated by $\lab{I}$. The \deft{product of~$I$
  by~$G$ preserving $I_0$} 
  is the finite instance
  $(I, I_0) \sprod G$
  with domain $\dom(I) \times G$
  consisting of the following facts, for all $g \in G$:
  \begin{compactitem}
  \item For every fact $R(\mybf{a})$ of~$I_0$, the fact $R((a_1, g), \ldots,
    (a_{\card{R}}, g))$.
  \item For every fact $F = R(\mybf{a})$ of~$I \backslash I_0$, the
    following fact:\\
    $R((a_1, g \cdot \l^F_1), \ldots, (a_{\card{R}}, g \cdot \l^F_{\card{R}}))$.
  \end{compactitem}
  We identify $(a, e)$ to $a$ for $a \in \dom(I_0)$,
  so $(I, I_0) \sprod G$ is still a superinstance of~$I_0$.
\end{definition}

We say a superinstance $I$ of~$I_0$ is \deft{$k$-instance-sound} (for 
$\con$) if for any $\cq$ $q$ such that $\card{q} \leq k$, if $q$ has a match
in $I$ involving an element of~$I_0$, then $\chase{I_0}{\ids} \models q$. We can
ensure that $I_{\f}$ is $k$-instance-sound, up 
to having performed $k$ chase rounds on $I_0$ initially.
We can then state the following property:

\begin{restatable}[Simple product]{lemma}{prodppty}
  \label{lem:prodppty}
  Let $I$ be a finite superinstance of~$I_0$ and
  $G$ a finite $(2k+1)$-acyclic group generated by~$\lab{I}$.
  If $I$ is $k$-sound for $\acq$ and $k$-instance-sound, then $(I,
  I_0) \sprod G$ is $k$-sound for $\cq$.
\end{restatable}

\begin{example}
  Consider $F_0 \defeq R(a, b)$, $I_0 \defeq \{F_0\}$, and $\ids$ consisting of $\tau: \ui{R^2}{S^1}$, $\tau':
  \ui{S^2}{R^1}$, $\tau^{-1}$, and $(\tau')^{-1}$. Let $F \defeq S(b, a)$, and
  $I \defeq I_0 \sqcup \{F\}$. $I$ satisfies $\ids$ and is sound for $\acq$, but not for
  $\cq$: take for instance $q : \exists x y ~ R(x, y) \wedge S(y, x)$, which is
  cyclic and holds in $I$ while $(I_0, \ids) \not\entunr q$.

  We have $\lab{I} = \{\l^F_1, \l^F_2\}$. Identify $\l^F_1$ and $\l^F_2$
  to~$1$ and $2$ and
  consider the group $G \defeq (\{0, 1, 2\}, \cdot)$ where $\cdot$ is
  addition modulo~3. $G$ has girth 2 under $\lab{I}$.

  The product $I_\p \defeq (I, I_0) \sprod G$, writing pairs as subscripts for brevity, is
  $\{R(a_0, b_0), \allowbreak R(a_1, b_1),\allowbreak  R(a_2, b_2),\allowbreak
    S(b_1, a_2),\allowbreak  S(b_2, a_0),\allowbreak  S(b_0,
  a_1)\}$. In this case $I_\p$ happens to be $5$-sound for $\cq$.
\end{example}

We cannot conclude directly with the simple product, because
$I_\p \defeq (I_\f, I_0) \sprod G$ may violate $\ufds$ even though $I_\f \models
\fds$.
Indeed, there may be a relation~$R$,
a $\ufd$ $\phi: R^p \rightarrow R^q$ in~$\ufds$,
and two $R$-facts $F$ and $F'$ in~$I_\f \backslash I_0$
with $\pi_{R^p, R^q}(F) = \pi_{R^p, R^q}(F')$.
In $I_\p$ the images of~$F$ and $F'$ may overlap only on~$R^p$,
so they could violate~$\phi$.

\subsection{Mixed product}

What we need is a more refined notion of product, that does not attempt to
blow up cycles within fact overlaps.
To define it, we need to consider a \defo{quotient} of~$I_\f$:

\begin{definition}
  The \deft{quotient} $\quot{I}{\sim}$ of an instance $I$ by an equivalence
  relation $\sim$ on $\dom(I)$ is defined as follows:
  \begin{compactitem}
  \item $\dom(\quot{I}{\sim})$ is the equivalence classes of $\sim$ on $\dom(I)$,
  \item $\quot{I}{\sim}$ contains one fact $R(\mybf{A})$
    for every fact $R(\mybf{a})$ of~$I$,
    where $A_i$ is the $\sim$-class of~$a_i$ for all $R^i \in \pos(R)$.
  \end{compactitem}
  The \deft{quotient homomorphism} $\chi_{\sim}$ is the homomorphism from~$I$ to
  $\quot{I}{\sim}$ defined accordingly.
\end{definition}

We quotient $I_{\f}$ by the equivalence relation $\bbsim_k$ (recall
Definition~\ref{def:bbsim}), yielding $I_{\f}' \defeq
\quot{I_{\f}}{\bbsim_k}$. The resulting $I_\f'$ may no longer satisfy $\con$. However, it is still $k$-sound for $\acq$, for the following reason:

\begin{restatable}{lemma}{ksimhomom}
  \label{lem:ksimhomom}
  Any $k$-bounded simulation from an instance~$I$ to an instance $I'$ defines a
  $k$-bounded simulation from $\quot{I}{\bbsim_k}$ to~$I'$.
\end{restatable}

We then consider the homomorphism $\chi_{\bbsim_k}$ from $I_\f$ to~$I'_\f$,
and blow up cycles in $I_\f$ by a \defo{mixed product} that only distinguishes facts with a
different image in $I'_\f$ by $\chi_{\bbsim_k}$.
The point is that, as we show from our construction, facts of $I_\f$ that have
the same elements at the same positions always have the same $\bbsim_k$-class.
Hence, they are mapped to the same fact by~$\chi_{\bbsim_k}$ and will not be
distinguished by the mixed product.
Let us formalize this:

\begin{definition}
  \label{def:cautious}
  Let $I$ be a superinstance of~$I_0$ and $h$ be a homomorphism from $I$ to some
  instance~$I'$. 
  We say $I$ is \deft{cautious} for $h$ (and~$I_0$) if for any relation~$R$, for
  any two $R$-facts $F$ and $F'$ such that $\pi_{R^p}(F) = \pi_{R^p}(F')$
  for some $R^p \in \pos(R)$, either $F, F' \in I_0$, or $h(F) = h(F')$.
\end{definition}

\begin{restatable}[Cautiousness]{lemma}{cautious}
  \label{lem:cautious}
  The superinstance $I_\f$ of~$I_0$ constructed by the 
  Acyclic Universal Models Theorem (Theorem~\ref{thm:myuniv2}) is cautious for $\chi_{\bbsim_k}$.
\end{restatable}

The reason why $I_\f$ is cautious for $h \defeq \chi_{\bbsim_k}$ is that, except for facts of $I_0$,
overlaps between facts only occur when reusing envelope elements at
non-dangerous positions, in which case the $\cov$-images of both facts are
$\bbsim_k$-equivalent in $\chase{I_0}{\ids}$. We can then
show that, from our construction, such elements are actually $\bbsim_k$-equivalent in $I_\f$.

We now define the notion of mixed product, which uses the same fact label for
facts with the same image by~$h$:

\begin{definition}
  \label{def:prodm}
  Let $I$ be a finite superinstance of~$I_0$
  with a homomorphism $h$ to another finite superinstance $I'$ of~$I_0$
  such that $\restr{h}{I_0}$ is the identity and $\restr{h}{(I\backslash I_0)}$
  maps to~$I' \backslash I_0$.
  Let $G$ be a finite group generated by $\lab{I'}$.

  The \deft{mixed product} of~$I$ by $G$ via $h$ preserving $I_0$,
  written $(I, I_0) \mprod^h G$,
  is the finite superinstance of~$I_0$ with domain $\dom(I) \times G$
  consisting of the following facts, for every $g \in G$:
  \begin{compactitem}
  \item For every fact $R(\mybf{a})$ of~$I_0$, the fact $R((a_1, g), \ldots,
    (a_{\card{R}}, g))$.
  \item For every fact $R(\mybf{a})$ of~$I \backslash I_0$, the following
    fact:\\
    $R((a_1, g \cdot \l^{h(F)}_1), \ldots, (a_{\card{R}}, g \cdot \l^{h(F)}_{\card{R}}))$.
  \end{compactitem}
\end{definition}

We now show that the mixed product preserves $\uid$s and $\fd$s when
cautiousness is assumed.

\begin{restatable}[Mixed product preservation]{lemma}{mixedprod}
  \label{lem:mixedprod}
  For any $\uid$ or $\fd$ $\tau$, if $I \models \tau$
  and $I$ is cautious for~$h$, then $(I, I_0) \mprod^h G \models \tau$.
\end{restatable}

Second, we show that $h: I \to I'$ lifts to a homomorphism from the mixed product to the
simple product.

\begin{restatable}[Mixed product homomorphism]{lemma}{mixedhom}
  \label{lem:mixedhom}
  There is a homomorphism from $(I, I_0) \mprod^h G$ to~$(I', I_0) \sprod G$
  which is the identity on~$I_0 \times G$.
\end{restatable}

We can now conclude our proof of the Universal Models Theorem
(Theorem~\ref{thm:univexists}).
We construct $J_{\f} =
(I_{\f}, \cov)$ by the Acyclic Universal Models Theorem
(Theorem~\ref{thm:myuniv2}) and consider $I_\f$.
It is a finite superinstance of~$I_0$
which is $k$-universal for~$\con$ and~$\acq$.
Further,
up to having distinguished the elements of~$I_0$ with fresh predicates
and having performed initial chasing,
we can ensure that $I'_\f \defeq \quot{I_\f}{\bbsim_k}$
is $k$-instance-sound
and that the homomorphism $\chi_{\bbsim_k}: I_\f \to I_\f'$
satisfies the hypotheses of the mixed product.

Let $G$ be a $(2k+1)$-acyclic group generated by $\lab{I'_\f}$,
and consider $I_\p \defeq (I'_\f, I_0) \sprod G$.
As $I_\f$ was $k$-sound for $\acq$,
so is~$I'_\f$ by Lemma~\ref{lem:ksimhomom},
and as $I'_\f$ is also $k$-instance-sound,
$I_\p$ is $k$-sound for~$\cq$ by the Simple Product Lemma
(Lemma~\ref{lem:prodppty}).
However, as we explained, in general $I_\p \not\models \con$.
We thus construct $I_\m \defeq (I_\f, I_0) \mprod^h G$,
with $h \defeq \chi_{\bbsim_k}$.
By the Mixed Product Homomorphism Lemma,
$I_\m$ has a homomorphism to $I_\p$,
so it is also $k$-sound for $\cq$.
Further,
$I_\f$ is cautious for~$\chi_{\bbsim_k}$ by the Cautiousness Lemma,
so, by the Mixed Product Preservation Lemma,
we have $I_\m \models \con$ because $I_\f \models \con$. 

Hence, the mixed product $I_\m$ is a finite $k$-universal instance for $\con$ and $\cq$.
This concludes the proof of the Universal Models Theorem, and hence of our main
theorem (Theorem~\ref{thm:main}).

\section{Conclusion}
\label{sec:conclusion}
In this work we have developed the first techniques on arbitrary arity schemas
to build finite models that satisfy both referential constraints and number
restrictions, while controlling which $\cq$s are satisfied. We
have used this to prove that finite open-world query answering for $\cq$s,
$\uid$s and $\fd$s is finitely controllable up to finite closure of the
dependencies. Using this, we have isolated the complexity of FQA for $\uid$s and $\fd$s.

As presented the constructions are quite specific to dependencies, but
in future work we will look to extend them to constraint languages
containing disjunction, with the goal of generalizing to higher arity
the rich arity-2 constraint languages of, e.g.,
\cite{ibanezgarcia2014finite,pratt2009data}, while maintaining the decidability
of FQA.

\myparagraph{Acknowledgements}
This work was supported in part by the
Engineering and Physical Sciences Research Council, UK (EP/G004021/1) and the
French ANR NormAtis project. We are very grateful to Balder ten
Cate, Thomas Gogacz, Andreas Pieris, and Pierre Senellart for comments on earlier
drafts, and to the anonymous reviewers of LICS for their valuable feedback.

\bibliographystyle{abbrv}
\bibliography{lics-full}

\appendix
\section{Details about the $\uid$ chase and Unique Witness Property}
\label{apx:chase}

Recall the \emph{Unique Witness Property}:

\medskip 

For any element $a \in \dom(\chase{I}{\ids})$ and position $R^p$ of~$\sigma$,
if two facts of~$\chase{I}{\ids}$ contain~$a$ at position~$R^p$,
then they are both facts of~$I$.

\medskip

We first exemplify why this may not be guaranteed by the first round of the
$\uid$ chase. Consider the instance $I = \{R(a), S(a)\}$ and the $\uid$s
$\tau_1: \ui{R^1}{T^1}$ and $\tau_2: \ui{S^1}{T^1}$, where $T$ is binary.
Applying a round of the $\uid$ chase creates the instance $\{R(a), S(a), T(a, b_1), T(a,
b_2)\}$, with $T(a, b_1)$ being created by applying $\tau_1$ to the active fact
$R(a)$, and $T(a, b_2)$ being created by applying $\tau_2$ to the active fact
$S(a)$.

By contrast, the core chase would create only one of these two facts, because it
would consider that two new facts are \defp{equivalent}: they have the same
exported element occurring at the same position. In general, the core chase
keeps only one fact within each class of equivalent facts.

However, after one chase round by the core chase, there is no longer any
distinction between the $\uid$ chase and the core chase, because the following property
holds on the result $I'$ of a chase round (by the core chase or the UID chase)
on any instance $I''$: (*)~for any $\tau \in \ids$ and element $a \in \appelem{I'}{\tau}$, $a$ occurs in only one fact of~$I'$.
This is true because $\ids$ is transitively closed, so we know
that no $\uid$ of~$\ids$ is applicable to an element
of $\dom(I'')$ in $I'$; hence the only elements that witness violations occur
in the one fact where they were introduced in~$I'$.

We now claim that (*) implies the Unique Witness Property. Indeed, assume to
the contrary that $a \in \dom(\chase{I}{\ids})$ violates it.

If $a \in \dom(I)$, because $\ids$ is transitively closed, after the first chase
round on $I$, we no longer create any fact that involves~$a$. Hence, each one of $F_1$ and
$F_2$ is either a fact of~$I$ or a fact created in the first round of the chase
(which is a chase round by the core chase). However, if one of $F_1$ and $F_2$
is in $I$, then it witnesses that we could not have $a \in \appelem{I}{R^p}$, so it
is not possible that the other fact was created in the first chase round. It
cannot be the case either that $F_1$ and $F_2$ were both created in the first
chase round, by definition of the core chase. Hence, $F_1$ and $F_2$ are
necessarily both facts of~$I$.

If $a \in \dom(\chase{I}{\ids}) \backslash \dom(I)$, assume that $a$ occurs at position $R^p$ in two
facts $F_1$, $F_2$. As $a \notin \dom(I)$, none of them is a fact of~$I$. We
then show a contradiction. It is not possible that one of those facts was
created in a chase round before the other, as otherwise the second created fact
could not have been created because of the first created fact. Hence,
both facts must have been created in the same chase round.
So there was a chase round from $I''$ to $I'$ where we had $a \in
\appelem{I''}{R^p}$ and
both $F_1$ and $F_2$ were created respectively from active facts $F_1'$ and
$F_2'$ of $I''$ by $\uid$s
$\tau_1 : \ui{S^q}{R^p}$ and $\tau_2: \ui{T^r}{R^p}$.
But then, by property (*), $a$ occurs in only one fact, so as
it occurs in $F'_1$ and $F_2'$ we have $F_1' = F_2'$. Further, as $a \notin
\dom(I)$, $F'_1$ and $F'_2$ are not facts of~$I$ either, so by definition of the
UID chase and of the core chase, it is easy to see $a$ occurs at only one
position in $F'_1 = F'_2$.
This implies that $\tau_1 = \tau_2$. Hence, we must have $F_1 = F_2$.

\section{Proofs for Section~\ref{sec:overall}: Main Result and
Overall Approach}
\label{apx:overall}

\mysubsection{Proof of Proposition~\ref{prp:complexity} (Complexity of UQA for
$\fd$s and $\uid$s)}
\label{apx:prf_complexity}

\complexity*

We first show the results for $\uid$s in isolation.
UQA for $\uid$s is NP-complete in
combined complexity: the lower bound is immediate from query
evaluation~\cite{AHV},
the upper bound is by Johnson \& Klug~\cite{johnsonklug} and actually holds for
$\incd$s of arbitrary fixed arity (which they call ``width'').
For data complexity, Calì et al.~\cite{cali2003query} showed a PTIME (in fact,
AC$^0$) upper bound for arbitrary $\incd$s by observing that the certain answers
can be expressed by another first-order query.

We now show that the same upper bounds apply to UQA for $\uid$s and $\fd$s (the
lower bound clearly also applies). This result is implicit in prior
work of~\cite{cali2003decidability,gottlob2012towards}, but we prove
it here for completeness.
 We argue that  $\uid$s and $\fd$s are
 \emph{separable} 
This means that for any conjunction $\con$ of~$\fd$s $\fds$ and $\uid$s
$\ids$, for any instance~$I_0$ and $\cq$~$q$, if $I_0 \models \fds$
then we have $(I_0, \con) \entunr q \leftrightarrow (I_0, \ids) \entunr q$.
From this result, the upper bounds  follow from the
bounds for the $\uid$ case above, since checking whether $I_0 \models \fds$
can be done in PTIME.
Separability follows from the \emph{non-conflicting condition}
of~\cite{cali2003decidability,gottlob2012towards} but we give a
simpler argument.

Assume that $I_0$ satisfies  $\fds$.
Clearly if $(I_0, \ids) \entunr q$ then $(I_0, \con) \entunr q$.
We thus need to show that if $(I_0, \con) \entunr q$ then $(I_0, \ids) \entunr q$.
Consider $\chase{I_0}{\ids}$. If $\chase{I_0}{\ids} \models \fds$, then
$\chase{I_0}{\ids}$ is a superinstance of~$I_0$ that satisfies $\con$, so
because $(I_0, \con) \entunr q$ we must have $\chase{I_0}{\ids} \models q$. By
universality of the chase, this implies $(I_0, \ids) \entunr q$.

Hence, it suffices to show that $\chase{I_0}{\ids} \models \fds$.
Assume to the contrary the existence of~$F$ and $F'$ in
$\chase{I_0}{\ids}$ violating an $\fd$ of~$\fds$. There must exist a position
$R^p \in \pos(\sigma)$ such that $\pi_{R^p}(F) = \pi_{R^p}(F')$. By the 
Unique Witness Property, this implies that $F$ and $F'$ are facts of~$I_0$,
which is impossible by our
assumption that $I_0 \models \fds$.

\mysubsection{Proof of the Main Theorem (Theorem~\ref{thm:main}) from the
Universal Models Theorem (Theorem~\ref{thm:univexists})}
\label{apx:prf_main}

To show the Main Theorem from the Universal Models Theorem, let $\con$ be a
conjunction of $\fd$s and $\uid$s, $\con'$ its finite closure, and $I_0$ 
a finite instance. We want to show finite controllability up to finite closure,
namely, $(I_0, \con) \entfin q$ iff $(I_0, \con') \entunr q$.

We can assume without loss of generality that $I_0$ satisfies the $\fd$s of
$\con'$, as otherwise there is no superinstance of $I_0$ satisfying~$\con'$,
and both problems are always vacuously true.

It is clear that for any $\cq$ $q$, we have $(I_0, \con) \entfin q$ iff $(I_0,
\con') \entfin q$. Indeed, $\con'$ includes $\con$ and conversely any finite
superinstance of $I_0$ which satisfies $\con$ must satisfy $\con'$,
by definition of the finite closure. So in fact, to
prove finite controllability up to finite closure, it
suffices to show that $(I_0, \con') \entfin q$ iff $(I_0, \con') \entunr q$ for
any $\cq$ $q$. The
backward implication is immediate as all finite superinstances of~$I_0$ satisfying
$\con'$ are also unrestricted superinstances. We prove the contrapositive of the
forward implication.

Let $q$ be a $\cq$, let $k \defeq \card{q}$, and assume that $(I_0, \con')
\not\entunr q$. By the
Universal Models Theorem, let $I$ be a finite superinstance of~$I_0$ that is
$\card{q}$-sound and satisfies $\con'$. As $I$ is $\card{q}$-sound, we have $I
\not\models q$, so, as $I$ is a finite superinstance of $I_0$ that satisfies
$\con'$, it witnesses that $(I_0, \con') \not\entfin q$. This proves the desired
equivalence. Hence, we have established that $\con'$ is finitely controllable up to
finite closure, and have proved the Main Theorem.

\mysubsection{Proof of Corollary~\ref{cor:complexityfqa} (Complexity of FQA for
  $\fd$s and $\uid$s)}
\label{apx:prf_complexityfqa}
\complexityfqa*

By our Main Theorem (Theorem~\ref{thm:main}), any instance $(I, \con, q)$ to the
FQA problem, formed of an
instance $I$, a conjunction $\con$ of~$\incd$s $\ids$ and $\fd$s $\fds$, and a $\cq$ $q$,
reduces to the UQA instance $(I, \con', q)$, where $\con'$ is the finite
closure of~$\con$. Computing $\con'$ from $\con$ is data-independent, so the
PTIME data complexity result of Proposition~\ref{prp:complexity} clearly
still applies. It is also clear that the NP-hardness combined complexity bound of
Proposition~\ref{prp:complexity} can be re-proven for FQA, as it already held
even when $\con = \emptyset$. So we only need to show that the combined
complexity of FQA is in NP. A naive approach would be to compute explicitly
$\con'$ and solve the UQA instance $I$, $\con'$, $q$; but materializing $\con'$
may take exponential time.

Instead, remember that from our study of UQA complexity in the proof of
Proposition~\ref{prp:complexity}, UQA for $\uid$s and $\fd$s
can be performed by first checking the $\fd$s on the initial instance, and
then performing UQA for the $\uid$s in isolation. Hence, let $\ids'$ and $\fds'$
be the $\uid$s and $\fd$s of~$\con'$. Rather than materializing $\con'$, we will
show that we can decide whether $I \models \fds'$ in PTIME, and compute $\ids'$
in PTIME, which suffices
to prove the claim as the combined complexity of deciding whether $(I, \ids')
\entunr q$ is then in NP.

We first justify that we can indeed compute $\ids'$ in PTIME. We
consider every possible $\uid$ on positions occurring in $\con$ (there are
polynomially many), and for each of them, determine in PTIME from $\con$ whether
it is in $\con'$, using the implication procedure of Cosmadakis et al.~\cite{cosm}.
This allows us to compute $\ids'$ in PTIME.

We next justify that we can decide whether $I \models \fds'$ in PTIME.
For the same reason as for the $\uid$s, we can
compute in PTIME from $\con$ the set $\ufds'$ of the $\ufd$s which are in
$\con'$, by deciding implication for each possible $\ufd$. We now argue that to
test whether $I \models \fds'$, it suffices to test whether $I \models \fds$ and
whether $I \models \ufds'$. This follows if we can show that 
$\fds'$ is implied by $\ufds' \cup \fds$ by the usual axiomatization of
unrestricted and finite implication for $\fd$s alone, from
Armstrong~\cite{armstrong1974dependency}. Indeed, in this case, if $I \models
\fds'$ then $I \models \ufds' \cup \fds$ as it is a subset of~$\fds'$, and
conversely if $I \models \ufds' \cup \fds$ then $I$ satisfies $\fds'$ because
they are implied by $\ufds' \cup \fds$ so are also satisfied by any instance that
satisfies $\ufds' \cup \fds$.

To justify that $\fds'$ is implied by $\ufds' \cup \fds$, we use Theorem~4.1
of~\cite{cosm}, according to which a sound and complete axiomatization of the
finite closure of~$\fd$s and $\uid$s consists of the usual $\fd$ implication
rules, the standard $\uid$ axiomatization of Casanova et al.~\cite{casanova},
and the \emph{cycle rule}. 
So, consider any $\fd$ $\phi$ of~$\fds'$ and let us justify that it is implied
by $\ufds' \cup \fds$. If $\phi$ is a $\ufd$, then $\phi \in \fds'$.
Otherwise the last steps of a
derivation of~$\phi$ with the axiomatization of~\cite{cosm} must be rules from
the $\fd$ implication rules, as they are the only ones which can deduce
higher-arity $\fd$s. Let us group together the last $\fd$ implication rules that
were applied, and consider the set $S$ of the hypotheses to $\fd$ implication
rules that were not themselves produced by $\fd$ implication rules.
Each hypothesis from $S$ is either an $\fd$ of $\fds$ or was produced by the
cycle rule. Now, the cycle rule can only deduce $\ufd$s (and $\uid$s).
Hence, $S \subseteq \fds \cup \ufds'$, which implies that we can
construct a derivation of~$\phi$ from $\fds \cup \ufds'$ using the $\fd$
implication rules. Thus, we can indeed compute in PTIME $\ufds' \cup \fds$,
and check in PTIME whether $I \models \ufds' \cup \fds$, and we have shown that this is equivalent
to checking whether $I \models \fds'$. This concludes the proof.

\section{Proofs for Section~\ref{sec:wsound}: Weak-Soundness and
Reversible $\uid$s}
\label{apx:wsound}

This section proves the Acyclic Unary Weakly-Sound Models Proposition
(Proposition~\ref{prp:auwsm}), which weakens the Acyclic Unary Models Theorem
(Theorem~\ref{thm:myuniv}) by making assumption \assm{reversible} and replacing
$k$-soundness by weak-soundness (Definition~\ref{def:wssinstance}).

\mysubsection{Proof of Proposition~\ref{prp:balwsnd} (Satisfying $\uid$s in
balanced instances)}
\label{apx:prf_balwsnd}

\balwsnd*

For every relation $R$ of~$\sigma$, let $f_R$ be a bijection between
$\appelem{I}{R^1}$ and $\appelem{I}{R^2}$; this is possible, because $I$ is
balanced.

Consider the superinstance $I'$ of~$I$, with $\dom(I') = \dom(I)$, obtained by
adding, for every $R$ of~$\sigma$, the fact $R(a, f_R(a))$ for every $a \in
\appelem{I}{R^1}$. $I'$ is clearly a finite weakly-sound superinstance of~$I$,
because for every $a \in \dom(I')$, if $a$ occurs at some position $R^p$ in some
fact $F$ of~$I'$, then either $F$ is a fact of~$I$ and $a \in \pi_{R^p}(I)$, or
$F$ is a new fact and by definition $a \in \appelem{I}{R^p}$.

Let us show that $I' \models \ufds$. Assume to the contrary that there are two
facts $F$ and $F'$ in $I'$ that witness a violation of a $\ufd$ $\phi: R^p
\rightarrow R^q$ of~$\ufds$. As $I \models \ufds$, one of~$F$ and $F'$ is
necessarily a new fact; we assume without loss of generality that it is~$F$.
Consider $a \defeq \pi_{R^p}(F)$.
By definition of the new facts, we have $a \in \appelem{I}{R^p}$, so that $a
\notin \pi_{R^p}(I)$. Now, 
as $\{F, F'\}$ is a violation, we must have $\pi_{R^p}(F) = \pi_{R^p}(F')$, so
as $a \notin \pi_{R^p}(I)$, $F'$ must also be a new fact. Hence, by definition
of the new facts, letting $b \defeq \pi_{R^q}(F)$ and $b' \defeq
\pi_{R^q}(F')$, depending on whether $p = 1$ or $p = 2$ we have either $b = b' = f_R(a)$ or $b = b' = f_R^{-1}(a)$,
which is well-defined because $f_R$ is a bijection. This contradicts the fact
that $F$ and $F'$ violate $\phi$.

Let us now show that $I' \models \ids$. Assume to the contrary that there is
an active fact $F = R(a_1, a_2)$, for a $\uid$ $\tau : \ui{R^p}{S^q}$. If $F$ is a fact of
$I$, we had $a_p \in \appelem{I}{S^q}$, so $F$ cannot be an active fact in
$I'$ by construction of~$f_S$. So we must have $F \in I' \backslash I$. Hence,
by definition of the new facts,
we had $a_p \in \appelem{I}{R^p}$; so there must be $\tau' : \ui{T^r}{R^p}$ in $\ids$ such that
$a_p \in \pi_{T^r}(I)$. 
Hence, because $\ids$ is transitively closed, either $T^r = S^q$ or
the $\uid$
$\ui{T^r}{S^q}$ is in $\ids$. In the first case, as $a_p \in \pi_{T^r}(I)$, $F$
cannot be an active fact for~$\tau$, a contradiction.
In the second case, we had $a_p
\in \appelem{I}{S^q}$, which is a contradiction for the same reason as before.

Hence, $I'$ is a finite weakly-sound superinstance of~$I$ that satisfies $\ucon$
and with $\dom(I') = \dom(I)$, the desired claim.

\mysubsection{Proof of the Balancing Lemma
(Lemma~\ref{lem:hascompletion})}
\label{apx:prf_hascompletion}
\hascompletion*

We prove the lemma without assumption \assm{binary}, as we will use it without
this assumption later in Section~\ref{sec:wsound}.

For any position $R^p$ define $o(R^p) \defeq \appelem{I}{R^p} \sqcup
\pi_{R^p}(I)$. Intuitively, those are the elements that either appear at $R^p$
or want to appear there. We claim that $o(R^p) = o(S^q)$ whenever $R^p \eqids S^q$.
Indeed, we have $\pi_{R^p}(I) \subseteq o(S^q)$: elements in $\pi_{R^p}(I)$
want to appear at $S^q$ unless they already do, and in both cases they are in
$o(S^q)$. Likewise, elements of
$\appelem{I}{R^p}$ either occur at $S^q$, or at some other position $T^r$
such that $\ui{T^r}{R^p}$ is a $\uid$ of~$\ids$, so that by transitivity
$\ui{T^r}{S^q}$
also is, and so they want to be at $S^q$ unless they already are. Hence
$o(R^p) \subseteq o(S^q)$, and symmetrically $o(S^q) \subseteq o(R^p)$.

Let $N \defeq \max_{R^p \in \pos(\sigma)} \card{o(R^p)}$, which is finite.
We write $\eqidsclass{R^p}$ the $\eqids$-class of any position~$R^p$.
We define for each $\eqids$-class
$\eqidsclass{R^p}$ a set $p(\eqidsclass{R^p})$ of~$N - \card{o(R^p)}$ fresh
values. We let $\calH$ be the disjoint union of the $p(\eqidsclass{R^p})$ for all
classes $\eqidsclass{R^p}$, and set $\mapb$ to map the elements of
$p(\eqidsclass{R^p})$ to $\eqidsclass{R^p}$. We have thus defined
our pssinstance $P = (I, \calH, \mapb)$.

Let us now show that $P$ is balanced. Consider now two positions $R^p$ and
$R^q$ such that $\phi: R^p \rightarrow R^q$ and $\phi': R^q \rightarrow R^p$
are in
$\ufds$, and show that $\card{\appelem{P}{R^p}} = \card{\appelem{P}{R^q}}$. We
have $\card{\appelem{P}{R^p}} = \card{\appelem{I}{R^p}} +
\card{p(\eqidsclass{R^p})} = \card{o(R^p)} -
\card{\pi_{R^p}(I)} + N - \card{o(R^p)}$, which simplifies to
$N - \card{\pi_{R^p}(I)}$.
Similarly $\card{\appelem{P}{R^q}} = N - \card{\pi_{R^q}(I)}$.
Since  $I \models \ufds$ and $\phi$ and $\phi'$ are in $\ufds$
we know that
$\card{\pi_{R^p}(I)} = \card{\pi_{R^q}(I)}$.
From this the conclusion follows.

\mysubsection{Proof of the Binary Realizations Lemma
(Lemma~\ref{lem:balwsndcb})}
\label{apx:prf_balwsndcb}
\balwsndcb*

Let us construct a realization~$I'$ of~$P$. We construct bijections $f_R$ for
every relation $R$ between $\appelem{P}{R^1}$ and $\appelem{P}{R^2}$ as for
Proposition~\ref{prp:balwsnd}; this is possible, as $P$ is balanced.
We then construct $I'$ in the same way, by adding to $I$, for every $R$
of~$\sigma$, the fact $R(a, f_R(a))$ for every $a \in \appelem{P}{R^1}$.

We prove that $I'$ is a realization again by observing that whenever we create a
fact $R(a, f_R(a))$, then we have $a \in \appelem{P}{R^1}$ and $f_R(a) \in
\appelem{P}{R^2}$.

The fact that $I'$ satisfies $\ufds$ is for the same reason as for
Proposition~\ref{prp:balwsnd}.

We now show that $I'$ satisfies $\ids$. Assume to the contrary that there is an
active fact $F = R(a_1, a_2)$, for a $\uid$ $\tau: \ui{R^p}{S^q}$, so that $a_p
\in \appelem{I'}{R^p}$. If $a_p \in
\dom(I)$, then the proof is exactly as for Proposition~\ref{prp:balwsnd}. Otherwise, if $a_p \in \calH$,
clearly by construction of~$f_R$ and $I'$ we have $a_p \in \pi_{T^r}(I')$ iff
$T^r \in \mapb(a_p)$. Hence, as $a_p \in \pi_{R^p}(I')$ and as $\tau$ witnesses
by assumption \assm{reversible} that $R^p \eqids S^q$ , we
have $a_p \in \pi_{S^q}(I')$, contradicting the fact that $a_p \in
\appelem{I'}{S^q}$.

\mysubsection{Proof of Lemma ``Binary realizations are completions''
(Lemma~\ref{lem:wsrwss})}
\label{apx:prf_wsrwss}
\wsrwss*

Clearly $I'$ is a superinstance of~$I$. Let us show that it is weakly-sound.
Recall the definition of a weakly-sound superinstance:

\wssinstance*

Consider $a \in \dom(I')$ and $R^p \in \pos(\sigma)$ such that $a \in
\pi_{R^p}(I')$. As $I'$ is a realization, we know that either $a
\in \pi_{R^p}(I)$ or $a \in \appelem{P}{R^p}$. By definition of
$\appelem{P}{R^p}$, and because $\calH = \dom(I') \backslash \dom(I)$, this means
that either $a \in \dom(I)$ and $a \in \pi_{R^p}(I) \sqcup \appelem{I}{R^p}$, or
$a \in \dom(I') \backslash \dom(I)$ and $R^p \in \mapb(a)$.
Hence:
\begin{itemize}
  \item For any $a \in \dom(I)$ and $R^p \in \pos(\sigma)$, we have established
that $a \in \pi_{R^p}(I')$ implied that either $a \in \pi_{R^p}(I)$ or $a \in
\appelem{I}{R^p}$.
  \item For any $a \in \dom(I') \backslash \dom(I)$ and for any $R^p, S^q \in
\pos(\sigma)$, we know that $R^p, S^q \in \mapb(a)$, which implies that $R^p
\eqids S^q$, so $R^p = S^q$ or $\ui{R^p}{S^q}$ is in $\ids$.
\end{itemize}
So indeed the two conditions of weak-soundness hold.

\mysubsection{Proof of the Realizations Lemma
(Lemma~\ref{lem:balwsndc})}
\label{apx:prf_balwsndc}
\balwsndc*

Let $P = (I, \calH, \mapb)$ be the balanced pssinstance.
Recall
that the $\eqfun$-classes of~$\sigma$ are numbered $\Pi_1, \ldots,
\Pi_{\neqfunc}$.
By definition of being balanced (Definition~\ref{def:balanced}),
for any $\eqfun$-class $\Pi_i$,
for any two positions $R^p, R^q \in \Pi_i$,
we have $\card{\appelem{P}{R^p}} = \card{\appelem{P}{R^q}}$.
Hence, for all $1 \leq i \leq \neqfunc$,
let $s_i$ be the value of~$\card{\appelem{P}{R^p}}$
for any $R^p \in \Pi_i$.
For $1 \leq i \leq \neqfunc$, we let $m_i$ be the arity of~$\Pi_i$, and
number the positions of $\Pi_i$ as $R^{p^i_1}, \ldots, R^{p^i_{m_i}}$.
We define for each $1 \leq i \leq \neqfunc$ and $1 \leq j \leq m_i$
a bijection $\phi^i_j$
from $\{1, \ldots, s_i\}$ to $\appelem{P}{R^{p^i_j}}$.
We construct the piecewise realization $\pire = (K_1, \ldots, K_n)$
by setting each $K_i$ for $1 \leq i \leq \neqfunc$
to be $\pi_{\Pi_i}(I)$ plus the tuples
$(\phi^i_1(l), \ldots, \phi^i_{m_i}(l))$ for $1 \leq l \leq s_i$.

\medskip

It is clear that $\pire$
is indeed a piecewise realization, because whenever we create a tuple
$\mybf{a} \in \Pi_i$ for any $1 \leq i \leq \neqfunc$,
then, for any $R^p \in \Pi_i$, we have
$a_p \in \appelem{P}{R^p}$.

\medskip

Let us then show that $\pire$ is $\ufds$-compliant.
Assume by contradiction that
there is $1 \leq i \leq \neqfunc$
and $\mybf{a}, \mybf{b} \in K_i$ such
that $a_l = b_l$ but $a_r \neq b_r$ for some $R^l, R^r \in \Pi_i$.
As $I$ satisfies $\ufds$, we assume without loss of generality that
$\mybf{a} \in K_i \backslash \pi{\Pi_i}(I)$. Now either $\mybf{b} \in
\pi{\Pi_i}(I)$ or $\mybf{b} \in K_i \backslash \pi{\Pi_i}(I)$.

If $\mybf{b} \in \pi{\Pi_i}(I)$,
then we know that $b_l \in \pi_{R^l}(I)$, but we know by
construction that, as $\mybf{a} \in K_i \backslash \pi{\Pi_i}(I)$, we 
have $a_l \in \appelem{P}{R^l}$. Now, as $a_l = b_l$ and $b_l \in \dom(I)$, we
have $a_l \in \dom(I)$, so that by definition of $\appelem{P}{R^l}$ we have $a_l
\in \appelem{I}{R^l}$. Thus, as $a_l = b_l$, we have a contradiction.

Now, if $\mybf{b} \in K_i \backslash \pi{\Pi_i}(I)$, then,
writing $R^l = R^{p^i_j}$ and $R^r = R^{p^i_{j'}}$,
the fact that $a_l = b_l$ but $a_r \neq b_r$
contradicts the fact that $\phi^i_{j} \circ (\phi^i_{j'})^{-1}$ is injective.
Hence, $\pire$ is $\ufds$-compliant.

\medskip

Let us now show that $\pire$ is $\ids$-compliant.

We must show
that, for every $\uid$
$\tau: \ui{R^p}{S^q}$ of~$\ids$, we have $\appelem{\pire}{\tau} = \emptyset$, which means
that we have $\pi_{R^p}(\pire) \subseteq \pi_{S^q}(\pire)$. Let
$\Pi_i$ be the $\eqfun$-class of $R^p$, and assume to the contrary the existence of a
tuple $\mybf{a}$ of~$K_i$ such that $a_p \notin \pi_{S^q}(\pire)$.
Either we have $a_p \in \dom(I)$, or we have $a_p \in \calH$.

In the first case, as $a_p \notin \pi_{S^q}(\pire)$,
in particular $a_p \notin \pi_{S^q}(I)$,
and as $a_p \in \pi_{R^p}(I)$,
we have $a_p \in \appelem{I}{\tau}$,
so $a_p \in \appelem{I}{S^q}$.
By construction of~$\pire$, then,
letting $i'$ be the $\eqfun$-class of $S^q$ and letting $S^q = S^{p^{i'}_j}$,
as $\phi^{i'}_j$ is surjective,
we must have $a_p \in \pi_{S^q}(K_{i'})$,
that is, $a_p \in \pi_{S^q}(\pire)$, a contradiction.

In the second case,
clearly by construction we have $a_p \in \pi_{T^r}(\pire)$
iff $T^r \in \mapb(a_p)$,
so that, given that $\tau$ witnesses $R^p \eqids S^q$,
if $a_p \in \pi_{R^p}(\pire)$ then $a_p \in \pi_{S^q}(\pire)$,
a contradiction.

\medskip

We deduce that $\pire$ is indeed a $\ucon$-compliant piecewise realization
of~$P$, completing the proof.

\mysubsection{Proof of the Relation-Saturated Solutions Lemma
(Lemma~\ref{lem:wtsol})}
\label{apx:prf_wtsol}
\wtsol*

Recall the definition of an instance being \defo{relation-saturated}:
\wsaturated*

We now prove the lemma. For every relation $R$, either $R$ is not achieved by $I$
and $\ids$, or there is $n_R \in \mathbb{N}$ such that there is a $R$-fact of
$\chase{I}{\ids}$ generated at the $n_R$-th round of the chase. Let $n \defeq
\max_{R \in \sigma} n_R$. As the number of relations in $\sigma$ is finite,
$n$ is finite. Hence, letting $I'$ be the result of applying $n$ chase rounds
to $I$, it is clear that $I'$ is relation-saturated.

\mysubsection{Proof of Lemma ``Using realizations to get completions''
(Lemma~\ref{lem:sreal})}
\label{apx:prf_sreal}
\sreal*

Recall that we number $\Pi_1, \ldots, \Pi_n$ the $\eqfun$-classes of
$\pos(\sigma)$.
We first define the following notion:

\begin{definition}
  \label{def:inout}
  We say that $\Pi_j$ is an \deft{inner} $\eqfun$-class if it contains a
  position occurring in~$\ids$; otherwise, it is an \deft{outer} $\eqfun$-class.
\end{definition}

Intuitively, ``outer'' $\eqfun$-classes are those to which no $\uid$ of
$\ids$ can apply, so we can create fresh elements at the positions of these
classes without fear that $\uid$s will be applicable to the fresh elements.

We will use the notion of dangerous and non-dangerous positions from
Section~\ref{sec:ksound}:

\dangerdef*

Observe that, if $R^p \eqfun R^q$, then for $R^r \notin \{R^p, R^q\}$, we have 
$R^r \in \danger(R^p)$ iff $R^r \in \danger(R^q)$, and likewise for
$\nondanger(R^p)$ and $\nondanger(R^q)$. So it makes sense to define
$\danger(\Pi_i)$ or $\nondanger(\Pi_i)$, for $\Pi_i$ an $\eqfun$-class of
positions of some relation~$R$, to refer to the positions of $\pos(R) \backslash
\Pi_i$ that are dangerous or non-dangerous for some $R^p \in \Pi_i$ (and hence
for all of them).

We show a first lemma about the positions where $\fd$ violations may be
introduced:

\begin{lemma}
  \label{lem:nondgr}
  For any relation $R$ and $\fd$s $\fds$, for any $R^p \in \pos(R)$ and $\ufd$
  $R^q \rightarrow R^r$ of~$\fds$, if $R^q \in \nondanger(R^p)$ then $R^r \in
  \nondanger(R^p)$.
\end{lemma}

\begin{proof}
  Assume by contradiction that $R^r \notin \nondanger(R^p)$. Then either $R^r =
  R^p$ or $R^r \in \danger(R^p)$. The first case is impossible because of the
  $\ufd$ $R^q \rightarrow R^r$. So we have $R^r \in \danger(R^p)$. Hence, the $\ufd$ $R^r
  \rightarrow R^p$ is in $\ufds$, so that by transitivity the $\ufd$ $R^q
  \rightarrow R^p$ is in $\ufds$, again contradicting the fact that $R^q \in
  \nondanger(R^p)$.
\end{proof}

Fix the finite relation-saturated instance $I$ that satisfies $\ufds$, the
pssinstance $P$ of~$I$, and the finite $\ucon$-compliant piecewise realization
$\pire = (K_1, \ldots, K_{\neqfunc})$ of~$P$. Our
approach is to construct the desired superinstance~$I'$
as $I \sqcup I_1 \sqcup \cdots \sqcup
I_{\neqfunc}$, where the facts of each $I_i$ are constructed from $K_i$, as we
now explain. We call $\calF$ the set of
the fresh elements (not in $\dom(\pire)$) that will be created in the
construction, so that we will have $\dom(I') \subseteq \dom(\pire) \sqcup \calF$.

We consider every $1 \leq i \leq \neqfunc$. 
Let $R$ be the relation to which the positions of $\Pi_i$ belong.
If the relation $R$ is not achieved by $I$ and $\ids$, or if $\Pi_i$ is outer,
then we do not create any fact for $R$, and set $I_i \defeq \emptyset$. 
Otherwise, 
as $I$ is relation-saturated,
we choose one fact $R(\mybf{c})$ in $I$.
For every $\mybf{a} \in K_i \backslash \pi_{\Pi_i}(I)$, 
we create a fact
$F^i_{\mybf{a}} \defeq R(\mybf{b})$
in~$I_i$, with $b_p$ defined as follows for every $R^p \in \pos(\sigma)$:
\begin{itemize}
  \item If $R^p \in \Pi_i$, take $b_p \defeq a_p$. In
    other words, the tuple $\mybf{a}$ is used to fill $\mybf{b}$ at the positions of~$\Pi_i$.
  \item If $R^p \in \danger(\Pi_i)$, use a fresh element in $\calF$ for $b_p$.
    In other words, dangerous positions have to be filled with fresh elements
    (but this is no problem because we will show later that their classes are
    outer).
  \item If $R^p \in \nondanger(\Pi_i)$ is non-dangerous, take $b_p \defeq c_p$.
    In other words, we reuse the fact $R(\mybf{c})$ guaranteed by $I$ being
    relation-saturated to complete the non-dangerous positions.
\end{itemize}

We have thus constructed $I'$, which is clearly a finite superinstance of~$I$.
We first show the following claim:

\begin{lemma}
  \label{lem:nodupes}
  For any $1 \leq i \leq \neqfunc$ and $\mybf{a} \in K_i$ for which we create
  a fact $F^i_{\mybf{a}}$, for any $R^p \in \Pi_i$, the fact
  $F^i_{\mybf{a}}$ is the only fact of~$I'$ where $a_p$ occurs at position $R^{p}$.
\end{lemma}

This claim implies that the facts of $I$, and all the facts of the $I_i$ for $1
\leq i \leq \neqfunc$, are pairwise distinct. By this, we mean that we did not
try to recreate in~$I_i$ a fact that already existed in~$I$, and that we never
tried to create the same fact twice in the same~$I_i$ or in different~$I_i$.

\begin{proof}
  Fix $1 \leq i \leq \neqfunc$ and $\mybf{a} \in K_i$, and assume that we have
  created a fact $F^i_{\mybf{a}}$; fix $R^p \in \Pi_i$.

  We first show that we cannot have $a_p \in \pi_{R^p}(I)$. Assuming by
  contradiction that we do, let $F$ be a witnessing fact. By definition of a
  piecewise realization we have $\pi_{\Pi_i}(I) \subseteq K_i$, so
  $\pi_{\Pi_i}(F) \in K_i$. Hence, as $\pire$ is $\fds$-compliant, we have
  $\mybf{a} = \pi_{\Pi_i}(F)$; but we do not create facts for the tuple
  $\mybf{a} \in K_i$ if $\mybf{a} \in \pi_{\Pi_i}(I)$, which contradicts the
  fact that we created $F^i_{\mybf{a}}$.

  Second, we show that there cannot be another fact $F$ of $I' \backslash I$ such
  that $a_p = \pi_{R^p}(F)$.
  As $\pire$ is $\ufds$-compliant, there clearly cannot be such a fact $F^i_{\mybf{a}'}$
  for $\mybf{a}' \in K_i$, $\mybf{a} \neq \mybf{a}'$,
  with $a_p$ occurring at position $R^p$ of $F^i_{\mybf{a}'}$.
  Hence, $F$ is a fact $F^{i'}_{\mybf{a}'}$ for $i' \neq i$.
  Now, $\Pi_i$ and~$\Pi_{i'}$ are disjoint as $\eqfun$-classes,
  and thus we cannot have $R^p \in \Pi_{i'}$.
  So either $R^p \in \danger(\Pi_{i'})$ and $b_p \in \calF$,
  or $R^p \in \nondanger(\Pi_{i'})$ and $b_p \in \pi_{R^p}(I)$. The first case
  is impossible because elements of $\calF$ occur in only one fact, and we
  showed above that the second case was impossible. This concludes.
\end{proof}
  
We now show that $I'$ has the required properties.
Let us first show that $I'$ is weakly-sound. Recall the definition:

\wssinstance*

We begin by checking the first condition. Let $a \in \dom(I)$ and $R^p \in
\pos(\sigma)$ such that $a \in \pi_{R^p}(I')$, and
let $F$ be a fact of~$I'$ that witnesses it. 
If $F$ is a fact of $I$ then $a \in \pi_{R^p}(I)$ and $a$
does not witness a violation of weak-soundness.
So $F$ is a fact of $I' \backslash I$. Let
$i$ be the index of the $I_i$ that contains $F$, and $\mybf{a}$ be such that
$F = F^i_{\mybf{a}}$ (this is uniquely defined according to Lemma~\ref{lem:nodupes}).

We cannot have $R^p \in \danger(\Pi_i)$, because we would then have
$\Pi_{R^p}(F) \in \calF$, contradicting $a \in \dom(I)$. We cannot have $R^p \in
\nondanger(\Pi_i)$ either, because then $a = \pi_{R^p}(F)$ would imply that $a
\in \pi_{R^p}(I)$ which we already excluded.
Hence $R^p \in \Pi_i$.
Now, by definition of~$\pire$ being a piecewise realization, as $\mybf{a} \in
K_i$, we know that $a \in \pi_{R^p}(I)$ or $a \in
\appelem{P}{R^p}$. But we excluded $a \in \pi_{R^p}(I)$ above, and 
we assumed $a \in \dom(I)$, so $a \in \appelem{P}{R^p}$ translates to $a \in
\appelem{I}{R^p}$. Hence, $a$ does not witness a violation of weak-soundness.

We now check the second condition. Let $a \in \dom(I') \backslash \dom(I)$ and
$R^p, S^q \in \pos(\sigma)$ such that $a \in \pi_{R^p}(I') \cap \pi_{S^q}(I')$.
We must show that $R^p = S^q$ or $\ui{R^p}{S^q}$ is in $\ids$, that is,
$R^p \eqids S^q$.
Now either $a \in \calF$, or $a \in \calH$. If $a \in \calF$, 
observe that elements of $\calF$ occur at only one position in $I'$. Hence, 
necessarily $R^p = S^q$ which implies $R^p \eqids S^q$, and $a$ does not witness
a violation of weak-soundness. 
Thus, $a \in \calH$.

Let $F$ be a fact witnessing that $a \in \pi_{R^p}(I')$, and $F'$ a fact
witnessing that $a \in \pi_{S^q}(I')$. As $a \in \calH$, necessarily $F$ and
$F'$ are facts of $I' \backslash I$, so there are
$i$ and $i'$ such that $F$ and $F'$ are respectively facts of~$I_i$ and
$I_{i'}$. Clearly $a$ cannot occur in $F$ or $F'$ at a position of
$\danger(\Pi_i)$ or $\danger(\Pi_{i'})$ (they contain elements of $\calF$) or at
a position of $\nondanger(\Pi_i)$ or $\nondanger(\Pi_{i'})$ (they contain
elements of $\dom(I)$). Hence, $R^p \in \Pi_i$ and $S^q \in \Pi_{i'}$. Now, 
as $\pire$ is a piecewise realization, as $a \notin \dom(I)$, we conclude that $a \in
\appelem{P}{R^p}$ and $a \in \appelem{P}{S^q}$, and as $a \notin \dom(I)$ this
implies that $R^p \in \mapb(a)$ and $S^q \in \mapb(a)$, so that $R^p \eqids
S^q$, and $a$ does not witness a violation of weak-soundness.

Hence, $I'$ is weakly-sound.

\medskip

Let us now show that $I' \models \ufds$. Assume to the contrary the existence of
two facts $F$ and $F'$ that witness a violation of a $\ufd$ $\phi: R^p
\rightarrow R^q$ of~$\ufds$. As $I \models \ufds$, we assume without loss of
generality that $F$ is a fact of~$I' \backslash I$;
let $1 \leq i \leq \neqfunc$ and $\mybf{a}
\in K_i$ be such that $F = F^i_{\mybf{a}}$. We cannot have $R^p \in
\danger(\Pi_i)$, as then we would have $a_p \in \calF$, and elements of $\calF$
only occur in a single fact in $I'$. We cannot have $R^p \in \Pi_i$ either
because, by Lemma~\ref{lem:nodupes}, $F^i_{\mybf{a}}$ is the only fact of $I'$
where $a_p$ occurs at position $R^p$. So $R^p \in \nondanger(\Pi_i)$, and by
Lemma~\ref{lem:nondgr} we have $R^q \in \nondanger(\Pi_i)$ as well. 
Hence, letting $F'' = R(\mybf{c})$ be the fact of~$I$ used to fill the positions of
$\nondanger(\Pi_i)$ in~$F$, we
know that $a'_p = c_p$ and $a'_q = c_q$. Thus, as this makes it impossible that
$F' = F''$, we deduce that $F''$ and $F'$ also violate $\phi$.

Now, either $F'$ is also a fact of $I$ and we have a contradiction because $F''
\in I$ but $I \models \ufds$, or it is a
fact of $I' \backslash I$ and, by the same process that we applied to $F$, we
can replace it by a fact of $I$, reaching a
contradiction again. This proves that $I' \models \ufds$.

\medskip

Let us last show that $I' \models \ids$. Assume to the contrary the existence of
a $\uid$ $\tau: \ui{R^p}{S^q}$ of~$\ids$ and an element $a \in \dom(I')$ such
that $a \in \pi_{R^p}(I') \backslash \pi_{S^q}(I')$.
Let $F$ be a fact of~$I'$ witnessing that $a \in
\pi_{R^p}(I')$. Either $F$ is a fact of~$I$ or it is a fact of $I' \backslash I$. 

For the first case, if $F$ is a fact
of $I$, by definition of~$\pire$ being a realization, we have $a \in
\pi_{R^p}(\pire)$. As $\pire$ is $\ids$-compliant, we have $a \in
\pi_{S^q}(\pire)$, and letting $\mybf{a}$ be the witnessing tuple in $K_i$
where $\Pi_i$ is the $\eqfun$-class of $S^q$, we know that either $a \in
\pi_{S^q}(I)$ or $a \in \pi_{S^q}(F^i_{\mybf{a}})$. In the first sub-case there is
nothing to show. In the second sub-case it suffices to show that $F^i_{\mybf{a}}$
was indeed created, and this is the case because $\tau$ witnesses that~$\Pi_i$
is inner, and $F \in I$ witnesses that $R$ was achieved in $\chase{I}{\ids}$,
so $S$ must also be because of~$\tau$. This concludes the first case.

For the second case, if $F$ is a fact of $I' \backslash I$, write $F =
F^{i'}_{\mybf{a}}$. The existence of $F^{i'}_{\mybf{a}}$ implies that
$\Pi_{i'}$ is inner and $R$ is achieved in $\chase{I}{\ids}$; hence $S$ is,
because of $\tau$.
There are three possibilities: $R^p \in \nondanger(\Pi_{i'})$, $R^p \in
\Pi_{i'}$, or $R^p \in \danger(\Pi_{i'})$.
The first sub-case is $R^p \in \nondanger(\Pi_{i'})$; but then we could have
picked as witness for $a \in \pi_{R^p}(I')$ the fact $S(\mybf{c})$ of $I$ used
to define the non-dangerous positions, and we are back to the first case. The
second sub-case is $R^p \in \Pi_{i'}$; then we have $a \in \pi_{R^p}(\pire)$ by construction, so that as
$\pire$ is $\ids$-compliant we have $a \in \pi_{S^q}(\pire)$, and we conclude as
before. The only remaining sub-case
is the third sub-case, $R^p \in \danger(\Pi_{i'})$, so that $a_p \in \calF$. 
Now, as $R^p \in \danger(\Pi_{i'})$, we know that $R^p \rightarrow R^r$ is in
$\ufds$
for any
position $R^r$ of~$\Pi_{i'}$ that occurs in $\ids$ (such an $R^r$ exists because $\Pi_{i'}$ is
inner). Now, as $\tau$ witnesses that $R^p$ occurs in $\ids$, we know by
assumption \assm{reversible}
that $R^r \rightarrow R^p$ is in $\ufds$, so that $R^p \in \Pi_{i'}$. But we
assumed $R^p \in \danger(\Pi_{i'})$, a contradiction.

Hence we conclude that $I' \models \ids$.

\medskip

Hence, $I'$ is a finite superinstance of~$I$ which is weakly-sound and satisfies
$\ucon$. This concludes the proof.

\section{Proofs for Section~\ref{sec:ksound}: $k$-Soundness and
Reversible $\uid$s}
\label{apx:ksound}

This section completes the proof of the Acyclic Unary Models Theorem
(Theorem~\ref{thm:myuniv}) under assumption \assm{reversible}.

\mysubsection{Proof of Lemma~\ref{lem:ksimacq} ($\acq$s are preserved
  through $k$-bounded simulations)}
\label{sec:prf_ksimacq}

\ksimacq*

Fix the instance $I$. We will prove by induction on $n$ the following stronger
claim: for any $n \in \mathbb{N}$, for any $\acq$ $q$ of size $\leq n$ and any
variable $x$ of~$q$, if $q$ 
has a match in $I$ that maps $x$ to $a \in \dom(I)$, then for any
$b \in \dom(I')$ such that $(I, a) \leq_n (I', b)$, $q$ has a match in
$I'$ mapping $x$ to $b$.
The base case of~$n = 0$ corresponds to queries with no atoms, and it is trivial.

For the induction step, fix $n \in \mathbb{N}$, the query $q$, the variable
$x$ and the match $h$ from $q$ to $I$ that
maps $x$ to $a \in \dom(I)$. 
We define a reachability relation between variables of~$q$ as the
reflexive and transitive closure of the relation of co-occurring in some atom of~$q$.
If this relation consists of a single class, we say that $q$ is
\deft{connected}. As we can otherwise
rewrite $q$ as a conjunction of strictly smaller queries of~$\acq$ and process
all such queries separately using the induction hypothesis,
we assume without loss of generality that $q$ is connected.

Let $\calA = A_1, \ldots, A_m$ be the atoms of~$q$ where
$x$ occurs (this set of atoms is non-empty, by the connectedness assumption). Because
$q$ is an $\acq$,
each variable $y$ occurring in one of the $A_i$ occurs at most once: once per atom
(as the same variable cannot occur multiple times in an atom), and in only one
atom (as if $y$ occurs both in $A_{i_1}$ and $A_{i_2}$ then $A_{i_1}$,  $y$,
$A_{i_2}$, $x$ is a Berge cycle of~$q$). Let $Y$ be the set of the variables
occurring in the $A_i$ (not including $x$).

Because $q$ is acyclic and connected, the other variables of~$q$ can be partitioned
depending on the variable in $Y$ from which they are reachable without
using~$\calA$.
Hence, 
we can partition the remaining atoms of~$q$ into strictly smaller acyclic subqueries
$q_1(y_1, {\mybf{z}}_{\mybf{1}})$, \ldots, $q_l(y_l, {\mybf{z}}_{\mybf{l}})$ in $\acq$,
for $Y =
\{y_1, \ldots, y_l\}$, where the $\mybf{z}_{\mybf{j}}$ are pairwise disjoint sets of
variables.

Now, let $b \in \dom(I')$ be such that $(I, a) \leq_n (I', b)$.
For each atom $A_i = R({\mybf{x}})$ in $\calA$,
let $1 \leq p_i \leq \arity{R}$ be the one position such that $x_{p_i} = x$.
Consider the fact $F_i = R({\mybf{a}}_{\mybf{i}})$
that is the image of~$A_i$ in~$I$ by~$h$.
As $(I, a) \bsim_n (I', b)$, there exists a fact $F_i' = R({\mybf{b}}_{\mybf{i}})$ of~$I'$
with $b_{p_i} = b$ and with $(I, a_q) \leq_{n-1} (I', b_q)$
for all $1 \leq q \leq \arity{R}$.
Consider now each variable $y_j \in Y$ that occurs in $A_i$, letting $1 \leq q
\leq \arity{R}$ be the one position such that $x_q = y_j$,
and let $q_j(y_j, {\mybf{z}}_{\mybf{j}})$ be the subquery corresponding to~$y_j$.
We know that $(I, a_q) \bsim_{\card{q_j}} (I', b_q)$,
and that $q_j$ has a match in $I$ that maps $y_j$ to~$a_q$
(namely, the restriction~$h_j$ of the match $h$ to the subquery~$q_j$)
so that, by the induction hypothesis,
$q_j$ has a match $h_j'$ in $I'$ where $y_j$ is matched to~$b_q$.
Now, we can assemble the~$F_i'$ and all the matches $h_j'$ thus obtained,
because the $\mybf{z}_{\mybf{j}}$ are pairwise
disjoint, yielding a match~$h'$ of~$q$ in~$I'$ where $x$ is matched to~$b$.
This concludes the induction step.

Hence, the stronger claim is proven by induction. It remains to observe that it implies
the desired claim. Indeed, if $I \models q$ and there is a $n$-bounded
simulation $\cov$ from $I$ to
$I'$,
choose any variable $x$ in $q$ (if $q$ has no variables, the result is
vacuous), consider any match of~$q$ in~$I$ matching $x$ to~$a$, use $\cov$ to
define $b \defeq \cov(a)$, and 
deduce the existence of a match of~$q$ in $I'$ (matching~$x$ to~$b$)
using the claim that we have shown by induction.

\mysubsection{Proof of Lemma~\ref{lem:clasfino} ($\rfcl$ is finite)}
\label{apx:prf_clasfino}

\clasfino*

We first show that $\bbsim_k$ has only a finite number
of equivalence classes on $\chase{I_0}{\ids}$. Indeed, for any element $a \in
\dom(\chase{I_0}{\ids})$, by the Unique Witness Property,
the number of facts in which $a$ occurs is bounded by a constant depending only
on $I_0$ and $\ids$. Hence, there is a constant $M$ depending only on $I_0$, $\ids$, and
$k$, so that, for any element $d \in \dom(\chase{I_0}{\ids})$, the number of
elements of~$\dom(\chase{I_0}{\ids})$ which are relevant to
determine the $\bbsim_k$-class of~$d$ (that is, the elements whose distance to
$d$ in the Gaifman graph of $\chase{I_0}{\ids}$ is $\leq k$) is bounded by $M$.

This clearly implies that $\rfcl$ is finite, because the number of~$m$-tuples of
equivalence classes of~$\bbsim_k$ that occur in $\chase{I_0}{\ids}$ is then finite
for any $m \leq \max_{R \in \sigma} \arity{R}$, and $\positions(\sigma)$ is finite.

\mysubsection{Proof of the Fact-Saturated Solutions Lemma (Lemma~\ref{lem:nondetexhaust})}
\label{apx:prf_nondetexhaust}
\nondetexhaust*

For every $D \in \rfcl$, let $n_{D} \in \mathbb{N}$ be such
that $D$ is achieved
by a fact of~$\chase{I_0}{\ids}$ created at round $n_{D}$.
As $\rfcl$ is
finite, $n \defeq \max_{D \in \rfcl} n_D$ is finite. Hence, all classes
of $\rfcl$ are achieved after $n$ chase rounds on $I_0$.

Consider now $I_0'$ obtained from the aligned superinstance $I_0$ by $n$ rounds
of the $\uid$ chase, and $J_0 = (I_0', \ids)$.
It is clear that for any $D \in \rfcl$, there is an achiever $F =
R(\mybf{b})$ of
$D$ in $I_0'$.
Hence,
the corresponding fact in $J_0$ is an achiever of~$D$ in~$J_0$.

\mysubsection{Proof of the Fact-Thrifty Chase Steps Lemma
(Lemma~\ref{lem:ftok})}
\label{apx:prf_ftok}

We first prove the following lemma, which we will use to justify that we can
extend aligned instances.

\begin{lemma}
  \label{lem:addfact}
  Let $n \in \mathbb{N}$.
  Let $I_1$ and $I$ be instances and $\cov$ be a $n$-bounded simulation from
  $I_1$ to $I$. Let $I_2$ be a superinstance of~$I_1$ 
  defined by adding one fact $F_\n = R(\mybf{a})$ to $I_1$,
  and let $\cov'$ be a mapping from $I_2$
  to $I$ such that $\restr{\cov'}{I_1} = \cov$. Assume there is a fact $F_\w =
  R(\mybf{b})$ in $I$ such that, for all $R^i \in \pos(R)$, $\cov'(a_i) \bbsim_n b_i$. Then
  $\cov'$ is a $n$-bounded simulation from $I_2$ to $I$.
\end{lemma}

\begin{proof}
  We prove the claim by induction on $n$. The base case of $n = 0$ is immediate.

  Let $n > 0$, assume that the claim holds for $n-1$, and show that it holds for
  $n$. As $\cov$ is a $n$-bounded simulation, it is a $(n-1)$-bounded
  simulation, so we know by the induction hypothesis that $\cov'$ is a
  $(n-1)$-bounded simulation.

  Let us now show that it is a $n$-bounded simulation.
  Let $a \in \dom(I_2)$ be an element and
  show that $(I_2, a) \bsim_{n} (I, \cov'(a))$. To do this,
  choose $F = S(\mybf{a})$ a fact of~$I_2$ with $a_p = a$ for some $p$, 
  and show that there exists a fact $F' =
  S(\mybf{a}')$ of~$I$ with $a'_p = \cov'(a_p)$ and 
  $(I_2, a_q) \bsim_{n-1} (I, a'_q)$ for all $S^q \in \pos(S)$.
  
  The first possibility is that $F$ is the new fact $F_\n = R(\mybf{a})$.
  In this case, as we have $(I, b_p) \bsim_n (I, \cov'(a_p))$,
  considering $F_\w$, we deduce the existence of a fact
  $F_\w' = R(\mybf{c})$ in~$I$ such that $c_p = \cov'(a_p)$
  and $(I, b_q) \bsim_{n-1} (I, c_q)$ for all $1 \leq q \leq \arity{R}$.
  We take $F' = F_\w'$.
  By construction we have $c_p = \cov'(a_p)$.
  Fixing $1 \leq q \leq \arity{R}$,
  to show that $(I_2, a_q) \bsim_{n-1} (I, c_q)$,
  we use the fact that $\cov'$ is an $(n-1)$-bounded simulation
  to deduce that $(I_2, a_q) \bsim_{n-1} (I, \cov'(a_q))$.
  Now, we have $(I, \cov'(a_q)) \bsim_{n-1} (I, b_q)$,
  and as we explained we have $(I, b_q) \bsim_{n-1} (I, c_q)$,
  so we conclude by transitivity.
  
  If $F$ is another fact, then it is a fact of~$I_1$,
  so its elements are in $\dom(I_1)$,
  and as $\cov'$ coincides with $\cov$ on such elements,
  we conclude because $\cov$ is a $n$-bounded simulation.
\end{proof}
 
We then prove the main result:

\ftok*

We first observe that fact-thrifty chase steps are well-defined because a suitable
$F_\r = S(\mybf{c})$ always exists, as $J$ is fact-saturated. It is immediate
that $J'$ is finite.

It is immediate that, letting $J' = (I', \cov')$ be the result of the process,
$I'$ is still a superinstance of~$I_0$, and the previously active fact $F_\a$ is no longer active in $I'$.
To show that $\cov'$ is still a $k$-bounded simulation, use Lemma~\ref{lem:addfact}
with $F_\n = S(\mybf{b})$ and $F_\w = S(\mybf{b}')$.
The fact that $\cov'$ is the identity on~$I_0$
is immediate because $\restr{\cov'}{I_0} = \restr{\cov}{I_0}$.

We now show that $J'$ satisfies $\ufds$, using the fact that $J$ does. Indeed, any
violation of~$\ufds$ in $J'$ would have to include the one new fact $F_\n =
S(\mybf{b})$,
By way of contradiction, let $\phi: S^l \rightarrow S^r$ be a violated $\ufd$ in
$\ufds$ and let $\{F, F_\n\}$ be a violation, where $F = S(\mybf{d})$ is some
fact of~$I'$. It is clear that we cannot have $d_q = b_q$, as otherwise this
would contradict the fact that $F_\a$ was an active fact.
Hence, 
by construction of the new fact $F_\n$, we can only have $b_i =
d_i$ if $S^i \in \nondanger(S^q)$. As $\{F, F_\n\}$ violates $\phi$,
this implies that $S^l \in \nondanger(S^q)$, 
so that, by Lemma~\ref{lem:nondgr}, $S^r \in
\nondanger(S^q)$. Now, observe that we have $\pi_{\nondanger(S^q)}(F_\n) =
\pi_{\nondanger(S^q)}(F_\r)$, with $F_\r$ the fact used
to fill the non-dangerous position in the definition of fact-thrifty chase steps.
Now, we cannot have $F = F_\r$ because
they must disagree on $S^r$, so that $\{F, F_\r\}$ also witnesses a violation of~$\phi$ in
$J$.
This contradicts our assumption that $J \models \ufds$.

We must now check the last part of the definition of aligned
superinstances, which only needs to be verified for the fresh elements:
for $S^r \neq S^q$, if $b_r$ is fresh, then it occurs in $J'$ at the position where
$\cov(b_r)$ was introduced in $\chase{I_0}{\ids}$. For this, it suffices to show
that $b'_q$ was the exported element of $F_\w$. In this case, as $\cov(b_r) =
b'_r$, we will know that $b'_r$ was introduced at position $S^r$
in~$F_\w$ in~$\chase{I_0}{\ids}$,
so the condition is respected. We make this a separate
lemma:

\begin{lemma}
  \label{lem:wadef}
  Let $J$ be an aligned superinstance of~$I_0$ and consider
  the application of a thrifty chase step for a $\uid$ $\tau: \ui{R^p}{S^q}$. 
  Consider the chase witness $F_\w = S(\mybf{b}')$. Then $b'_q$ is the
  exported element of~$F_\w$.
\end{lemma}

Using this lemma, it is also clear that $\restr{\cov'}{I' \backslash I_0}$ maps
to $\chase{I_0}{\ids} \backslash I_0$, which is the last thing we had to verify.
Indeed, for all fresh elements $b_r \in \dom(I') \backslash \dom(I)$
(with $S^r \neq S^q$),
which are clearly not in~$I_0$,
we have fixed $\cov'(b_r)$ to be~$b_r'$,
which by the lemma is introduced in $F_\w$
so it cannot be an element of~$I_0$; hence it is indeed
an element of $\chase{I_0}{\ids} \backslash I_0$.

We conclude by proving Lemma~\ref{lem:wadef}:

\begin{proof}
  Let $F_\a = R(\mybf{a})$ be the active fact in $J$,
  $F_\n = S(\mybf{b})$ be the new fact of~$J'$,
  and $\tau: \ui{R^p}{S^q}$ be the $\uid$,
  so $a_p = b_q$ is the exported element of this chase step.
  Assume by way of contradiction 
  that $b'_q$ was not the exported element in~$F_\w$,
  so that it was introduced in $F_\w$.
  In this case, as $\cov(a_p) = \cov(b_q) = b'_q$,
  by the last part of the definition of aligned superinstances,
  we have $a_p \in \pi_{S^q}(J)$,
  which contradicts the fact that $a_p \in \appelem{J}{\tau}$.
  Hence, we have proved by contradiction
  that $b'_q$ was the exported element in~$F_\w$.
\end{proof}

\mysubsection{Proof of the Fact-Thrifty Completion Proposition
(Proposition~\ref{prp:ftcomp})}
\label{apx:prf_ftcomp}

\ftcomp*

There are two steps to the proof.
The first one is to apply initial chasing by
fresh fact-thrifty chase steps to ensure a certain property, \emph{$k$-reversibility}.
The
second one is to use fact-thrifty chase steps to satisfy $\ids$, using the
constructions of Section~\ref{sec:wsound}.

\medskip

We start with the first step.
We consider a forest structure on the facts of
$\chase{I_0}{\ids}$: the facts of~$I_0$ are the roots, and the parent of a fact
$F$ not in $I_0$
is the fact~$F'$ that was the active fact for which $F$ was created,
so that $F'$ and~$F$ share the exported element of~$F$.
For $a \in \dom(\chase{I_0}{\ids})$, if $a$ was introduced at position $S^r$ of an
$S$-fact $F = S(\mybf{a})$ created by applying the $\uid$ $\tau: \ui{R^p}{S^q}$
(with $S^q \neq S^r$) to its parent fact~$F'$,
we call $\tau$ the \defo{last $\uid$} of~$a$. The last two
$\uid$s of~$a$ are $(\tau, \tau')$ where $\tau'$ is the last $\uid$ of the
exported element~$a_q$ of~$F$ (which was introduced in~$F'$).
For $n \in \mathbb{N}$,
we define the last $n$ $\uid$s in the same way, for elements of
$\chase{I_0}{\ids}$ introduced after sufficiently many rounds. We say that $a$ is
\deft{$n$-reversible} if its last $n$ $\uid$s are reversible.

We accordingly define the notion of 
\defp{$n$-reversible aligned superinstance},
which requires that elements where a $\uid$ is violated
are mapped by $\cov$ to a $n$-reversible element in the chase.
Recall that, for any position $R^p$, we write $\eqidsclass{R^p}$
the $\eqids$-class of~$R^p$.

\begin{definition}
  \label{def:nreverse}
  An aligned superinstance $J$ of~$I_0$ is \mbox{\deft{$n$-reversible}} if for any
  position $S^q$
  and $a \in \appelem{J}{S^q}$,
  $\cov(a)$
  is a $n$-reversible element of~$\chase{I_0}{\ids}$ introduced at a position of
  $\eqidsclass{S^q}$ in $\chase{I_0}{\ids}$.
\end{definition}

The first step of the proof of Proposition~\ref{prp:ftcomp} is to perform $k+1$
fresh fact-thrifty chase rounds on the input fact-saturated aligned
superinstance~$J$, to ensure that the result~$J'$ is $k$-reversible for
$\ids$:

\begin{restatable}[Ensuring $n$-reversibility]{proposition}{achkrev}
  \label{prp:achkrev}
  For any $n \in \mathbb{N}$, applying $n+1$ fresh fact-thrifty chase rounds on
  a fact-saturated aligned superinstance $J$ by the
  $\uid$s of~$\ids$ yields a fact-saturated aligned superinstance $J'$ that is
  $n$-reversible for $\ids$.
\end{restatable}

This proposition is proved in Appendix~\ref{apx:prf_achkrev}.

\medskip

The second step of the proof is simply to apply the following lemma to~$J'$.

\begin{restatable}[Guided chase]{lemma}{guidedb}
  \label{lem:guidedb}
  For any fact-saturated \mbox{$k$-reversible} aligned superinstance $J = (I, \cov)$
  of~$I_0$,
  we can build by fact-thrifty chase steps
  an aligned superinstance $J' = (I', \cov')$
  of~$I_0$ such that $I \subseteq I'$, $\restr{\cov'}{I} =
  \cov$, and $J'$ satisfies $\ids$.
\end{restatable}

The lemma is proved in Appendix~\ref{apx:prf_guidedb}.
It uses the constructions of Section~\ref{sec:wsound},
and relies on an independent result about the UID chase,
the Chase Locality Theorem (Theorem~\ref{thm:locality}),
proved in Appendix~\ref{apx:prf_locality}.
Clearly, applying the Guided Chase Lemma to~$J'$ concludes the proof
of the Fact-Thrifty Completion Proposition.

\mysubsection{Proof of Proposition ``Ensuring $n$-reversibility''
(Proposition~\ref{prp:achkrev})}
\label{apx:prf_achkrev}

We first make the following easy observation:

\begin{lemma}
  \label{lem:chasestruct}
  Let $J$ be an aligned superinstance, and $J'$ be the result of applying 
  one chase round to $J$ with fresh fact-thrifty chase steps. Let $a \in
  \appelem{J'}{\tau}$  for any $\uid$ $\tau$. Then we have $a \in \dom(J')
  \backslash \dom(J)$, and $a$ occurs in a single fact $F$ (which is an active
  fact for $\tau$).
\end{lemma}

\begin{proof}
  For the first part of the claim, let us assume by way of contradiction that $a
  \in \dom(J)$. Note that, by definition of chase rounds, we
  cannot have $a \in \appelem{J}{\tau}$, otherwise we could not have $a \in
  \appelem{J'}{\tau}$. Hence, if we have $a \in \dom(J)$ but $a \notin
  \appelem{J}{\tau}$, any active fact $F$ witnessing $a \in \appelem{J}{\tau}$ must be
  in $J'$.

  Now, by definition of fact-thrifty chase steps, if $a \notin \dom(J)$, 
  there are two
  possibilities. Either $a$ was the exported element in~$F$, or it was an element
  reused at a non-dangerous position. The first case is impossible: because $\ids$ is
  transitively closed, the new facts created in $J'$ cannot make new
  $\uid$s applicable to old elements of~$J$. The second case is also
  impossible: elements reused at non-dangerous positions already occurred at the
  same position in $J$, so this cannot make new $\uid$s applicable to them.
  This proves the first claim.
  
  The second part of the claim
  is by observing that elements created in $J'$ occur in a
  single fact, by definition of chase rounds, and by definition of fresh
  fact-thrifty chase steps (elements in new facts are either in $\dom(J)$ or are
  fresh). So the one fact where $a$ occurs must be the active fact witnessing
  that $a \in \appelem{J'}{\tau}$.
\end{proof}

We then show the following simple lemma about $n$-reversibility:

\begin{lemma}
  \label{lem:newkrev}
  Let $n \in \NN$,
  let $J$ be a $n$-reversible aligned superinstance of~$I_0$ and let $F_\n
  = S(\mybf{b})$
  be a new fact obtained by applying a thrifty chase step to $J$. For all $S^r
  \in \pos(S)$, such that $b_r \notin \dom(J)$, $\cov(b_r)$ is $(n+1)$-reversible
  and introduced at position~$S^r$ in $\chase{I_0}{\ids}$.
\end{lemma}

\begin{proof}
  Let $F_\a$ be the active fact, $F_\w$ be the chase witness, and $\tau:
  \ui{R^p}{S^q}$ be the $\uid$ for this chase step. By Lemma~\ref{lem:wadef} we
  know that $b'_q$ is the exported element of~$F_\w$. Hence, for all $S^r \in
  \pos(S) \backslash \{S^q\}$, $b'_r$ is $(n+1)$-reversible and introduced at
  position~$S^r$.
  Now, for all $S^r \in \pos(S)$ such that $b_r$ is fresh in $F_\n$,
  we have set $\cov(b_r) = b'_r$, so the result follows.
\end{proof}

We now prove the main result:
\achkrev*

Fix the aligned superinstance $J = (I, \cov)$.
We prove the result by induction on~$n$. For the base case $n = 0$, 
letting $J'$ be the result of applying one chase round to~$J$,
we need only show that for any position $S^q$ and $a \in \appelem{J'}{S^q}$,
$\cov(a)$ was introduced at a position of~$\eqidsclass{S^q}$ in
$\chase{I_0}{\ids}$. By Lemma~\ref{lem:chasestruct}, $a$ occurs in a single fact
$F$ at some position~$R^p$ (so that, using assumption \assm{reversible},
$R^p \eqids S^q$), and we have $a \in \dom(J') \backslash \dom(J)$, so it was
created by the application of a thrifty chase step to~$J$. By
Lemma~\ref{lem:newkrev}, we conclude that $\cov(a)$ was introduced at position
$R^p$ in~$\chase{I_0}{\ids}$, which implies the desired claim.

For the induction, fix $n > 0$ and assume that the result is true for $n-1$. Let
$J' = (I', \cov')$ be the result of applying $(n-1)+1$ chase rounds to $J$.
By induction hypothesis,
$J$ is $(n-1)$-reversible. We want to show that $J'' = (I'', \cov'')$ obtained
by applying one more chase round to $J'$ is $n$-reversible.
This is shown exactly as in the base case, except that,
when applying Lemma~\ref{lem:newkrev}, we use the $(n-1)$-reversibility of~$J$
to deduce the $n$-reversibility of the element under consideration.

This proves the desired claim by induction. Note that
we have relied implicitly on the
Fact-Thrifty Chase Steps Lemma (Lemma~\ref{lem:ftok})
to justify that the result of chase rounds by fact-thrifty chase steps
are indeed aligned superinstances; it is immediate that
fact-saturation is preserved.

\mysubsection{Proof of the Guided Chase Lemma (Lemma~\ref{lem:guidedb})}
\label{apx:prf_guidedb}

Recall that the $\eqfun$-classes of~$\pos(\sigma)$
are numbered $\Pi_1, \ldots, \Pi_{\neqfunc}$.
Recall the notion of inner and outer $\eqfun$-classes
(Definition~\ref{def:inout}),
and the notion of piecewise realization
(Definition~\ref{def:superbvpw}).
We define:

\begin{definition}
  \label{def:follows}
  A superinstance $I'$ of the instance $I$
  \deft{follows} the piecewise realization $\pire = (K_1, \ldots, K_{\neqfunc})$
  if for every inner $\eqfun$-class $\Pi_i$, we have
  $\pi_{\Pi_i}(I') \subseteq K_i$.
\end{definition}

We show the main claim:

\guidedb*

Fix the fact-saturated $k$-reversible aligned superinstance $J = (I, \cov)$
of $I_0$. Let $P = (I, \calH, \mapb)$ be a balanced pssinstance of~$J$
obtained by the Balancing Lemma (Lemma~\ref{lem:hascompletion}) and let
$\pire = (K_1, \ldots, K_{\neqfunc})$
be a finite $\ucon$-compliant piecewise realization of~$P$ obtained by the
Realizations Lemma (Lemma~\ref{lem:balwsndc}).

We will prove the result by satisfying $\uid$ violations in~$J$
with fact-thrifty chase steps
using the piecewise realization~$\pire$,
yielding a finite aligned superinstance $J_\f = (I_\f, \cov_\f)$
such that $I \subseteq I_\f$,
the restriction of~$\cov_\f$ to $I$ is $\cov$, $J_\f$
satisfies $\ids$, and
$I_\f$ follows $\pire$.
The process is a variant
of Lemma ``Using realizations to get completions''
(Lemma~\ref{lem:sreal}).

We call $J' = (I', \cov')$ the current state of our superinstance, starting at $J' \defeq J$. We will perform
fact-thrifty chase steps on $J'$.
We call $\calF$ the set of all fresh elements (not in $\dom(P)$) that we will introduce
(only in outer classes)
during the chase steps.
It is immediate that our construction will maintain the following:

\begin{description}
  \item[\invar{fsat}:~] $J'$ is a fact-saturated aligned superinstance of~$I_0$
    (this uses Lemma~\ref{lem:ftok});
  \item[\invar{sub}:~] $I \subseteq I'$;
  \item[\invar{sim}:~] $\restr{\cov'}{\dom(I)} = \cov$.
\end{description}

\noindent Further, we will additionally maintain the following invariants:

\begin{description}
  \item[\invar{fw}:~] $I'$ follows $\pire$;
  \item[\invar{krev}:~] $J'$ is $k$-reversible;
  \item[\invar{out}:~] elements of outer classes are only in $\calF$ or in
    $\dom(I)$.
\end{description}

\medskip

We now describe formally how we apply each fact-thrifty chase step. Choose an
element $a \in
\appelem{J'}{\tau}$ to which some $\uid$ $\tau:
\ui{R^p}{S^q}$ is applicable. Let $F_\a = R(\mybf{a})$ be the active fact,
with $a = a_p$.
The $\uid$ $\tau$ witnesses that the $\eqfun$-classes~$\Pi_i$ and~$\Pi_{i'}$,
of~$R^p$ and~$S^q$ respectively, are inner,
so by invariant \invar{fw} we have $a \in \pi_{R^p}(\pire)$.
As $\pire$ is $\ids$-compliant, we must have $a \in \pi_{S^q}(\pire)$,
and there is a $\card{\Pi_{i'}}$-tuple $\mybf{t} \in K_{i'}$
such that $t_q = a$.

We choose a fact $F_\r = S(\mybf{c})$ of~$J$ that achieves the fact class of
the chase witness~$F_\w$ (this is possible by invariant \invar{fsat}), and
create a new fact $F_\n = S(\mybf{b})$ with the fact-thrifty chase step
defined as follows:
\begin{samepage}
\begin{itemize}
  \item For the exported position $S^q$, we set $b_q \defeq a_p$.
  \item For any $S^r \in \Pi_{i'}$, noting that necessarily $S^r \in
    \danger(S^q)$, we set $b_r \defeq t_r$.
  \item For any position $S^r \in \danger(S^q) \backslash \Pi_{i'}$,
    we take $b_r$ to be a fresh element from~$\calF$.
  \item For any position $S^r \in \nondanger(S^q)$, we set $b_r \defeq c_r$.
\end{itemize}
\end{samepage}

\medskip

We must verify that this satisfies the conditions of thrifty chase steps. The
fact that $b_r \in \pi_{S^r}(J')$ for $S^r \in \nondanger(S^q)$ is immediate by
definition of $F_\r$. We now show the two other points.

First, we show that $b_r \notin \pi_{S^r}(J')$ for $S^r \in \danger(S^q)$.
Obviously this needs only to be checked for $S^r \in \Pi_{i'}$ (as the
other~$b_r$ are always fresh).
Assume to the contrary that $t_r \in \pi_{S^r}(J')$,
and let $F = S(\mybf{d})$ be a
witnessing fact.
As $\Pi_{i'}$ is inner,
by invariant
\invar{fw}, we deduce that
$\pi_{\Pi_{i'}}(\mybf{d}) \in \pi_{\Pi_{i'}}(\pire)$.
Now, as $d_r = t_r$ and $\pire$ is $\ufds$-compliant, we deduce that
$\mybf{d} = \mybf{t}$, so that $F$ witnesses that
$d_q$ is in $\pi_{S^q}(J')$.
As we have $d_q = t_q = a$,
this contradicts the applicability of~$\tau$ to~$a$.
Hence, the claim is proven.

Second, we check that reused elements have the right $\cov$-image. This is the
case by definition of fact-thrifty chase steps for the non-dangerous positions,
so again we need only check this for elements at a position $S^r \in \Pi_{i'}$, and
only if they are not fresh. We start by showing that, for such $S^r$, we have $b_r
\in \appelem{J'}{S^r}$.

Indeed, we have $b_r = t_r$ which is in
$\pi_{S^r}(\pire)$, and we cannot have $\mybf{t} \in \pi_{\Pi_{i'}}(J')$,
as otherwise this would contradict the applicability of
$\tau$ to $a$; so in particular, by invariant \invar{sub},
we cannot have $\mybf{t} \in \pi_{\Pi_{i'}}(I)$.
Thus, by definition of a piecewise realization, we have
$t_r \in \appelem{P}{S^r}$. Recalling that we have $t_r \in \dom(J')$, we show
that this implies $t_r \in \appelem{J'}{S^r}$.
Recalling the definition of $t_r \in \appelem{P}{S^r}$,
we distinguish two subcases: (1.) $t_r \in \dom(J)$
and $t_r \in \appelem{J}{S^r}$, or (2.) $t_r
\in \calH$ and $S^r \in \mapb(t_r)$.

In the subcase (1.) $t_r \in \dom(J)$ and $t_r \in \appelem{J}{S^r}$,
we remember that in the first point
we showed that $t_r \notin \pi_{S^r}(J')$.
So we still have $t_r \in \appelem{J'}{S^r}$,
which is what we claimed.

In the subcase (2.) $t_r \in \calH$ and $S^r \in \mapb(t_r)$,
consider a fact $F'$ of~$J'$ witnessing $t_r \in \dom(J')$, where $t_r$ occurs at a
position $T^l$; let $\Pi_{i''}$ be the $\eqfun$-class of $T^l$.
As $t_r \in \calH$, by invariant \invar{out}, $\Pi_{i''}$ is inner,
so by invariant \invar{fw} there is a tuple $\mybf{t}'$ of~$K_{i''}$
such that $t'_l = t_r$.
Now, as $t_r \in \calH$,
by definition of piecewise realizations,
we have $T^l \in \mapb(t_r)$.
Hence, either the $\uid$ $\tau' : \ui{T^l}{S^r}$ is in~$\ids$
or we have $T^l = S^r$.
As $t_r \in \pi_{T^l}(J')$ and we have shown in the first point that
$t_r \notin \pi_{S^r}(J')$, we know that $T^l \neq S^r$, so $\tau'$ is in
$\ids$.
Hence, as $F'$ witnesses that $t_r \in \pi_{T^l}(J')$,
and as $t_r \notin \pi_{S^r}(J')$,
we have $t_r \in \appelem{J'}{S^r}$, as we claimed.

Hence, we know that $b_r = t_r$ is in~$\appelem{J'}{S^r}$ in either subcase.
By invariant \invar{krev},
this implies that $\cov(b_r)$ is a $k$-reversible element of
$\chase{I_0}{\ids}$ introduced at a position of~$\eqidsclass{S^r}$. By
Lemma~\ref{lem:newkrev}, we know that the $\cov$-image $b'_r$ of a fresh element at
position $S^r$ would be $k$-reversible and
introduced at position $S^r$. Hence, by the Chase Locality Theorem
(Theorem~\ref{thm:locality}), we have $\cov(b_r) \bbsim_k b'_r$, so the
condition is satisfied. This proves that, indeed,
we can perform the fact-thrifty chase step that we described.

\medskip

We now check that the invariants are preserved. We first observe that for any
$S^r \in \danger(S^q) \backslash \Pi_{i'}$, the $\eqfun$-class of
$S^r$ is outer. Indeed, if $S^r$ occurred in $\ids$,
as $S^q$ does because of~$\tau$, we know by
assumption \assm{reversible} that, as the $\ufd$ $S^r \rightarrow S^q$ is in
$\ufds$ by dangerousness of~$S^r$, the $\ufd$ $S^q \rightarrow S^r$ also should,
but then we would have $S^r \eqfun S^q$, so $S^r \in \Pi_{i'}$, a contradiction.
Hence, the $\eqfun$-class of~$S^r$ is indeed outer.

Now, invariant \invar{fw} is preserved because, by the above observation, the
new fact~$F_\n$
is defined on the inner classes either following $\mybf{t}$ or
following an existing fact of~$J'$.
Invariant \invar{krev} is preserved by Lemma~\ref{lem:newkrev} for the fresh
elements, or by \invar{krev} on the previous state~$J'$ for the existing elements.
Invariant \invar{out}
is preserved because the only elements of~$F_\n$ that are not in~$\calF$ or
in~$\dom(I)$ are those of $\Pi_{i'}$, which is inner. 
This shows that the invariant is preserved by the fact-thrifty chase step.

\medskip

We perform fact-thrifty chase steps until no violations of~$\ids$ remain: invariant
\invar{fw} guarantees that we terminate.
Indeed, $\pire$ is finite, the domain of the resulting instance is bounded by
that of~$\pire$ for all inner classes, and new elements created in outer classes
cannot create violations of~$\ids$ or cause the creation of further elements, by
definition of their class being outer. Hence, the result of the process is
finite, and it satisfies $\ids$ because no violations remain.
This concludes the proof.

\mysubsection{Proof of the Chase Locality Theorem
(Theorem~\ref{thm:locality})}
\label{apx:prf_locality}

We give an equivalent rephrasing
of the Chase Locality Theorem (Theorem~\ref{thm:locality})
using the notion of $n$-reversible elements (Definition~\ref{def:nreverse}):

\begin{restatable}[Chase locality theorem]{theorem}{locality2}
  \label{thm:locality2}
  For any instance~$I_0$, transitively closed set of $\uid$s $\ids$,
  and $n \in \mathbb{N}$,
  for any two
  elements $a$ and $b$
  respectively introduced at positions~$R^p$ and~$S^q$ in $\chase{I_0}{\ids}$
  such that $R^p \eqids S^q$,
  if $a$ and $b$ are $n$-reversible
  then $a \bbsim_n b$.
\end{restatable}

Note that this result is for an arbitrary set of~$\uid$s and $\fd$s,  not relying
on any finite closure properties, or on assumption \assm{reversible}. (It only
assumes that the last $n$ dependencies used to create~$a$ and~$b$ were
reversible.)
However we still assume that $\ids$ is transitively closed.

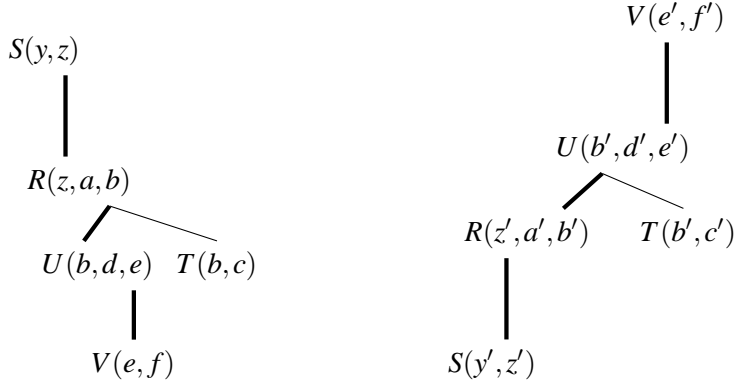
\begin{figure}
  \centering
  \begin{forest}
    [{$S(y, z)$}, parent anchor=-50, anchor=-50
      [{$R(z, a, b)$}, edge=ultra thick,
          child anchor=120, anchor=120, parent anchor=-40
        [{$T(b, c)$}, child anchor=90, anchor=-35,
        before computing xy={s=(70)}]
        [{$U(b, d, e)$}, edge=ultra thick, child anchor=120, parent anchor=-35,
          anchor=-35
          [{$V(e, f)$}, edge=ultra thick ]
        ]
      ]
    ]
  \end{forest}
  \hspace{2cm}
  \begin{forest}
    [{$V(e', f')$}, parent anchor=-110, anchor=-110
      [{$U(b', d', e')$}, edge=ultra thick,
        child anchor=30, anchor=30, parent anchor=-130
        [{$R(z', a', b')$}, edge=ultra thick, child anchor=35, anchor=-130, parent
          anchor=-130, before computing xy={s=(-60)}
          [{$S(y', z')$}, edge=ultra thick, child anchor=60, anchor=60 ]
        ]
        [{$T(b', c')$}, anchor=-130, child anchor=100,
          before computing xy={s=(0)}
        ]
      ]
    ];
  \end{forest}
  \caption{Chase locality example. Elements $b$ and $b'$ are $1$-reversible and
    introduced at positions $R^3$ and $U^1$. Reversible $\uid$s are represented
  by thick edges.}
  \label{fig:chase}
\end{figure}

Figure~\ref{fig:chase} illustrates the result in a simple situation. The
intuition is the following: $n$-reversible elements in
the chase have the same neighborhoods up to distance~$n$, no matter their exact
histories, as long as they were introduced in $\eqids$-equivalent positions:
intuitively, the facts that go ``downwards'' in the neighborhood of~$a$ in the forest structure can be
matched to facts in the neighborhood of~$b$
because they are required by $\ids$, and the facts ``upwards'' are also
matched up to distance $n$ because of the reverses of the $\uid$s
used along this chain.

To prove the theorem,
fix the instance $I_0$ and the set $\ids$ of~$\uid$s. 
We first show the following easy lemma:

\begin{lemma}
  \label{lem:samepos}
For any $n > 0$ and position $R^p$, for any two elements $a, b$ of
$\chase{I_0}{\ids}$ introduced at position $R^p$ in two facts $F_a$ and $F_b$,
letting $a'$ and $b'$ be the exported elements of~$F_a$ and $F_b$, if $a'
  \bbsim_{n-1} b'$, then $a \bbsim_n b$.
\end{lemma}

\begin{proof}
  By symmetry, it suffices to show that $a \bsim_n b$.
  We proceed by induction on~$n$.

  For the base case  $n=1$, 
  observe that, for every fact $F$ of
  $\chase{I_0}{\ids}$ where $a$ occurs at some position $S^q$,
  there are only two cases. Either $F = F_a$, so
  we can pick $F_b$ as the representative fact, or the $\uid$ $\ui{R^p}{S^q}$
  is in $\ids$ so we can pick a corresponding fact for $b$ by definition of
  the chase.

  For the induction step, we proceed in the same way. If $F = F_a$, we pick
  $F_b$ and use either the hypothesis on $a'$ and $b'$ or the induction
  hypothesis (for other elements of~$F_a$ and $F_b$) to justify that $F_b$ is a
  suitable witness. Otherwise, we pick the corresponding fact for $b$ which must
  exist by definition of the chase, and apply the induction hypothesis to the
  other elements of the fact to conclude.
\end{proof}

We now prove the Chase Locality Theorem.
Recall the definition of~$\eqrev$ (Definition~\ref{def:eqids}). However, note
that, as we no longer make assumption \assm{reversible},
while $\eqids$ is still an equivalence relation,
it is no longer the case that all $\uid$s of~$\ids$ are reflected in $\eqrev$:
the $\uid$ $\ui{R^p}{S^q}$ may be in~$\ids$ even though $R^p \not\eqids S^q$ if
$\ui{S^q}{R^p}$ is not in~$\ids$.

We prove by induction on~$n$ the main claim:
for any positions $R^p$ and $S^q$ such that $R^p \eqrev
S^q$, for any two $n$-reversible elements $a$ and $b$ respectively introduced at
positions $R^p$ and $S^q$, we have $a \bbsim_n b$.
By symmetry it suffices to show that
$(\chase{I_0}{\ids}, a) \bsim_n (\chase{I_0}{\ids}, b)$. 

The base case of~$n = 0$ is immediate.

For the induction step, fix $n > 0$, and assume that the result holds
for $n-1$. Fix $R^p$ and $S^q$, and let $a,
b$ be two $n$-reversible elements introduced respectively at $R^p$ and $S^q$ in
facts $F_a$ and $F_b$.
Note that by the induction hypothesis we already
know that $(\chase{I_0}{\ids}, a) \bsim_{n-1} (\chase{I_0}{\ids}, b)$; we must
show that this holds for~$n$.

First, observe that, as~$a$ and~$b$ are $n$-reversible with $n > 0$, they are not
elements of~$I_0$. Hence, by definition of the chase, for each one of them, the
following is true: for each fact of the chase where the element occurs, it only
occurs at one position, and all other elements co-occurring with it in a fact of the chase
occur only at one position in only one of these facts.
Thus, to prove the claim,
it suffices to construct a mapping $\phi$ from the set $N_1(a)$ of the
facts of~$\chase{I_0}{\ids}$ where $a$ occurs, to the set $N_1(b)$ of the facts
where $b$ occurs, such that the following holds:
for every fact $F = T(\mybf{a})$ of~$N_1(a)$, letting $T^c$ be the position of
$F$ such that $a_c = a$ (there is only one such position by construction of the
chase), $b$ occurs at position $T^c$ in $\phi(F) = T(\mybf{b})$, and for every
$i$, $a_i \bsim_{n-1} b_i$.

By construction of the chase (using the
Unique Witness Property), $N_1(a)$ consists of exactly the following facts:

\begin{itemize}
  \item The fact $F_a = R(\mybf{a})$, where $a_d = a'$ is the
    exported element (for a certain $d$),
    $a_p = a$ was introduced at $R^p$ in $F_a$,
    and for $i \notin \{p, d\}$, $a_i$ was
    introduced at $R^i$ in $F_a$
  \item For every $\uid$ $\tau : \ui{R^p}{V^g}$ of~$\ids$, a $V$-fact $F^{\tau}_a$ where
    all elements were introduced in this fact except the one at position $V^g$
    which is~$a$.
\end{itemize}

A similar characterization holds for $b$, with the analogous notation.
We construct the mapping $\phi$ as follows:

\begin{itemize}
  \item If $R^p = S^q$ then set $\phi(F_a) = F_b$; otherwise, as $\tau : \ui{S^q}{R^p}$
    is in $\ids$, set $\phi(F_a)$ to be the fact $F^{\tau}_b$.
  \item For every $\uid$ $\tau : \ui{R^p}{V^g}$ of~$\ids$, as $R^p \eqids S^q$,
    by transitivity,
    either $S^q = V^g$ or the $\uid$ $\tau': \ui{S^q}{V^g}$ is in $\ids$. 
    In the first case, set $\phi(F^{\tau}_a) = F_b$.
    In the second case, set $\phi(F^{\tau}_a) = F^{\tau'}_b$.
\end{itemize}

We must now show that $\phi$ satisfies the required conditions. First,
verify that indeed, by construction, whenever $a$ occurs at position $T^c$ in $F$ then $b$ occurs
at position $T^c$ in $\phi(F)$. Second, fix $F \in N_1(a)$, write $F =
T(\mybf{a})$ and $\phi(F) = T(\mybf{b})$, with $a_c = a$ and $b_c = b$ for
some $c$, and show that $a_i \bsim_{n-1} b_i$ for all $T^i \in \pos(T)$.
If $n = 1$ there is nothing to show and we are done, so we assume $n \geq 2$.
If $i = c$ then the claim is immediate by the induction hypothesis; otherwise, we distinguish two cases:

\begin{enumerate}
  \item  $F = F_a$ (so that $T = R$ and $c = p$), or $F = F_a^{\tau}$ such that the $\uid$
    $\tau : \ui{R^p}{T^c}$ is
    reversible. In this case, by construction,
    either $\phi(F) = F_b$ or $\phi(F) = F^{\tau'}_b$ for $\tau' :
    \ui{S^q}{T^c}$; $\tau'$ is then reversible,
    because $R^p \eqids S^q$ and $R^p \eqids T^c$.

    We show that for all $1 \leq i \leq \arity{T}$, $i \neq c$, $a_i$ is
    $(n-1)$-reversible and was introduced in~$\chase{I_0}{\ids}$ at a position
    in the $\eqids$-class of~$T^i$. Once we have proved this, by symmetry we can
    show the same for all~$b_i$, so that we can conclude that $a_i \bsim_{n-1}
    b_i$ by induction hypothesis.
    To see why the claim holds, we distinguish two subcases.
    Either $a_i$ was introduced in $F$, or we have $F = F_a$, $i = d$
    and $a_i$ is the exported element for $a$.

    In the first subcase, $a_i$ was created by applying the reversible $\uid$
    $\tau$ and the exported element~$a$ is $n$-reversible, so $a_i$ is
    $(n-1)$-reversible (in fact it is $(n+1)$-reversible), and is introduced
    at position~$T^i$. In the second subcase, $a_i$ is the exported element used to
    create~$a$, which is $n$-reversible, so $a_i$ is $(n-1)$-reversible; and as
    $n \geq 2$, the last dependency applied to create $a_i$ is reversible, so
    that $a_i$ was introduced at a position in the same $\eqids$-class as~$T^i$.
    Hence, we have proved the desired claim in the first case.

  \item $F = F_a^{\tau}$ such that $\tau : \ui{R^p}{T^c}$ is not reversible. In this
    case, we cannot have $T^c = S^q$ (because we have $R^p \eqids S^q$),
    so that $\phi(F) = F_b^{\tau}$,
    and all $a_i$ for $i \neq c$ were introduced in~$F$ at position $T^i$,
    and likewise for the~$b_i$ in~$\phi(F)$. 
    Using Lemma~\ref{lem:samepos}, as $a \bbsim_{n-1} b$, we conclude
    that $a_i \bbsim_{n} b_i$, hence $a_i \bsim_{n-1} b$.
\end{enumerate}

This concludes the proof.

\section{Proofs for Section~\ref{sec:manyscc}: Arbitrary $\uid$s:
Lifting Assumption Reversible}
\label{apx:manyscc}

This appendix proves the claims needed to complete our proof of
Theorem~\ref{thm:myuniv}, the existence of universal instances for $\uid$s,
$\ufd$s, and acyclic
$\cq$s of fixed size. The main claim is the existence of manageable
partitions (Lemma~\ref{lem:nontrivconv}).

Remember that we are assuming the
``Unique Witness Property'' (Section~\ref{sec:prelim})
and that the constraints $\ucon$
are closed under the finite closure rule (in particular, $\ids$ is transitively
closed).

\mysubsection{Finite closure computation algorithm}
\label{apx:cosm}

For convenience we recall here how the finite closure is computed, from~\cite{cosm}.

Given a set $\con = \fds \sqcup \ids$ of~$\fd$s and $\uid$s,
an \defp{$\incd$ path} of~$\con$ is a
sequence of~$\uid$s of~$\ids$ of the following form:
$\ui{R_1^{i_1}}{R_2^{j_2}},\allowbreak \ui{R_2^{i_2}}{R_3^{j_3}}, \ldots,\allowbreak
\ui{R_{n-1}^{i_{n-1}}}{R_n^{j_n}}$, with $i_k \neq j_k$ for all $k$. The path is
\defp{functional} if,
for all $1 < k < n$,
$R_k^{i_k} \rightarrow R_k^{j_k} \in \fds$.
Note that our definition of the $\idprec$ relation ensures that $\tau \idprec
\tau'$ iff $\tau, \tau'$ is a functional $\incd$ path.

An \defp{invertible cycle} $C$ of~$\con$ is a functional $\incd$ path with
$R_{n} = R_1$ and $j_{n} = j_1$ (so that $R_1^{i_1}
\rightarrow R_1^{j_1} \in \fds$): a $\uid$ that occurs in an invertible cycle is
said to be \defp{invertible}. The \defp{reverse}
$\overline{C}$ of an invertible cycle $C$ is
$\ui{R_n^{j_n}}{R_{n-1}^{i_{n-1}}}, \ldots, \ui{R_2^{j_2}}{R_1^{i_1}}$.

Applying the \defp{cycle closure rule} in $\con$ means taking every invertible
cycle~$C$ of~$\con$ and adding to $\con$ the $\uid$s and $\ufd$s
needed to make $\overline{C}$ an invertible cycle in $\con$, namely,
$\ui{R_2^{j_2}}{R_1^{i_1}},\allowbreak \ui{R_3^{j_3}}{R_2^{i_2}}, \ldots,\allowbreak
\ui{R_n^{j_n}}{R_{n-1}^{i_{n-1}}}$, 
and $R_k^{j_k} \rightarrow R_k^{i_k}$ for $1 \leq k \leq n$.
The finite closure  is computed by closing under the rule above
and by implication of the $\uid$s and of the
$\fd$s in isolation.

The fact that the result is exactly the finite closure of~$\con$ is shown
in~\cite{cosm}.

\mysubsection{Proof of Lemma~\ref{lem:sccbyscc} (New violations follow
$\idprec$)}
\label{apx:prf_sccbyscc}
\sccbyscc*

Fix $J$, $J'$ and $\tau : \ui{R^p}{S^q}$ and $\tau'$. As chase steps add a
single fact, the only new $\uid$ violations in $J'$ relative to $I$ are on
elements in the newly created fact $F_\n = S(\mybf{b})$, As $\ids$ is
transitively closed, $F_\n$ can introduce no new violation on the exported element
$b_q$. Now, as thrifty chase steps
always reuse existing elements at
non-dangerous positions, we know that if $S^r \in \nondanger(S^q)$ then
no new $\uid$ can be applicable to $b_r$. Hence, if a new $\uid$ is applicable
to $b_r$ for $S^r \in \pos(S)$, then necessarily $S^r \in \danger(S^q)$. By
definition of dangerous positions, the $\ufd$ $S^r \rightarrow S^q$ is in
$\ufds$, and it is non-trivial because $S^r \neq S^q$.
Hence, writing $\tau' : \ui{S^r}{T^r}$, we see that $\tau \idprec \tau'$.

\mysubsection{Proof of Corollary~\ref{cor:trivscc} (Dealing with
trivial classes)}
\label{apx:prf_trivscc}
\trivscc*

Fix $J$, $J'$ and $\tau$. All violations of~$\tau$ in $J$ have been satisfied in
$J'$ by definition of~$J'$, so we only have to show that no new violations of
$\tau$ were introduced in $J'$. But by
Lemma~\ref{lem:sccbyscc}, as $\tau \not\idprec \tau$, each fresh fact-thrifty chase
step cannot introduce such a violation, hence there is no new violation of
$\tau$ in $J'$. Hence, $J' \models \tau$.

\mysubsection{Proof of Lemma~\ref{lem:nontrivconv} (Existence of manageable partitions)}
\label{apx:prf_nontrivconv}

Our goal in this section is to show:

\nontrivconv*

We assume that $\ids$ is closed under the finite closure rule (see
Appendix~\ref{apx:cosm}). Hence, in particular, it is transitively closed.

We start by introducing definitions about the $\idprec$ relation,
which we recall is defined so that $\tau \idprec \tau'$ for $\tau, \tau' \in
\ids$ whenever $\tau, \tau'$ is a functional $\incd$ path, namely: letting $\tau
: \ui{R^p}{S^q}$ and $\tau' : \ui{S^r}{T^u}$, the $\ufd$ $S^r \rightarrow S^q$
is non-trivial and is in $\ufds$.

We extend $\idprec$ to sets of~$\uid$s in the expected way: $P \idprec P'$ if
there exists $\tau \in P$, $\tau' \in P'$ such that $\tau \idprec \tau'$.

\begin{definition}
  The \deft{$\incd$ graph} $\idgraph(\ids)$ is the directed graph (with
  self-loops) defined on $\ids$ by the $\idprec$ relation. We define the
  \deft{strongly connected components} of~$\idgraph(\ids)$ as usual: an SCC is a
  maximal subset $P$ of~$\ids$ such that for all $\tau, \tau' \in P$, we
  have $\tau \idprec^* \tau'$, where $\idprec^*$ denotes the transitive and
  reflexive closure of the $\idprec$ relation. The \deft{SCC graph}
  $\sccgraph(\ids)$ is the
  directed acyclic graph (without self-loops) defined on the SCCs of
  $\idgraph(\ids)$ such that, for any two SCCs $P \neq P'$ of~$\idgraph(\ids)$,
  there is an edge from $P$ to $P'$ iff $P \idprec P'$.

  Note that the definition of SCCs allows both singleton SCCs $\{\tau\}$ where
  we have a self-loop ($\tau \idprec \tau$), and singletons where there is none
  ($\tau \not\idprec \tau$). We say that an SCC is \deft{trivial} if it is a
  singleton without self-loops.
  Otherwise, if the SCC is not a singleton or if it has a self-loop,
  we call it \deft{non-trivial}.
\end{definition}

We first show the following lemma to understand the structure of
the SCCs of $\idgraph(\ids)$.
This lemma is proved in Appendix~\ref{apx:prf_sccs}.

\begin{restatable}[SCC structure]{lemma}{sccs}
  \label{lem:sccs}
  The SCCs of $\idgraph(\ids)$ are transitively closed sets of $\uid$s. Further,
  for any non-trivial SCC $P$, letting $P^{-1} \defeq \{\tau^{-1} \mid \tau \in
  P\}$, all $\uid$s of $P^{-1}$ are in $\ids$, and $P^{-1}$ is an SCC
  of~$\idgraph(\ids)$.
\end{restatable}

Note that $P$ and $P^{-1}$, as SCCs of $\idgraph(\ids)$, may be equal or
disjoint. We accordingly call \deft{self-inverse} an SCC $P$ that is non-trivial
but satisfies $P = P^{-1}$; non-trivial SCCs such that $P$ and $P^{-1}$ are
disjoint are called \deft{non-self-inverse}.

Given the structure of the SCCs, the first step to construct a manageable
partition
is to construct a topological
sort of the SCC graph $\sccgraph(\ids)$ of~$\idgraph(\ids)$, but with an
additional property, motivated by what we showed in Lemma~\ref{lem:sccs}:

\begin{definition}
  A topological sort of~$\sccgraph(\ids)$ is \deft{inverse-sequential} if, for
  any non-self-inverse SCC $P$, the SCCs $P$ and $P^{-1}$ are enumerated consecutively.
\end{definition}

The first result, proven in Appendix~\ref{apx:prf_sortexists},
is to justify that we can indeed construct an inverse-sequential topological
sort of the SCC graph of $\idgraph(\ids)$:

\begin{restatable}[Inverse-sequential topological sort]{proposition}{sortexists}
  \label{prp:sortexists}
  For any conjunction $\ids$ of $\uid$s closed under finite implication, 
  $\sccgraph(\ids)$ has an inverse-sequential topological sort.
\end{restatable}

The second step is to construct the manageable partition itself from the
inverse-sequential topological sort. Here is how we define the ordered partition
from the topological sort:

\begin{definition}
  An inverse-sequential topological sort defines an ordered partition $(P_1,
  \ldots, P_{\neqidsc})$ of~$\ids$, in the following way: each class $P_i$ of
  the partition either corresponds to one SCC of~$\sccgraph(\ids)$ (which is
  either trivial or self-inverse), or to the union of an SCC and its inverse SCC
  (which were enumerated consecutively because the topological
  sort is inverse-sequential).
  It is immediate that $(P_1, \ldots, P_{\neqidsc})$ is indeed an ordered
  partition, as it is constructed from a topological sort by merging some
  classes that were enumerated consecutively.
\end{definition}

The second result
is to show that the resulting ordered partition
is indeed a manageable partition. In other words, we must
show that the classes of the partitions are either trivial, or that they are a
set of $\uid$ that is transitively closed and satisfies assumption
\assm{reversible}.

\begin{restatable}[Manageable partitions from sorts]{proposition}{sorttopart}
  \label{prp:sorttopart}
  For any conjunction $\ids$ of $\uid$s closed under finite implication, letting
  $\mybf{P}$ be an ordered partition obtained from an inverse-sequential topological
  sort of $\sccgraph(\ids)$, $\mybf{P}$ is a manageable partition.
\end{restatable}
 
This second result is proven in Appendix~\ref{apx:prf_sorttopart} and
concludes the proof of our original claim.

\mysubsection{Proof of the SCC Structure Lemma (Lemma~\ref{lem:sccs})}
\label{apx:prf_sccs}

\sccs*

We first show an general lemma:

\begin{lemma}
  \label{lem:cycleexists}
  Let $P$ be a non-trivial SCC of~$\idgraph(\ids)$.
  For any $\tau, \tau' \in P$, there is an invertible cycle of~$\uid$s
  of~$P$ in which $\tau$ and $\tau'$ occur. 
\end{lemma}

\begin{proof}
  Because $P$ is a non-trivial SCC, we have $\tau \idprec^* \tau'$ and
  $\tau' \idprec^* \tau$, and the desired invertible cycle is obtained by
  concatenating the functional $\incd$ paths from $\tau$ to $\tau'$, and from
  $\tau'$ to $\tau$. Because $P$ is an SCC, it is immediate that the $\uid$s
  of the resulting path are all in $P$.
\end{proof}

We then divide our claim in two lemmas:

\begin{lemma}
  \label{lem:transok}
  Let $P$ be an SCC of~$\idgraph(\ids)$. Then $P$ is closed under the transitivity
  rule.
\end{lemma}

\begin{proof}
  Let $P$ be an SCC.
  If $P$ consists of a single $\uid$, then transitivity is immediately
  respected, so we assume that $P$ contains $>1$ $\uid$s. In particular, $P$ is
  non-trivial.
  Let $\tau: \ui{R^p}{S^q}$ and $\tau': \ui{S^q}{T^r}$ be
  two $\uid$s of~$P$ with $R^p \neq T^r$. As $\ids$ is closed under transitivity, we know
  $\tau'' : \ui{R^p}{T^r}$ is in $\ids$. We show that $\tau'' \in P$.

  As $P$ is a non-trivial SCC, there is a functional
  $\incd$ path $\tau' = \tau_1 \idprec \cdots \idprec \tau_n
  = \tau$, where $\tau_i \in P$ for all $1 \leq i \leq n$. Because of the $\ufd$s that must
  be in $\ufds$ to make it a functional $\incd$ path, it is immediate that
  the following two paths are functional $\incd$ paths as well:
  $\tau'' \idprec \tau_2 \idprec \cdots \idprec \tau_n$ and $\tau_1 \idprec
  \cdots \idprec \tau_{m-1} \idprec
  \tau''$.
Thus we have $\tau''
  \idprec^* \tau$, and $\tau' \idprec^* \tau''$ where $\tau, \tau' \in P$, so
  that $\tau'' \in P$ by definition of an SCC.
\end{proof}

\begin{lemma}
  \label{lem:invscc}
  Let $P$ be a non-trivial SCC of~$\idgraph(\ids)$, and let $P^{-1}
  \defeq \{\tau^{-1} \mid \tau \in P\}$. Then $P^{-1} \subseteq \ids$, and
  $P^{-1}$ is an SCC of~$\idgraph(\ids)$.
\end{lemma}

\begin{proof}
  We first prove that, for any $\tau \in P$, $\tau^{-1} \in \ids$. This is a
  direct consequence of Lemma~\ref{lem:cycleexists}: there is an invertible
  cycle of~$P$ containing $\tau$, so that by definition of an invertible cycle,
  $\tau^{-1}$ is in $\ids$. We now turn to the second part of the claim.
  
  First, we show that for any two
  $\tau, \tau' \in P^{-1}$, there is a functional $\incd$ path from $\tau$ to
  $\tau'$, so that $P^{-1}$ is strongly connected. This is clear: by
  Lemma~\ref{lem:cycleexists}, there exists an invertible cycle $C$ of~$P$
  containing $\tau^{-1}$ and $(\tau')^{-1} \in P$, and the reverse $\overline{C}$ of
  this cycle is also an invertible cycle, because $\ucon$ is finitely closed;
  $\overline{C}$ is then 
  a cycle of~$\uid$s of~$P^{-1}$ containing $\tau$ and $\tau'$.
  
  Second, we show that for any
  $\uid$ $\tau \in \ids$, if $P^{-1} \idprec^* \tau$ and $\tau \idprec^* P^{-1}$
  then $\tau \in P^{-1}$.
  Consider such a $\uid$ $\tau$,
  and let $p_1: \tau' =
  \tau'_1 \idprec \cdots \idprec \tau'_n = \tau$
  and $p_2 : \tau = \tau''_1 \idprec \cdots \idprec \tau''_m =
  \tau''$ be the witnessing functional $\incd$ paths,
  with $\tau', \tau'' \in P^{-1}$. We showed in the previous paragraph that $P^{-1}$ is
  strongly connected: consider a (possibly empty) functional $\incd$ path $p_3$
  from $\tau''$ to $\tau'$ witnessing the fact that $\tau'' \idprec^* \tau$.
  Concatenating $p_1$, $p_2$ and $p_3$ yields an invertible cycle $C$, so that
  because $\ucon$ is finitely closed, its reverse $\overline{C}$ is also an invertible
  cycle. But $\overline{C}$ witnesses the fact that $(\tau'')^{-1} \idprec^*
  \tau^{-1}$ and $\tau^{-1} \idprec^* (\tau')^{-1}$. Now, as $(\tau')^{-1},
  (\tau'')^{-1} \in P$ and $P$ is an SCC, we have $\tau^{-1} \in P$, so that
  $\tau \in P^{-1}$, the desired claim. Hence, $P^{-1}$ is both strongly
  connected and maximal, so it is an SCC.
\end{proof}

This concludes the proof.
Note that, as $P$ and $P^{-1}$ are both SCCs of~$\idgraph(\ids)$, either they
are equal or they are disjoint. We observe that both cases may occur:

\begin{example}
  \label{exa:sccs}
  Consider the $\uid$s $\tau: \ui{R^2}{S^2}$ and $\tau': \ui{S^1}{R^1}$, and the
  $\ufd$s $\phi: R^2 \rightarrow R^1$ and $\phi': S^1 \rightarrow S^2$. $\tau,
  \tau'$ is an invertible cycle, so that by the finite closure rule, the $\uid$s
  $\tau^{-1}$ and $(\tau')^{-1}$ and the reverse $\ufd$s are implied. However in
  $\idgraph(\ids)$ we have $\tau \idprec \tau'$, $\tau' \idprec \tau$, 
  $\tau^{-1} \idprec (\tau')^{-1}$, $(\tau')^{-1} \idprec \tau^{-1}$, so that
  $\{\tau, \tau'\}$ and $\{\tau^{-1}, (\tau')^{-1}\}$ are two disjoint SCCs.

  Consider now the $\uid$s $\tau: \ui{R^2}{S^2}$, $\tau^{-1} : \ui{S^2}{R^2}$,
    $\tau': \ui{R^1}{R^3}$, $\tau'': \ui{S^3}{S^1}$, and the $\ufd$s $R^1
    \rightarrow R^2$, $R^2 \rightarrow R^3$, $S^3 \rightarrow S^2$ and $S^2
    \rightarrow S^1$. We can construct the invertible cycles $\tau'$ and
    $\tau''$, so that $(\tau')^{-1}$ and $(\tau'')^{-1}$ are implied
    by the finite closure rule.
    However, besides $\tau' \idprec \tau'$, $\tau'' \idprec
    \tau''$, $(\tau')^{-1} \idprec (\tau')^{-1}$, $(\tau'')^{-1} \idprec
    (\tau'')^{-1}$, it is also the case that $\tau \idprec \tau''$, $\tau''
    \idprec \tau^{-1}$, $\tau^{-1} \idprec \tau'$ and $\tau' \idprec \tau$, and
    using the reverse $\ufd$s the same is true of the inverses of~$\tau'$,
    $(\tau')^{-1}$, $\tau''$, and $(\tau'')^{-1}$. So in fact there is only one
    SCC $P = \{\tau, \tau^{-1}, \tau', (\tau')^{-1}, \tau'', (\tau'')^{-1}\}$,
    with $P^{-1} = P$.
\end{example}

\mysubsection{Proof of the Inverse-Sequential Topological Sort Proposition
(Proposition~\ref{prp:sortexists})}
\label{apx:prf_sortexists}

We now prove that $\sccgraph(\ids)$
has an inverse-sequential topological sort:

\sortexists*

For this we need the following observation about $\sccgraph(\ids)$:

\begin{lemma}
  \label{lem:goodorder}
  Let $P$ be a non-self-inverse SCC and consider $\tau \in
  \ids \backslash (P \cup P^{-1})$ such that $\tau \idprec P$. Then one of the following holds:
  \begin{itemize}
    \item we have $\tau \idprec P^{-1}$
    \item the SCC of~$\tau$ is trivial, and for any
      $\tau_\p \in \ids$ such that $\tau_p \idprec \tau$, we have
      $\tau_\p \idprec^* P^{-1}$.
  \end{itemize}
\end{lemma}

\begin{proof}
  Fix $\tau \in \ids \backslash (P \cup P^{-1})$ and assume
  that we have $\tau \idprec P$, i.e., $\tau \idprec \tau'$ for some $\tau' \in P$.
  As $P$ is non-trivial,
  using Lemma~\ref{lem:cycleexists}, consider the predecessor $\tau'_{n-1} $ of
  $\tau'$ in an invertible cycle containing $\tau'$ (possibly $\tau'_{n-1} =
  \tau'$). Let $R^p$ be the second
  position of~$\tau$, $R^q$ be the first position of~$\tau'$, and $R^r$ be the
  second position of~$\tau'_{n-1}$. Note that we have $R^r
  \neq R^q$ because $\tau'_{n-1} \idprec \tau'$, and $R^p
  \neq R^q$ because $\tau \idprec \tau'$.
  Observe that if $R^p \neq R^r$, then $\tau \idprec (\tau'_{n-1})^{-1}$
  because $R^r \rightarrow R^q$ and $R^q \rightarrow R^p$ hold in $\ufds$ (as
  these $\ufd$s are used in an invertible cycle) and
  $\ufds$ is closed under transitivity. This proves the claim, as taking $\tau''
  \defeq
  (\tau'_{n-1})^{-1} \in P^{-1}$, we have $\tau \idprec \tau''$.

  If $R^p = R^r$, let $P'$ be the SCC of~$\tau$. Assume first that $P'$ is
  non-trivial. In this case, by Lemma~\ref{lem:cycleexists}, there is an
  invertible cycle $\tau = \tau_1, \ldots, \tau_m
  = \tau$ in $P'$. But then, we have $\tau_{n-1}' \idprec \tau_2$, so that $P
  \idprec P'$, and as $P' \idprec P$ we have $P = P'$, so $\tau \in P$, a
  contradiction.
  
  Hence, $P'$ is trivial. Let $S^q$ be the first position of~$\tau$ and
  $T^u$ be the first position of~$\tau'_{n-1}$. We must have $S^q \neq T^u$, as
  otherwise we have $\tau = \tau'_{n-1}$ so $\tau \in P$, a contradiction.
  Hence, because $(\tau'_{n-1})^{-1}$ is in $\ids$ (as $\tau'_{n-1} \in P$), by
  transitivity $\tau'': \ui{S^q}{T^u}$ is in $\ids$.
  We can then see that $\tau'' \idprec (\tau'_{n-2})^{-1}$
  because we had $(\tau'_{n-1})^{-1} \idprec (\tau'_{n-2})^{-1}$ and both
  $\uid$s share the same second position; hence, $\tau'' \idprec
  P^{-1}$.
  Now as $\tau''$ and $\tau$
  have the same first position, for any $\tau_\p \in \ids$, clearly $\tau_p
  \idprec \tau$ implies that $\tau_\p \idprec \tau'' \idprec P^{-1}$, proving the last part of the claim.
\end{proof}

We now construct the inverse-sequential topological sort of~$\sccgraph(\ids)$ by enumerating
the SCCs in a certain way that respects the $\idprec$ relation and maintains the
following invariant: whenever $P$ is non-self-inverse,
then $P$ and $P^{-1}$ are enumerated consecutively; this guarantees that the
result is a topological sort and that it is inverse-sequential.

First, whenever trivial or self-inverse SCCs can be enumerated,
enumerate them. Second, whenever the SCCs that can be enumerated are
all non-self-inverse,
choose one such~$P$ to enumerate.
By the invariant, $P^{-1}$ has not yet been enumerated,
otherwise $P$ would have been enumerated immediately after.
We want to enumerate $P$, and then enumerate $P^{-1}$.

To see why this is doable, we must show that, assuming that $P \neq P^{-1}$, if
$P$ can be enumerated and no trivial or self-inverse SCCs can be enumerated, then $P^{-1}$ can
also be enumerated.
Let $P'$ be
a parent SCC of~$P^{-1}$ in $\sccgraph(\ids)$ (so that $P' \idprec P^{-1}$), and show that it
has been enumerated already.
If we have $P' = P$, meaning that $P \idprec P^{-1}$, then this is not a
problem, because we are about to enumerate first $P$ and then $P^{-1}$, so we
may assume that $P' \neq P$. Hence, $P'$ is different from~$P$ and~$P^{-1}$, so
it is disjoint from it. We apply Lemma~\ref{lem:goodorder} to any $\tau \in P'$.
In the first case, we also have $P' \idprec P$, so as $P$ can be enumerated,
$P'$ was enumerated already.
In the second case, $P' = \{\tau\}$ is trivial; further, considering any $P''
\idprec P'$, we have $P'' \idprec^*
P$, so $P''$ was enumerated already. Hence, all such $P''$ are already
enumerated, so that $P'$ can be enumerated, but as it is trivial, it must have been
enumerated already. Hence, in both cases $P'$ was already enumerated unless it
is~$P$.
This ensures that we can indeed enumerate $P$ and~$P^{-1}$ consecutively,
maintaining our invariant.
Thus, we have constructed an inverse-sequential topological sort
of~$\sccgraph(\ids)$. This concludes the proof.

\mysubsection{Proof of the Manageable Partitions From Sorts Proposition
(Proposition~\ref{prp:sorttopart})}
\label{apx:prf_sorttopart}

\sorttopart*

Let $(P_1, \ldots, P_{\neqidsc})$ be the ordered partition. We 
prove that it is manageable.
Trivial SCCs are indeed trivial classes of the partition, so we must only
justify that any other class $P_i$ is transitively closed and satisfies
assumption \assm{reversible}.

We define $\positions(P)$ for $P$ a set of $\uid$s as the set of positions
occurring in~$P$, as in the definition of assumption \assm{reversible}. We first
prove a general lemma to take care of the second part of the assumption:

\begin{lemma}
  \label{lem:fdrev}
  Let $P$ be a non-trivial SCC of~$\idgraph(\ids)$.
  For any two positions $R^i \neq R^j$ of $\positions(P)$, if $R^i \rightarrow
  R^j$ is in $\ufds$ then so is $R^j \rightarrow R^i$.
\end{lemma}

\begin{proof}
  Fix $R^i$ and $R^j$, assume that $\phi : R^i \rightarrow R^j$ is in
  $\ufds$, and show that $\phi^{-1} : R^j \rightarrow R^i$ also is.
  Let $\tau_i$ be a $\uid$ of~$P$ where $R^i$ occurs, and $\tau_j$ be a $\uid$
  of~$P$ where $R^j$ occurs. By Lemma~\ref{lem:cycleexists}, there exists an
  invertible cycle $C_1$ where $R^i$ and $R^j$ occur.
 
  We write $C_1 = \ui{R_1^{i_1}}{R_2^{j_2}}, \ldots, \ui{R_n^{i_n}}{R_1^{j_1}}$,
  with some $1 \leq p, q \leq n$ such that $R_p =
  R_q = R$, and either $i_p = i$ or $j_p = i$, and either $i_q = j$ or $j_q =
  j$. By definition of an invertible cycle, the $\ufd$s $\phi_p : R^{i_p} \rightarrow
  R^{j_p}$, $\phi_p^{-1} : R^{j_p} \rightarrow R^{i_p}$, $\phi_q : R^{i_q} \rightarrow R^{j_q}$ and
  $\phi_q^{-1} : R^{j_q} \rightarrow R^{i_q}$ are in $\ufds$. Thus, because $\ufds$ is closed under
  transitivity, it is clear that if two positions among $S = (R^{j_p}, R^{i_p},
  R^{j_q}, R^{i_q})$ are equal (in particular, if $p = q$), then we have $R^x
  \eqfun R^y$ for any two positions $R^x$, $R^y$ in~$S$. Hence, as we know that $R^i
  \neq R^j$, and $R^i$ and $R^j$ are in~$S$, 
  the only case where we cannot conclude is the one where
  all the positions of~$S$ are different.

  If all positions of~$S$ are different, then, because of~$\phi_p$, $\phi_q$,
  $\phi_p^{-1}$ and $\phi_q^{-1}$, by transitivity of~$\ufds$, we know that for
  any $x_1, x_2
  \in \{i_p, j_p\}$, $y_1, y_2 \in \{i_q, j_q\}$, the $\ufd$ $R^{x_1} \rightarrow
  R^{y_1}$ is in $\ufds$ iff the $\ufd$ $R^{x_2} \rightarrow R^{y_2}$ is. Hence,
  since $\phi$ is in $\ufds$, as $i \in \{i_p, j_p\}$
  and $j \in \{i_q, j_q\}$, we know that $R^x \rightarrow R^y$ is in $\ufds$
  for \emph{all} $x \in \{i_p, j_p\}, y \in \{i_q, j_q\}$, and,
  to prove that $\phi^{-1}$ is in~$\ufds$,
  it suffices to show that $R^y \rightarrow R^x$ is in
  $\ufds$ for \emph{some} $x \in \{i_p, j_p\}, y \in \{i_q, j_q\}$.

  So let us construct the cycle $C_2 = \ui{R_1^{i_1}}{R_2^{j_2}}, \ldots,
  \ui{R_{q-1}^{i_{q-1}}}{R_q^{j_q}}, \ui{R_p^{i_p}}{R_{p+1}^{j_{p+1}}}, \allowbreak \ldots,
  \ui{R_n^{i_n}}{R_1^{j_1}}$. This is an invertible cycle, because $R_q = R_p = R$, and
  $R^{i_p} \neq R^{j_q}$ and the $\fd$ $R^{i_p} \rightarrow R^{j_q}$ is in
  $\ufds$ by
  our assumption. Hence, as $C_2$ is an invertible cycle, and because $\ucon$ is
  finitely closed, the reverse $\fd$ $R^{j_q} \rightarrow R^{i_p}$ is in
  $\ufds$, which implies that $\phi^{-1}$ is in $\ufds$.
\end{proof}

We then show a lemma to help justify that the classes are transitively closed:

\begin{lemma}
  \label{lem:notrans}
  For any non-trivial SCC $P$, if there is $\tau \in P$ and $\tau' \in P^{-1}$
  such that $\tau^{-1} \neq \tau'$ but the second position of~$\tau$ is the
  first position of~$\tau'$, then $P = P^{-1}$.
\end{lemma}

\begin{proof}
  We first observe that we have $P \idprec^* P^{-1}$. Indeed, as $P$ and
  $P^{-1}$ are non-trivial, consider $\tau_0 \in P$
  and $\tau_0' \in P^{-1}$ such that $\tau_0 \idprec \tau$
  and $\tau' \idprec \tau_0'$.
  Letting $\tau''$ be the $\uid$ which is
  transitively implied by $\tau$ and $\tau'$, we know that it must be in
  $\ids$ as it is transitively closed, and we observe that $\tau_0 \idprec \tau''
  \idprec \tau_0'$, so that $P \idprec^* P^{-1}$.
  
  Now, write $\tau: \ui{R^p}{S^q}$ and $\tau': \ui{S^q}{T^r}$, with $R^p \neq
  T^r$. As $P$ and $P^{-1}$ are non-trivial, using Lemma~\ref{lem:cycleexists},
  we can consider a functional $\incd$ path $\tau = \tau_1 \idprec \tau_2
  \idprec \cdots \idprec \tau_n =
  (\tau')^{-1}$, and a functional $\incd$ path $\tau^{-1} = \tau'_1 \idprec\cdots
  \idprec \tau'_m = \tau'$.
  By Lemma~\ref{lem:fdrev}, all $\ufd$s along these paths are
  such that their reverses are also in $\ufds$.
  Consider now the smallest~$k \geq 2$ such that we have
  $\tau_k^{-1} \neq \tau'_{m-k+1}$; such a~$k$ must exist because we have $\tau_n =
  (\tau')^{-1}$ and $\tau'_1 = \tau^{-1}$, and we know that $\tau^{-1} \neq
  \tau'$. Consider $\tau'' \defeq \tau_k \in P$, and $\tau''' \defeq \tau'_{m-k+1}
  \in P^{-1}$, and let~$S^u$ and~$S^v$ be respectively the first position
  of~$\tau''$ and the second position of $\tau'''$: indeed it is easily observed
  that these positions must be in the same relation~$S$, as this is true for
  $\tau_2$ and $\tau'_{m-1}$ and is preserved for~$\tau''$ and~$\tau'''$ because we
  have $\tau_l^{-1} = \tau'_{m-l+1}$ for all $1 \leq 2 \leq k$.

  We now distinguish two cases. The first case is $S^v \neq S^u$, and we then
  have $\tau''' \idprec \tau''$, so that $P^{-1} \idprec P$.
  The second case is $S^v = S^u$.
  In this case, $\tau'''$ and $\tau''$ are two $\uid$s of~$P^{-1}$ and~$P$ such that
  $(\tau''')^{-1} \neq \tau''$ but the second position of~$\tau'''$ is the first
  position of~$\tau''$. Hence, applying the reasoning of the first paragraph to~$\tau''$ and~$\tau$, we deduce
  that $P^{-1} \idprec^* P$. In either case, as we observed initially that $P
  \idprec^* P^{-1}$, we conclude that $P = P^{-1}$, the desired claim.
\end{proof}

\begin{corollary}
  \label{cor:invtrans}
  For any non-trivial SCC $P$, $P \cup P^{-1}$ is transitively closed.
\end{corollary}

\begin{proof}
  By Lemma~\ref{lem:transok}, $P$ and~$P^{-1}$ are transitively closed. Hence, if no
  $\uid$s is transitively implied by one $\uid$ from $P$ and one from
  $P^{-1}$ (or one from~$P^{-1}$ and one from~$P$), then the claim is proven.
  Otherwise, by Lemma~\ref{lem:notrans}, we have $P = P^{-1}$, so we can
  conclude by applying Lemma~\ref{lem:transok} to $P = P \cup P^{-1}$.
\end{proof}

We now conclude the proof of Proposition~\ref{prp:sorttopart}. Let $P_i$ be a
class of the ordered partition $(P_1, \ldots, P_{\neqidsc})$. We must show that it is
either trivial or reversible. If it is not trivial, then we must show three
things:
\begin{itemize}
  \item $P_i$ is transitively closed
  \item For every $\tau \in P_i$, we have $\tau^{-1} \in P_i$.
  \item For every two positions $R^p, R^q \in \positions(P_i)$ such that $R^p
    \rightarrow R^q$ is in $\ufds$, $R^q \rightarrow R^p$ is also in
    $\ufds$.
\end{itemize}

For the first claim,  as $P_i$ is not
trivial, it is either a self-inverse SCC $P$ of~$\idgraph(\ids)$ (and the claim
follows by Lemma~\ref{lem:transok}) or it is a union $P
\cup P^{-1}$ where $P$ is a non-self-inverse SCC (and the claim follows by
Corollary~\ref{cor:invtrans}).
The second claim is immediate by construction. The third claim is what is shown by
Lemma~\ref{lem:fdrev}, noting that for any SCC~$P$ of~$\ids$, we have $\positions(P) =
\positions(P^{-1})$.
This concludes the proof of Proposition~\ref{prp:sorttopart}.

\section{Proofs for Section~\ref{sec:hfds}: Higher-Arity $\fd$s}
\label{apx:hfds}

In this section, we show what is needed to adapt
the Acyclic Unary Universal Models Theorem (Theorem~\ref{thm:myuniv})
to
produce aligned superinstances that satisfy the full set of constraints
$\con$ rather than just the unary subset
$\ucon$.

\mysubsection{Proof of the Sufficiently Envelope-Saturated Solutions Proposition
(Proposition~\ref{prp:preproc})}
\label{apx:prf_preproc}

We now prove the following result, which provides our way to construct the
initial instance on which we apply the completion process of the previous
sections:

\preproc*

We define the notation 
$\arity{\sigma} \defeq \max_{R \in \sigma} \arity{R}$, and
also define the following:

\begin{definition}
  \label{def:overlap}
  The \deft{overlap} $\ovl(F, F')$ between two facts $F = R(\mybf{a})$ and $F'
  = R(\mybf{b})$ of the same relation~$R$ in an instance~$I$ is the subset $O$ of~$\pos(R)$ such that $a_s = b_s$ iff $R^s
  \in O$. If $\card{O} > 0$, we say that $F$ and $F'$ \deft{overlap}.
\end{definition}

We also define the following, which are the $\fd$s used in the definition
of envelopes (Definition~\ref{def:envelope}):

\begin{definition}
  \label{def:fdproj}
  Given a set $\fds$ of $\fd$s on a relation~$R$ and $O \subseteq \pos(R)$,
  the \deft{$\fd$ projection} $\fdrestr{\fds}{O}$ of 
  $\fds$ to~$O$ are the $\fd$s $R^L \rightarrow R^r$ of~$\fds$ such
  that $R^L \subseteq O$ and $R^r \in O$, plus, for every $\fd$ $R^L \rightarrow
  R^r$ of~$\fds$ where $R^L \subseteq O$ and $R^r \notin O$, the key dependency
  $R^L \rightarrow O$.
\end{definition}

We first note the following immediate consequence of the Dense Interpretations
Theorem (Theorem~\ref{thm:combinatorial}):

\begin{corollary}
  \label{cor:combinatorial2}
  We can assume in the Dense Interpretations Theorem
  (Theorem~\ref{thm:combinatorial}) that the resulting instance $I$
  is such that each element occurs at exactly one position of the relation $R$:
  formally, for all $a \in \dom(I)$, there exists exactly one $R^p \in
  \positions(R)$ such that $a \in \pi_{R^p}(I)$.
\end{corollary}

\begin{proof}
  Create from $I$ the instance $I'$ whose domain is $\{(a, R^p) \mid a \in
  \dom(I), R^p \in \positions(\sigma)\}$ and which contains for every fact $F =
  R(\mybf{a})$
  of~$I$ a fact $F' = R(\mybf{b})$ such that $b_p = (a_p, R^p)$ for every
  $R^p \in \positions(\sigma)$. Clearly this defines a bijection $\phi$ from the
  facts of~$I$ to the facts of~$I'$, and for any facts $F$, $F'$ of~$I'$,
  $\ovl(F, F') = \ovl(\phi^{-1}(F), \phi^{-1}(F'))$. Thus any violation
  of the $\fd$s $\fds$ in $I'$ would witness one in $I$. Of course,
  $\card{\dom(I')} = \arity{\sigma} \cdot \card{\dom(I)}$, so that, letting $K'$ be our
  target constant factor between $\card{\dom(I')}$ and $\card{I}$, we must use
  $K \defeq K' \arity{\sigma}$ as the constant for the Dense Interpretation
  Theorem, so that $\card{I} \geq K' \arity{\sigma} \cdot \card{\dom(I)}$, which
  implies $\card{I'} \geq K' \card{\dom(I')}$.
\end{proof}

We also show two easy lemmas:

\begin{lemma}
  \label{lem:liftovl}
  Let $I$ be an instance, $\fds$ be a conjunction of $\fd$s, and $F \neq F'$ be two 
  facts of~$I$. Assume there is a position $R^p \in \pos(\sigma)$ such that,
  writing $O \defeq \nondanger(R^p)$, we have $\ovl(F, F')
  \subsetneq O$, and
  that $\{\pi_O(F), \pi_O(F')\}$  is not a violation of~$\fdrestr{\fds}{O}$.
  Then $\{F, F'\}$ is not a violation of~$\fds$.
\end{lemma}

\begin{proof}
  Assume by way of contradiction that $F$ and $F'$ violate an $\fd$ $\phi: R^L
  \rightarrow R^r$ of~$\fds$, which implies that
  $R^L \subseteq \ovl(F, F') \subseteq O$ and $R^r \notin \ovl(F, F')$. Now, if $R^r
  \in O$, then $\phi$ is in $\fdrestr{\fds}{O}$, so that $\pi_O(F)$ and
  $\pi_O(F')$ violate~$\fdrestr{\fds}{O}$, a contradiction. Hence, $R^r \in \pos(R) \backslash O$, and
  the key dependency $\kappa: R^L \rightarrow O$ is in $\fdrestr{\fds}{O}$, so that
  $\pi_O(F)$ and $\pi_O(F')$ must satisfy~$\kappa$.
  Thus, because $R^L \subseteq \ovl(F, F')$, we must have $\ovl(F, F') = O$,
  which is a contradiction because we assumed $\ovl(F, F') \subsetneq O$.
\end{proof}

\begin{lemma}
  \label{lem:connectndg}
  For any $(R^p, \mybf{C}) \in \rdfcl$, letting $O \defeq \nondanger(R^p)$, if
  $(R^p, \mybf{C})$ is unsafe, then there is no position $R^q \in O$ that
  determines $O$ in $\fdrestr{\fds}{O}$: formally, there is no $R^q \in O$ such that
  we have $R^q \rightarrow R^r$ in $\fdrestr{\fds}{O}$ for all~$R^r \in O$.
\end{lemma}

\begin{proof}
  Fix $D = (R^p, \mybf{C})$ in $\rdfcl$ and let $O$ be the non-dangerous
  positions of~$R^p$. We first show that if $\fds$ implies that $O$ has a unary key $R^s \in O$ in $\fds$, then
  $D$ is safe. Indeed, assume the existence of such a unary key $R^s$. If there were a $\fd$ $R^L \rightarrow R^r$ in $\fds$
  with $R^L \subseteq O$ and $R^r \notin O$, then, by transitivity, the $\ufd$
  $R^s \rightarrow R^r$ would be in $\ufds$, which by
  Lemma~\ref{lem:nondgr}
  implies that $R^r$ is non-dangerous for $R^p$ because $R^s \in O$ is
  non-dangerous for $R^p$. This contradicts our assumption that $R^r \notin O$.

  We must now show that if $O$ has a unary key in $O$ according to $\fdrestr{\fds}{O}$
  then $O$ has a unary key in $O$ according to $\fds$. It suffices to show that
  for any two positions $R^q, R^s \in O$, if $\phi: R^q \rightarrow R^s$ holds
  in~$\fdrestr{\fds}{O}$ then it also does in~$\fds$.
  Assuming to the contrary that there there is such a $\phi$, consider its
  derivation from the dependencies of~$\fdrestr{\fds}{O}$. Clearly the 
  derivation must be using one of the key dependencies $\kappa: R^L
  \rightarrow O$, which are the only dependencies in~$\fdrestr{\fds}{O}$ that
  are not in~$\fds$.
  But this means that, the first time we used such a dependency, we had derived
  a unary key dependency $R^q \rightarrow R^L$ using only the $\fd$s of $\fds$.
  Considering that $\kappa$ was created to stand for a $\fd$ $R^L \rightarrow
  R^r$ in $\fds$, with $R^r \notin O$, we deduce that we can derive from $\fds$
  that $R^q \rightarrow R^r$, contradicting again the fact that $R^r \notin O$
  (because $R^r$ should then be in $\nondanger(R^p)$). Hence, if $O$ has a unary
  key in~$O$ according to~$\fdrestr{\fds}{O}$ then $D$ is safe. Thus, we have
  proven the contrapositive of the desired result.
\end{proof}

We now prove Proposition~\ref{prp:preproc}.
The bulk of the work is to show the following claim, for each unsafe class of
$\rdfcl$. The construction of global envelopes from the individual envelopes is
then easy.

\begin{lemma}
  \label{lem:oneenv}
  For any unsafe class $D$ in $\rdfcl$ and constant $K$, one
  can construct a superinstance $I_0'$ of~$I_0$ that is $k$-sound for $\cq$,
  and an aligned superinstance
  $J = (I, \cov)$ of~$I_0'$ that satisfies $\fds$ with an envelope $E$ for $D$ of
  size $K \card{J}$.
\end{lemma}

\begin{proof}
  Fix the unsafe achieved fact class $D = (R^p, \mybf{C})$ and
  choose $F = R(\mybf{b})$ a fact of~$\chase{I_0}{\ids} \backslash I_0$
  that achieves $D$.
  Let $I_1$ be obtained from $I_0$ by
  applying $\uid$ chase steps on $I_0$
  to obtain a finite truncation of~$\chase{I_0}{\ids}$ that
  includes $F$ but no child fact of~$F$, and consider the aligned superinstance
  $J_1 = (I_1, \cov_1)$ where $\cov_1$ is the identity.

  Let $O \defeq \nondanger(R^p)$, and define a $\card{O}$-ary relation
  $\relrestr{R}{O}$; for convenience, we index its positions by $O$.
  Because $D$ is unsafe, by
  Lemma~\ref{lem:connectndg}, $\relrestr{R}{O}$ has no unary key in
  $\fdrestr{\fds}{O}$. Apply the Dense Interpretations Theorem
  (Theorem~\ref{thm:combinatorial}) to $\relrestr{R}{O}$ and $\fdrestr{\fds}{O}$
  with the additional condition of Corollary~\ref{cor:combinatorial2},
  taking $K \card{J_1}$ as the constant.
  We thus obtain an instance $I_D$ of~$\relrestr{R}{O}$ that satisfies~$\fdrestr{\fds}{O}$ and such that, letting $N \defeq
  \card{\dom(I_D)}$, we have $\card{I_D} \geq N K \card{J_1}$.
  Let $I'_D \subseteq I_D$ be an
  subinstance of size $N$ of~$I_D$ such that $\dom(I'_D) = \dom(I_D)$, that is, each
  element of~$\dom(I_D)$ occurs in some fact of~$I'_D$. This can clearly be
  ensured
  by picking, for any element of $\dom(I_D)$, one fact of~$I_D$ where it occurs, removing duplicate
  facts, and completing with other arbitrary facts of~$I_D$ to have $N$
  distinct facts. Number the facts of $I'_D$ as $F'_1, \ldots, F'_N$.

  We create $N-1$ disjoint 
  copies of~$J_1$, numbered $J_2$ to $J_N$.
  We call $J' = (I', \cov)$ the disjoint union $J_1 \sqcup \cdots \sqcup J_N$.
  It is clear that $J'$ is indeed an aligned
  superinstance of~$I_0'$, where $I_0'$ is formed of the $N$ disjoint copies of
  $I_0$, and $I_0'$ is clearly a $k$-sound superinstance of~$I_0$ for $\cq$.
For $1 \leq i \leq N$, we call $F_i = R(\mybf{a}^{\mybf{i}})$ the fact of~$I_i$ that
  corresponds to the achiever $F$ in $\chase{I_0}{\ids}$. In particular, for all
  $1 \leq i \leq N$, we have
  that $\cov(a^i_j) = b_j$ for all $j$, and $a^i_p$ is the only element of~$F_i$
  that also occurs in other facts of~$J_i$.

  We consider the application $f$ that maps $a^i_j$, for $1 \leq i \leq N$ and
  $R^j \in O$, to $\pi_{R^j}(F_i')$. This application $f$ is well-defined, because the $a^i_j$ are
  pairwise distinct. We extend $f$ to $\dom(I')$, and call the extension $f'$,
  by setting $f'(a) \defeq a$ if $a$ is not in the domain of~$f$.
  We call $I$ the image of~$I'$ under $f'$. In other words, $I$ is the
  underlying instance of~$J'$ except that 
  elements at positions of~$O$ in the facts $F_i$ were identified 
  so that the projections to~$O$ of the $f'(F_i)$ are isomorphic to the~$F_i'$. Because $a^i_j$ occurs only in $F_i$ for all $R^j \neq R^p$,
  and $R^p \notin O$, this means that the identified elements only occurred in
  the $F_i$ in~$I'$.

  We now build $J = (I, \cov)$ obtained by defining $\cov$ from the $\cov_i$ as
  follows: any element $a$ not in the domain of $f$ is mapped to $\cov_i(a)$ for
  the one $i$ such that $a \in \dom(I_i)$, and any $a$ in the domain of $f$ is
  mapped to $\cov_i(a')$ for any preimage of $a'$ by $f$.
  All that remains to show is that $J$ is indeed an
  aligned superinstance of~$I_0'$ satisfying the required conditions.

  We note that it is immediate that $J$ is a superinstance of~$I_0'$, as the
  achiever $F$ is not a fact of~$I_0$, so that $\dom(I_0)$ is not in the domain
  of $f$. It is clear that $J$ has $N \card{J_1}$ facts, because, as $R^p \notin
  O$, no facts can be identified by $f'$.
  We now claim that $J$ is an aligned superinstance of~$I_0'$, and that $E$,
  defined as the set of the tuples of~$I_D$, is an envelope for $I'$ and $D$.
  The fact that $\card{E} = K \card{J}$ is
  immediate.

  The fact that $\cov$ is a $k$-bounded simulation from $J$ to
  $\chase{I_0'}{\ids}$ is by induction. The case of~$k=0$ is trivial. The
  induction case is trivial for all facts except for the $h'(F_i)$,
  because the $a^i_j$ only occurred in~$I$ in the facts $F_i$, by our assumption
  that the $F_i$ have no children in the $I_i$ and by the fact that the exported
  position of~$F$ is $R^p \notin O$. Consider now one fact $F' =
  R(\mybf{c})$ of~$I'$ which is the image by $f'$ of a $F_i$.
  Choose $1 \leq p \leq \arity{R}$. We show that there exists a fact $F'' =
  R(\mybf{d})$ of
  $\chase{I_0'}{\ids}$ such that $\cov(c_p) = d_p$ and for all $1 \leq
  q \leq \card{R}$ we have $(I, c_q) \bsim_{k-1} (\chase{I_0'}{\ids}, d_q)$,
  which by induction hypothesis is implied by $\cov(c_q) \bbsim_k d_q$.
  Let $a^{i_0}_{j_0}$ be the preimage of~$a_p$ used to define $\cov(a_p)$; by the
  condition of Corollary~\ref{cor:combinatorial2}, we must have $j_0 = p$. Consider the fact
  $F'' = R(\mybf{d})$ of~$\chase{I_0'}{\ids}$ corresponding to $F_{i_0}$ in
  $I$. By definition, $\cov(c_p) = \cov(a^{i_0}_{j_0}) = d_p$. Fix now $1 \leq q
  \leq \arity{R}$. Let $a^{i_0'}_{j_0'}$ used to define $\cov(c_q)$; again $j_0'
  = q$ and $\cov(c_q)$ is $\pi_{R^q}(F''')$ for the fact $F''' = R(\mybf{e})$ of
  $\chase{I_0'}{\ids}$ corresponding to $F_{i_0'}$ in $I$. But as both $F'''$
  and $F''$ are copies of the same achiever fact $F$ of $\chase{I_0}{\ids}$, we
  have $d_q \bbsim_k e_q$, so that $\cov(c_q) \bbsim_k d_q$, what we
  wanted to show.
  This proves that $\cov$ is indeed a $k$-bounded simulation from~$J$
  to $\chase{I_0}{\ids}$.
  
  We show that $J$ satisfies $\fds$. 
  As $I$ satisfies $\fds$,
  any new violation of~$\fds$ in $I'$ relative to $I$
  must include some fact $F = h'(F'_{i_0})$, and some fact $F'$ overlapping with
  $F$, so necessarily $F' = h'(F'_{i_1})$ for some $i_1$ by construction of~$I'$,
  and $\ovl(F, F') \subseteq O$. We now use Lemma~\ref{lem:liftovl} to deduce
  that we cannot have $\ovl(F, F') \subsetneq O$, so $\ovl(F, F') = O$. By our
  definition of $f$ and of the $F'_i$ this implies that $F'_{i_0} = F'_{i_1}$, a
  contradiction because $F \neq F'$.

  Thus, from the above, and as the technical conditions of the definition of
  aligned superinstances are clearly respected, $J$ is indeed an aligned
  superinstance of~$I_0'$.

  Last, we check that $E$ is indeed an envelope. Indeed, it satisfies
  $\fdrestr{\fds}{O}$ by construction, so the first two conditions are
  respected. The third condition is respected by the condition of
  Corollary~\ref{cor:combinatorial2}, and because the $f(a^i_j)$ always occur at
  position $R^j$ in some fact of $I'_D$, as we constructed $I'_D$ such that
  $\dom(I'_D) = \dom(I_D)$. The last condition is true because the envelope
  elements are only used in the $f(F_i)$, and the $\cov$-images of the $f(F_i)$
  are copies in $\chase{I_0'}{\ids}$ of the same achiever fact $F$ in $\chase{I_0}{\ids}$.

  Hence, $J$ is indeed an aligned superinstance of a $k$-sound $I_0'$ that
  satisfies $\fds$ and has an envelope of size $K \card{J}$, proving the desired
  claim.
\end{proof}

We now prove the main result by building $I_0'$ and the aligned superinstance $J = (I,
\cov)$ of~$I_0'$ that has a global envelope $\calE$.
As $\rdfcl$
is finite, we build one $J_D$ per $D \in \rdfcl$. When $D$ is unsafe, we use the
previous lemma. When $D = (R^p, \mybf{C})$ is safe, we just take a
single copy $J_D$ of the truncated chase to achieve the class $D$, and take as the
only fact of the envelope the projection to $\nondanger(R^p)$ of the fact of~$J_D$ corresponding to
the achiever of~$D$ in $\chase{I_0}{\ids}$.
As $\rdfcl$ is finite and its size is a constant, we can ensure that $\card{\calE(D)}$ for all unsafe $D
\in \rdfcl$ is $\geq
(K+1)\card{I}$, by taking sufficiently large
$K$ when we apply Lemma~\ref{lem:oneenv} for each unsafe class.

Let $J$ be the disjoint union of the $J_D$.
Each $J_D$ is an aligned superinstance of an $(I_0')_D$ which is a $k$-sound superinstance
of~$I_0$. Hence, $J$ is an aligned superinstance of the union of the $(I_0')_D$
which is also $k$-sound.
There are no violations of~$\fds$ in $J$ because there are none in
any of the $J_D$, and the union is disjoint. The disjointness of domains of
envelopes is because the $J_D$ are disjoint. It is easy to see that $J$ is $(K
\card{I})$-envelope-saturated, because $\card{\calE(D)} \geq (K+1)\card{I}$ for
all unsafe $D \in \rdfcl$, so the
number of remaining facts of each envelope for an unsafe class is $\geq K \card{I}$ (every
fact of $I$ eliminates at most one fact in each envelope).
Hence, the proposition is proven.

\mysubsection{Proof of the Dense Interpretations Theorem
(Theorem~\ref{thm:combinatorial})}
\label{apx:prf_combinatorial}
Remember that we want to show:
\combinatorial*

Fix the relation $R$, and let $\fds$ be an arbitrary set of~$\fd$s which we
assume is closed under $\fd$ implication.
Let $\ufds$ be the $\ufd$s implied by $\fds$; it is also closed under $\fd$
implication. Recall
the definition of~$\ovl$ (Definition~\ref{def:overlap}).
We introduce a notion of \defo{safe overlaps} for $\ufds$,
which depends only on $\ufds$ but (we will show) is
a sufficient condition to satisfy~$\fds$:

\begin{definition}
  We say a subset $O \subseteq \pos(R)$ is \deft{safe} for $\ufds$ if $O$ is
  empty or for every
  $R^p \in \pos(R) \backslash O$, there exists $R^q \in \pos(R)$ such that the
  unary key dependency $R^q \rightarrow O$ is implied by $\ufds$ but the $\ufd$
  $R^q \rightarrow R^p$ does not hold in $\ufds$.
  
  We say that an instance $I$ has the \deft{safe
  overlaps} property (for $\ufds$) if for every $F \neq F'$ of~$I$, $\ovl(F, F')$ is
  safe.
\end{definition}

We now claim the following lemma, and its immediate corollary:

\begin{restatable}{lemma}{safefdso}
  If $O \subseteq \pos(R)$ is safe for~$\ufds$ then there is no $\fd$ $\phi: R^L
  \rightarrow R^r$ in $\fds$ such that $R^L \subseteq O$ but $R^r \notin O$.
\end{restatable}

\begin{proof}
  If $O$ is empty the claim is immediate. Otherwise,
  assume to the contrary the existence of such an $\fd$~$\phi$. As $R^r \notin O$ and $O$ is
  safe, there is $R^q \in \pos(R)$ such that $R^q \rightarrow O$
  holds in $\ufds$ but $R^q \rightarrow R^r$ does not hold in $\ufds$.
  Now, as $R^L \subseteq O$, we know that $R^q \rightarrow R^L$
  holds in~$\ufds$, so that, by transitivity of~$\fds$,
  $\phi':R^q \rightarrow R^r$ holds in $\fds$. As $\phi'$ is a $\ufd$, this
  implies it holds in~$\ufds$, a contradiction.
\end{proof}

\begin{restatable}{corollary}{safefds}
  \label{lem:safefds}
  For any instance $I$,
  if $I$ has the safe overlaps property for $\ufds$, then $I$ satisfies $\fds$.
\end{restatable}

\begin{proof}
  Considering two facts $F$ and $F'$ in~$I$, as $\ovl(F, F')$ is safe, we know
  that for any $\fd$ $\phi: R^L \rightarrow R^r$ in $\fds$, we cannot have $R^L
  \subseteq O$ but $R^r \notin O$. Hence, $F$ and $F'$ cannot be a violation
  of~$\phi$.
\end{proof}

Thus, it suffices to show the following generalization of the 
Dense Interpretations Theorem:

\begin{theorem}
  \label{thm:fds}
  Let $R$ be a relation and $\ufds$ be a set of~$\ufd$s over $R$.
  Let $D$ be the number of positions
  of the smallest key of~$R$ for $\ufds$: formally, $D \defeq \card{K}$, where
  $K \subseteq \pos(R)$ is such that $R^K \rightarrow R^p$ holds in $\ufds$ for
  all $R^p \in \pos(R)$, and $K$ has minimal cardinality among all subsets
  of~$\pos(R)$ with this property.
  Let $x$ be ${D \over D-1}$ if $D > 1$ and $1$ otherwise.
  
  For every
  $N \geq 1$, there exists a finite instance~$I$ of~$R$ such that
  $\card{\dom(I)}$ is $O(N)$, $\card{I}$ is $\Omega(N^x)$,
  and $I$ has the safe overlaps property for~$\ufds$.
\end{theorem}

It is clear that this theorem implies the Dense Interpretations Theorem, because
if $R$ has no unary key for $\fds$ then $D > 1$ and thus $x > 1$, which implies that, for any~$K$, by taking
a sufficiently large $N$, we can obtain an instance~$I$ for $R$ with $N$
elements and $KN$ facts that has the safe overlaps property for $\ufds$; now, by
Lemma~\ref{lem:safefds}, this implies that $I$ satisfies $\fds$.

\medskip

We will now prove Theorem~\ref{thm:fds}. Fix the relation $R$ and set of~$\ufd$s
$\ufds$. The case of~$D = 1$ is vacuous and can be eliminated directly (consider
the instance $\{R(a_i, \ldots, a_i) \mid 1 \leq i \leq N\}$). Hence,
assume that $D > 1$, and let $x \defeq {D \over D-1}$.

\medskip

We first show the claim on a specific relation $R_0$ and set $\ufds^0$ of~$\ufd$s. We
will then generalize the construction to arbitrary relations and $\ufd$s. Let
$T_0 \defeq \{1, \ldots, D\}$, and consider a bijection $\nu: \{1, \ldots, 2^D\}
\to \parts(T_0) \backslash \{\emptyset\}$. Let $R_0$ be a $(2^D-1)$-ary relation,
and take $\ufds^0 \defeq \{R^i \rightarrow R^j \mid \nu(i)
\subseteq \nu(j)\}$. Note that $\ufds^0$ is clearly closed under implication of
$\ufd$s.
Fix $N \in \mathbb{N}$, and let us construct an instance $I_0$ with $O(N)$
elements and $\Omega(N^x)$ facts.

Fix $n \defeq \lfloor N^{1/(D-1)} \rfloor$.
Let $\calF$ be the set of partial functions from $T_0$ to $\{1, \ldots, n\}$,
and write $\calF = \calF_\t \sqcup \calF_\p$, where $\calF_\t$ and $\calF_\p$
are respectively the total and the strictly partial functions.
We take $I_0$ to consist of one fact $F_f$ for each $f \in \calF_\t$, where $F_f
= R_0(\mybf{a}^{\mybf{f}})$
is defined as follows: for $1 \leq i \leq 2^D$, $a^f_i \defeq \restr{f}{T_0 \backslash \nu(i)}$. In particular:
\begin{compactitem}
\item $a^{f}_{\nu^{-1}(T_0)}$, the element of~$F_{f}$ at the position mapped to
  $T_0 \in \parts(T_0) \backslash \{\emptyset\}$, is the strictly partial function
that is nowhere defined;
\item $a^{f}_{\{i\}}$, the element of
  $F_{f}$ at the position mapped to $\{i\} \in \parts(T_0) \backslash
  \{\emptyset\}$,
  is the strictly partial function equal to $f$ except that it is undefined on $i$.
\end{compactitem}
Hence, $\dom(I_0) = \calF_\p$ (because $\emptyset$ is not in the image of
$\nu$), so that $\card{\dom(I_0)} = \sum_{0 \leq i < D} {D \choose i} n^i$.
Remembering that $D$ is a constant, this implies that $\card{\dom(I_0)}$ is
$O(n^{D-1})$, so it is $O(N)$ by definition of $n$.
Further, we claim that $\card{I_0} = \card{\calF_\t} = n^D = N^x$.
To show this, consider
two facts $F_f$ and $F_g$, and show that $F_f = F_g$ implies $f = g$,
so there are indeed $\card{\calF_\t}$ different facts in~$I_0$. As
$\pi_{\nu^{-1}(\{1\})}(F_f) =
\pi_{\nu^{-1}(\{1\})}(F_{g})$, we have $f(t) = g(t)$ for all $t \in T_0 \backslash \{1\}$,
and looking at $\pi_{\nu^{-1}(\{2\})}(F_f)$ and $\pi_{\nu^{-1}(\{2\})}(F_{g})$
concludes (here we use the fact that $D \geq 2$).
Hence, the cardinalities of~$I_0$ and of its domain are suitable.

We must now show that $I_0$ has the safe overlaps property. For this we first make the
following general observation:

\begin{lemma}
  Let $\ufds$ be any conjunction of $\ufd$s and $I$ be an instance such that $I
  \models \ufds$. Assume that, for any pair of facts $F \neq F'$ of~$I$ that overlap, there
  exists $R^p \in \ovl(F, F')$ which is a unary key for $\ovl(F, F')$. Then $I$
  has the safe overlaps property for $\ufds$.
\end{lemma}

\begin{proof}
  Consider $F, F' \in I$ and $O \defeq \ovl(F, F')$. If
  $F = F'$, then $O = \pos(R)$, and $O$ is clearly safe.
  Otherwise, if $F \neq F'$, let $R^p \in \pos(R) \backslash O$. Let
  $R^q \in O$ be the unary key of~$O$. We know that $R^q \rightarrow O$ holds in
  $\ufds$, so to show that $O$ is safe it suffices to show that $\phi: R^q \rightarrow
  R^p$ does not hold in $\ufds$. However, if it did, then as $R^q \in O$ and
  $R^p \notin O$, $F$ and $F'$ would witness a violation of~$\phi$,
  contradicting the fact that $I$ satisfies $\ufds$.
\end{proof}

So we show that $I_0$ satisfies $\ufds^0$ and that every non-empty overlap
between facts of~$I_0$ has a unary key.

First, to show that $I_0$ satisfies $\ufds^0$, observe that whenever $\phi: R_0^i
\rightarrow R_0^j$ holds in $\ufds$,
then $\nu(i) \subseteq \nu(j)$, so that, for any fact~$F$ of~$I_0$, for any $1
\leq t \leq T_0$,
whenever $(\pi_j(F))(t)$ is defined,
so is $(\pi_i(F))(t)$, and we have $(\pi_j(F))(t) = (\pi_i(F))(t)$.
Hence, letting $F$ and $F'$ be two facts of~$I_0$ such that $\pi_i(F) = \pi_i(F')$,
we know that $\pi_j(F)$ is defined iff $\pi_j(F')$ is (as this only depends
on~$j$), and, if both are defined, the previous observation shows
that $\pi_j(F) = \pi_j(F')$.
Hence, $F$ and $F'$ cannot witness a violation of $\phi$.

Second, considering two facts 
$F_{f} = R_0(\mybf{a}^{\mybf{f}})$ and $F_{g} = R_0(\mybf{a}^{\mybf{g}})$, with $f \neq g$ so
that $F_{f} \neq F_{g}$, we show that if $\ovl(F_f, F_g)$ is non-empty then it has a unary key.
Let $O \defeq \{t \in T_0 \mid f(t) = g(t)\}$, and let $X = T_0
\backslash O$; we have $X \neq \emptyset$, because otherwise $f =
g$, so we can define $p \defeq \nu^{-1}(X)$.
We will show that $\ovl(F_f, F_g) = \{R^i \in \positions(R_0) \mid X \subseteq
\nu(i)\}$. This implies that $R^p \in \ovl(F_f, F_g)$ and that $R^p$ is a unary
key of~$\ovl(F_f, F_g)$, because,
for all $R^q \in \overlap(F_{f}, F_{g})$, $X \subseteq \nu(R^q)$, so that
$R^p \rightarrow R^q$ holds in $\ufds$.

Indeed, consider $R^i$ such that $X \subseteq \nu(i)$. Then $T_0 \backslash
\nu(i) \subseteq T_0 \backslash X$, so that,
because $a^f_i = \restr{f}{T_0 \backslash \nu(I)}$ and $a^g_i = \restr{g}{T_0 \backslash
\nu(I)}$, we have $a^f_i = a^g_i$ by definition of~$O = T_0 \backslash X$. Thus
$R^i \in \overlap(F_{f}, F_{g})$. Conversely, if $R^i \in
\overlap(F_{f}, F_{g})$, then we have $a^f_i =
a^g_i$, so by definition of~$O$ we must have $T_0
\backslash \nu(i) \subseteq O' = T_0 \backslash X$, which implies $X \subseteq
\nu(i)$.

Hence, $I_0$ is a finite instance of~$\ufds$ which satisfies the safe overlaps
property and contains $O(N)$ elements and
$\Omega(N^{D/(D-1)})$ facts. This concludes the proof of Theorem~\ref{thm:fds}
for the specific case of~$R_0$ and $\ufds^0$.

\medskip

Let us now show the claim for the actual $R$ and $\ufds$. Let $K$ be 
a key of~$R$ of minimal cardinality, so that $\card{K} = D$.
Let $\lambda$ be any bijective labeling from $K$ to $T_0$.
Extend $\lambda$ to a function $\mu$ from $\positions(R)$ to $\parts(T_0) \backslash
\{\emptyset\}$ such that, for every $R^p \in \positions(R)$ and $R^k \in K$, we have
$\lambda(R^k) \in \mu(R^p)$ iff $R^k = R^p$ or $R^k \rightarrow R^p$ holds in
$\ufds$.

Now, create the instance $I$ of~$R$ from $I_0$ by creating, for every fact $F_0 =
R_0(\mybf{a})$ of~$I_0$, a fact $F = R(\mybf{b})$ in $I$, with $b_i =
a_{\nu^{-1}(\mu(R^i))}$ for all $1 \leq i \leq \card{R}$.

We do not create duplicate facts by the same argument as before, considering the
projection of the facts of~$I$ to $R^{k_1} \neq R^{k_2}$ in~$K$, because
$\mu(R^{k_1}) = \{\lambda(R^{k_1})\}$ and $\mu(R^{k_2}) = \{\lambda(R^{k_2})\}$
(otherwise this contradicts the minimality of~$K$). Hence $I$, as $I_0$, has a
suitable number of facts, and a suitable domain cardinality 
because $\dom(I) \subseteq \dom(I_0)$.

Let us now show that overlaps are safe in~$I$. Consider two facts $F, F'$ of
$I$ that overlap, and let $O \defeq \ovl(F, F')$.
We first claim that there exists $\emptyset \subsetneq K' \subseteq K$,
such that, letting $X' \defeq \{\lambda(R^k) \mid R^k \in K'\}$, we have
$\ovl(F, F') = \{R^i \in \positions(R) \mid X' \subseteq \mu(R^i)\}$. Indeed,
letting $F_f$ and $F_g$ be the facts of~$I_0$ used to create $F$ and $F'$, we
previously showed the existence of~$\emptyset \subsetneq X \subseteq T_0$ such that 
$\ovl(F_f, F_g) = \{R^i \in \positions(R_0) \mid X \subseteq \nu(i)\}$.
Our definition of~$F$ and $F'$ from $F_f$ and $F_g$ makes it clear that we can
satisfy the condition by taking $K' \defeq \lambda^{-1}(X)$, so that $X' = X$.

Consider now $R^p \in \pos(R) \backslash O$. We
cannot have $X' \subseteq \mu(R^p)$, otherwise $R^p \in O$. Hence, there exists
$R^k \in K'$ such that $\lambda(R^k) \notin \mu(R^p)$. This implies that $R^k
\rightarrow R^p$ does not hold in $\ufds$. However, as $R^k \in K'$, we have
$\lambda(R^k) \in \mu(R^q)$ for all $R^q \in O$, so that $R^k \rightarrow O$
holds in $\ufds$. This proves that~$O = \ovl(F, F')$ is safe. Hence, $I$ has the
safe overlaps property, which concludes the proof.

\mysubsection{Proof of Lemma~\ref{lem:envfds} (Envelope-thrifty chase steps
satisfy $\fds$)}
\envfds*

Consider an application of an envelope-thrifty chase step:
let $\tau : \ui{R^p}{S^q}$ be the $\uid$, 
let $O \defeq \nondanger(S^q)$,
let $J = (I, \cov)$ be the aligned superinstance of~$I_0$, 
let $F_\w = S(\mybf{b}')$ the chase witness,
let $D = (S^q, \mybf{C})$ be the fact class,
let $F_\n = S(\mybf{b})$ be the new fact to be created, and
let $\mybf{t}$ be the remaining tuple of~$\calE(D)$ used to define $F_\n$.

We first check that envelope-thrifty chase steps are well-defined in the sense that the
fact class $D = (S^q, \mybf{C})$ is indeed achieved in
$\chase{I_0}{\ids}$, so it is in $\rdfcl$. To see why,
observe that $F_\w$ is a fact of
$\chase{I_0}{\ids}$ whose fact class is $(S^q, \mybf{C})$. Indeed, by
Lemma~\ref{lem:wadef}, $b'_q$
is the exported element of~$F_\w$, and clearly $b'_i \in C_i$ for
all $S^i \in \pos(S)$.
Hence indeed $D \in \rdfcl$.

It is then clear that envelope-thrifty chase steps are well-defined, in the sense that
they are indeed thrifty chase steps: elements reused from the envelopes already
occur at the positions where they are used in the new fact~$F_\n$. Further, their
$\cov$-image is the right one, by definition of an envelope.

\medskip

We first prove that~$J'$ is still an aligned superinstance. This is shown
exactly as in Lemma~\ref{lem:ftok}, except for the fact that $J' \models \ufds$
which was specific to fact-thrifty chase steps.
We show instead that 
$J' \models \fds$, using the assumption that $J \models \fds$.
Recall the definition of~$\ovl$ (Definition~\ref{def:overlap}), and
assume by contradiction the existence of a violation of~$\fds$ in $J'$.
The violation must be between $F_\n$ and an
existing fact $F = S(\mybf{c})$. However, because only the elements at positions in $O$
already occur at their position, we must have $\ovl(F_\n, F) \subseteq O$.
As $\pi_O(F_\n)$ was defined using elements of
$\dom(\calE(D))$, taking $S^r \in \ovl(F_\n, F) \subseteq O$, we have $c_r = b_r
\in \pi_{S^r}(\calE(D))$, so that, by definition of~$\calE(D)$,  we know that
$\pi_O(\mybf{c})$ is a tuple of~$\calE(D)$.
If $\ovl(F_\n, F'') \subsetneq O$ then we have a contradiction by applying
Lemma~\ref{lem:liftovl} to $\mybf{t}$ and $\pi_O(\mybf{c})$ in
$\calE(D)$. Hence $\ovl(F_\n, F'') = O$
So, if $D$ is unsafe, we have a contradiction because 
$F$ witnesses that $\mybf{t}$ was not a remaining tuple,
so we cannot have used it to define $F_\n$. If $D$ is safe,
there is no $\fd$ $R^L \rightarrow R^r$ of~$\fds$ with $R^L \subseteq O$ and
$R^r \notin O$, so $F$ and $F_\n$ cannot violate $\fds$, a contradiction again.

\medskip

We now prove that $\calE$ is still a global envelope of~$J'$ after performing an
envelope-thrifty chase step. The condition on the disjointness of the envelope
domains only concerns $\calE$, which is unchanged. Hence, we need only show
that, for any $D' \in \rdfcl$,
$\calE(D')$ is still an envelope.
Except the last one, all conditions of the definition of envelopes either
concern only the envelope~$\calE(D')$, which is unchanged, or they are preserved
when more facts are created in~$J'$.
The last condition needs only to be checked about the new fact $F_\n$ created in this chase step.

Except for the elements of
$F_\n$ at positions in $O$, all elements of~$F_\n$ did not occur at
the positions where they occur in $F_\n$, by definition of a
thrifty chase step. So they cannot be elements of~$\dom(\calE)$ occurring in
$F_\n$ at the one position where they occur in the one envelope
where they occur, because
we know that elements from any envelope already occur in $J$ at that
position. So we only need to check the
condition for the~$b_r$ for~$S^r \in O$. But because the envelopes of~$\calE$
are pairwise disjoint and as the~$b_r$ are all in $\dom(\calE(D))$, we only need to check
the condition for $\calE(D)$. Now, $\mybf{t}$ witnesses that $\pi_O(\mybf{b}) \in \calE(D)$.
Hence $\calE$ is still a global envelope of~$J'$.

\medskip

Last, to see that the resulting $J'$ is $(n-1)$-envelope-saturated, it suffices to
observe that the new fact $F_\n$ witnesses that, for each unsafe class $D \in
\rdfcl$, the remaining tuples of $\calE(D)$ for $J'$ are those of $\calE(D)$ for
$J$ minus at most one tuple (namely, some projection of $F_\n$). This concludes the proof.

\mysubsection{Proof of the Envelope-Thrifty Completion Proposition
(Proposition~\ref{prp:etcomp})}
\label{apx:prf_etcomp}

\etcomp*

The completion process for envelope-thrifty chase steps is defined in the same
way as for fact-thrifty chase steps, except that the elements reused at
non-dangerous positions are different. By definition of thrifty chase steps, the
choice of elements reused at those positions cannot make any new $\uid$
applicable, or satisfy any $\uid$, because the elements thus reused are required to already
occur at the positions where they are used in the new fact. Further,
envelope-thrifty chase steps do not introduce $\ufd$ violations (in fact, they
do not introduce $\fd$ violations), as
follows from Lemma~\ref{lem:envfds}. Hence, we can indeed define the completion
process for envelope-thrifty chase steps exactly like the completion process for
fact-thrifty chase steps, are long as the instance is envelope-saturated.
Whenever an envelope-exhausted instance is obtained
at any point of the process, we abort and set it to be the final instance.

Assuming that we do not reach any envelope-exhausted instance,
the fact that $\calE$ is still a global envelope of the result $J'$ of the envelope-thrifty
completion process, and that $J'$ satisfies $\fds$ in addition to $\ids$, is by
Lemma~\ref{lem:envfds}.

\mysubsection{Proof of the Envelope Blowup Lemma (Lemma~\ref{lem:eblowup})}
\label{apx:prf_blowup}

\eblowup*

We first observe that applying a chase round to an aligned superinstance $J =
(I, \cov)$ of~$I_0$ by any form of thrifty chase steps
(Definition~\ref{def:thrifty}) only increases its size by a multiplicative constant. This
is because $\card{\dom(I)} \leq \arity{\sigma} \cdot \card{I}$, and the number of
facts created per element of~$I$ in a chase round is at most
$\card{\positions(\sigma)}$.

Remember that the envelope-completion process starts by constructing an ordered
partition $\mybf{P} = (P_1, \ldots, P_{\neqidsc})$
of $\ids$ (Definition~\ref{def:opartition}).
This $\mybf{P}$ does not depend on the aligned superinstance.
Hence, as we satisfy the
$\uid$s of each $P_i$ in turn, if we can show that the instance size only
increases by a multiplicative constant for each class, then the blow-up
for the entire process is by a multiplicative constant (obtained as the product
of the constants for each $P_i$).

For trivial classes, we apply one chase round by fresh envelope-thrifty chase
steps (Corollary~\ref{cor:trivscc}), so
the blowup is by a multiplicative constant by our initial observation.

For non-trivial classes, we apply the Fact-Thrifty Completion Proposition
(Proposition~\ref{prp:ftcomp}), modified to use envelope-thrifty rather than
fact-thrifty chase steps (but the exact same steps are applied). Remember that
this proposition first ensures $k$-reversibility by applying $k+1$
envelope-thrifty chase rounds (Proposition~\ref{prp:achkrev}) and then makes the
result satisfy $\ids$ using the Guided Chase Lemma (Lemma~\ref{lem:guidedb}).
Ensuring $k$-reversibility only implies a blowup by a multiplicative constant,
because it means applying $k+1$ envelope-thrifty chase rounds. Hence, we focus
on the Guided Chase Lemma.

The lemma starts by constructing a balanced
pssinstance $P$ using the Balancing Lemma (Lemma~\ref{lem:hascompletion}), and a
$\ucon$-compliant piecewise realization $\pire$ of $P$ by the Realizations Lemma
(Lemma~\ref{lem:balwsndc}), and then performs envelope-thrifty chase steps to
satisfy $\ids$ following $\pire$.
We know that, whenever we apply a envelope-thrifty chase step to an element $a$
in the guided chase, $a$ occurs after the chase step at a new
position where it did not occur before. Hence, it suffices to show that
$\card{\dom(P)}$ is within a constant factor of $\card{J}$, because then we know
that the final number of facts once the guided chase is over will be $\leq
\card{\dom(P)} \cdot \card{\positions(\sigma)}$.

To show this,
remember that $\dom(P) = \dom(J) \sqcup \calH$, where $\calH$ is the helper set.
Hence, we only need to show that $\card{\calH}$ is within a
multiplicative constant factor of~$\card{J}$.
From the proof of the Balancing
Lemma, we know that $\calH$ is a disjoint union of $\leq \card{\positions(\sigma)}$ sets
whose size is linear in $\card{\dom(J)}$ which is itself $\leq \arity{\sigma}
\cdot \card{J}$. Hence, the Guided Chase Lemma only gives rise to a blowup by a
constant factor.
As we justified, this implies the same about the entire
completion process, and concludes the proof.

\section{Proofs for Section~\ref{sec:cycles}: Cyclic Queries}
\label{apx:cycles}

In this section, we extend our construction of superinstances that satisfy
$\con$ and are $k$-sound for $\acq$, to superinstances that are $k$-sound for
$\cq$ while still satisfying $\con$.

\mysubsection{Proof of the Simple Product Lemma
(Lemma~\ref{lem:prodppty})}
\label{apx:prf_prodppty}
\prodppty*

Fixing the superinstance $I$ of~$I_0$ that is $k$-sound for $\acq$ and
$k$-instance-sound, and the $(2k+1)$-acyclic group $G$ generated by
$\lab{I}$, consider $I' \defeq (I, I_0) \sprod G$, which is a
superinstance of~$I_0$ (up to our identification of~$(a, e)$ to $a$ for $a \in
\dom(I_0)$, where $e$ is the neutral element of~$G$). We must show that $I'$ is
$k$-sound for $\cq$.

We start by proving a simple lemma:

\begin{lemma}
  \label{lem:diffact}
  For any $\cq$ $q$ and instance $I$, if
  $I \models q$ and some match $h$ of~$q$ in $I$ maps two different atoms
  of~$q$ to the same fact $F$, then there is a strictly smaller $q'$
  which entails $q$ and has a match $h'$ in $I$ such that, seeing matches as
  subinstances of~$I$, $\dom(h') \subseteq \dom(h)$.
\end{lemma}

% similar to chandra1997optimal
\begin{proof}
  Fix $q$, $I$, $h$, and let $A = R(\mybf{x})$ and $A' = R(\mybf{y})$ be the two
  atoms of~$q$ mapped to the same fact $F$ by $h$.
  Necessarily $A$ and $A'$ are atoms for the same relation $R$ of the fact $F$,
  and as $h(A) = h(A')$ we know that $h(x_i) = h(y_i)$ for all $R^i \in \pos(R)$.

  Let $\dom(q)$ be the set of variables occurring in $q$. Consider the
  application $f$ from $\dom(q)$ to $\dom(q)$ defined by $f(y_i) = x_i$ for all
  $i$, and $f(x) = x$ if $x$ does not occur in~$A'$. Observe that this ensures
  that $h(x) = h(f(x))$ for all $x \in \dom(q)$. Let $q' = f(q)$ be the query
  obtained by replacing every variable $x$ in $q$ by $f(x)$, and, as $f(A') =
  f(A)$, removing one of those duplicate atoms so that $\card{q'} < \card{q}$.
  Let $h' = \restr{h}{\dom(q')}$. Clearly the image of~$h'$ is a subset of that
  of~$h$, and to see why this is a match of~$q'$ observe that any atom $f(A'')$
  of~$q'$ is homomorphically mapped by $h'$ to $h(A'')$ because $h'(f(x)) =
  h(x)$ for all $x$ so $h'(f(A'')) = h(A'')$.

  To see why $q'$ entails $q$, observe that $f$ defines a homomorphism from
  $q$ to $q'$, so that, for any match $h''$ of~$q'$ on an instance $I'$, $h''
  \circ f$ is a match of~$q$ on $I'$.
\end{proof}

Fix now a $\cq$ $q$ such that $\card{q} \leq k$, and assume that $I' \models q$:
let $h$ be a match of~$q$ in~$I$.
Let us show that $\chase{I_0}{\ids} \models q$.

Let $\pr$ be the application from $I'$ to $I$ defined by $\pr: (a, g) \mapsto a$ for
all $a \in \dom(I)$ and $g \in G$. It is clear that $\pr$ is a homomorphism from
$I'$ to $I$ that
maps $\dom(I_0) \times G$ to $\dom(I_0)$.
Hence, if $h$ involves some element of~$\dom(I_0)
\times G$, then
$q$ has a match in $I$ involving an element of~$I_0$.
Hence, as $I$ is $k$-instance-sound, $\chase{I_0}{\ids} \models q$.
We accordingly assume that $h$ does not
involve an element of~$\dom(I_0) \times G$.

If we
can show that there is a query $q'$ of~$\acq$, $\card{q'} \leq k$, such that $q'$
entails $q$ and $I \models q'$, then, as $I$ is $k$-sound for $\acq$, this suffices to
conclude that $\chase{I_0}{\ids} \models q'$, hence $\chase{I_0}{\ids} \models
q$ because $q'$ entails $q$. So by way of contradiction we assume that $q$ is a
query with a match i$h$ n $I'$ involving no element of~$\dom(I_0) \times G$ such that there is no $q' \in \acq$, $\card{q'} \leq k$, where $q'$
entails $q$ and $I \models q'$; and we take this counterexample query $q$ to be of minimal size.

In particular, this means we assume that $q$ is not in $\acq$, otherwise we
could take $q' = q$, because $I \models q$, as evidenced by $\pr \circ h$.
So consider a Berge cycle $C$ of~$q$, of the form $A_1, x_1, A_2, x_2, \ldots,
A_n, x_n$, where the $A_i$ are pairwise distinct atoms and the $x_i$ pairwise
distinct variables, and for all $1 \leq i \leq n$, variable $x_i$ occurs at position $q_i$
of atom $A_i$ and position $p_{i+1}$ of~$A_{i+1}$, with addition modulo $n
\defeq \card{C}$. We assume without loss of generality that $p_i \neq q_i$ for
all $i$. However, we do not assume that $n \geq 2$: either $n \geq 2$ and $C$ is
really a Berge cycle according to our previous definition, or $n = 1$ and
variable $x_1$ occurs in atom $A_1$ at positions
$p_1 \neq q_1$, which corresponds to the case where there are
multiple occurrences of the same variable in an atom.

For $1 \leq i \leq n$,
we write $F_i = R_i(\mybf{a}^{\mybf{i}})$ the image of~$A_i$ by $h$ in $I'$;
by definition of~$I'$, because $h$ involves no element of~$I_0 \times G$ and
hence no fact of~$I_0 \times G$, there
is a fact $F'_i = R_i(\mybf{b}^{\mybf{i}})$ of~$I$ and $g_i \in G$ such that $a^i_j = (b^i_j,
g_i \cdot \l^{F'_i}_j)$ for $R_i^j \in \pos(R_i)$.
Now, for all $1 \leq i \leq n$, 
as $h(x_i) = a^i_{q_i} = a^{i+1}_{p_i+1}$ for all $1 \leq i \leq n$, we deduce
by projecting on the second component that $g_i \cdot \l^{F'_i}_{q_i} = g_{i+1}
\cdot \l^{F'_{i+1}}_{p_{i+1}}$, so that, by collapsing the equations of the cycle
together, $\l^{F'_1}_{q_1} \cdot
(\l^{F'_{2}}_{p_{2}})^{-1} \cdot \cdots \cdot \l^{F'_{n-1}}_{q_{n-1}} \cdot
(\l^{F'_{n}}_{p_{n}})^{-1} \cdot \l^{F'_{n}}_{q_{n}} \cdot
(\l^{F'_{1}}_{p_{1}})^{-1} = e$.

As the girth of~$G$ under $\lab{I}$ is $\geq 2k+1$, and this product contains
$2n \leq 2k$ elements, we must have either $\l^{F'_i}_{q_i} =
\l^{F'_{i+1}}_{p_{i+1}}$ for some $i$, or $\l^{F'_i}_{p_i} = \l^{F'_i}_{q_i}$ for
some $i$. The second case is impossible because we assumed that $p_i \neq q_i$
for all $1 \leq i \leq n$. Hence, necessarily $\l^{F'_i}_{q_i} = \l^{F'_{i+1}}_{p_{i+1}}$, so in
particular $F'_i = F'_{i+1}$. Hence the atoms $A_i \neq A_{i+1}$ of~$q$ are
mapped by $h$ to the same fact $F'_i = F'_{i+1}$. We conclude by
Lemma~\ref{lem:diffact} that there is a strictly smaller $q'$ which entails $q$
and has a match in $I'$ which is a submatch of~$h$; so in particular it involves
no element of~$\dom(I_0) \times G$. Now, by minimality of $q$, $q'$ cannot be a
counterexample query. So
there is $q'' \in \acq$, $\card{q''} \leq k$, where $q''$ entails $q'$ and $I \models
q''$. Now,
as $q''$
entails $q'$ and $q'$ entails $q$, then $q''$ entails $q$, so this contradicts
the fact that $q$ was a counterexample.

Hence, there is no such counterexample query~$q$, and $I'$ is indeed $k$-sound for
$\cq$. This concludes the proof.

\mysubsection{Proof of Lemma~\ref{lem:ksimhomom} (Lifting $k$-bounded simulations
to the quotient)}
\label{apx:prf_ksimhomom}

\ksimhomom*

Fix the instance $I$ and the $k$-bounded simulation $\cov$ to an instance $I'$, and
consider $I'' \defeq \quot{I}{\bbsim_k}$. We show that there is a
$k$-bounded simulation $\cov'$ from $I''$ to $I$, because $\cov
\circ \cov'$ would then be a $k$-bounded simulation from $I''$ to $I'$, the
desired claim. We define $\cov'(A)$ for all $A \in I''$ to be $a$ for any member
$a \in A$ of the equivalence class~$A$, and show that $\cov'$ thus defined is
indeed a $k$-bounded simulation.

We will show the stronger result that $(I'', A) \bsim_{k} (I, a)$ for all $A
\in \dom(I'')$ and for any $a \in
A$. We do it by proving, by induction on~$0 \leq k' \leq k$,
that $(I'', A) \bsim_{k'} (I, a)$ for all $A \in \dom(I'')$ and $a \in
A$.
The case $k' = 0$ is trivial.
Hence, fix $0 < k' \leq k$, assume that $(I'', A) \bsim_{k'-1} (I, a)$ for all $A
\in \dom(I'')$ and $a \in A$, and show that this is also true for $k'$. Choose
$A \in \dom(I'')$, $a \in A$, and show that $(I'', A) \bsim_{k'} (I, a)$. To do so,
consider any fact $F = R(\mybf{A})$ of $I''$ such that $A_p = A$ for
some $R^p \in \pos(R)$. Let $F' = R(\mybf{a}')$ be a fact of~$I$ that is a preimage
of $F$ by $\chi_{\bbsim_k}$, so that $a'_q \in A_q$ for all $R^q \in \pos(R)$.
We have $a'_p \in A$ and $a \in A$, so that $a'_p \bbsim_k a$
holds in $I$. Hence, in particular we have $(I, a'_p) \bsim_{k'} (I, a)$ because
$k' \leq k$, so there
exists a fact $F'' = R(\mybf{a}'')$ of $I$ such that $a''_p = a$ and $(I, a'_q)
\bsim_{k'-1} (I, a''_q)$ for all $R^q \in \pos(R)$. We show that $F''$ is a
witness fact for $F$. Indeed, we have $a''_p = a$. Let us now choose $R^q \in
\pos(R)$ and show that $(I'', A_q) \bsim_{k'-1} (I, a''_q)$. By induction
hypothesis, as $a'_q \in A_q$, we have $(I'', A_q) \bsim_{k'-1} (I, a'_q)$, and as
$(I, a'_q) \bsim_{k'-1} (I, a''_q)$, by transitivity we have indeed $(I'', A_q)
\bsim_{k'-1} (I, a''_q)$. Hence, we have shown that $(I'', A) \bsim_{k'} (I, a)$.

By induction, we conclude that $(I'', A) \bsim_{k} (I, a)$ for all $A
\in \dom(I'')$ and $a \in A$, so that there is indeed a $k$-bounded simulation
from $I''$ to $I$, which, as we have explained, implies the desired claim.

\mysubsection{Proof of the Cautiousness Lemma (Lemma~\ref{lem:cautious})}
\label{apx:prf_cautious}

\cautious*

Let $J_\f = (I_\f, \cov)$ be the aligned superinstance of~$I_0$
constructed by the Acyclic Universal Models Theorem (Theorem~\ref{thm:myuniv2}),
and show that it is cautious for $\chi_{\bbsim_k}$.

We first observe that the definition of cautiousness
(Definition~\ref{def:cautious}) can be generalized to apply to any function, and
not just homomorphisms. In this case,
writing $F = R(\mybf{a})$ and $F' = R(\mybf{a}')$,
we define cautiousness as requiring,
instead of $h(F) = h(F')$,
that $h(a_i) = h(a'_i)$ for all $1 \leq i \leq \arity{R}$,

Now, let $\chi_{\bbsim_k}'$ be the homomorphism from $\chase{I_0}{\ids}$ to its
quotient by $\bbsim_k$. (We distinguish it from $\chi_{\bbsim_k}$, which is the
homomorphism from $I_\f$ to $\quot{I_\f}{\bbsim_k}$.)
We first show that our construction ensures the following:

\begin{lemma}
  \label{lem:covcautious}
  $I_\f$ is cautious for $\chi_{\bbsim_k}' \circ \cov$.
\end{lemma}

In other words, whenever two facts $F = R(\mybf{a})$ and $F' = R(\mybf{b})$
overlap in $I_\f$ and are not both in $I_0$, then, for any position $R^p \in
\pos(R)$, we have $\cov(a_p) \bbsim_k \cov(b_p)$ in $\chase{I_0}{\ids}$.

\begin{proof}
  In the proof of the Acyclic Universal Models Theorem
  (Theorem~\ref{thm:myuniv2}), $I_\f$ is constructed by first
  constructing an instance $I$
  using the Sufficiently Envelope-Saturated Solutions Proposition
  (Proposition~\ref{prp:preproc}),
  and then completing $I$
  using the Envelope-Thrifty Completion Proposition
  (Proposition~\ref{prp:etcomp}).

  Thus, we first check that this claim holds for~$I$.
  Indeed, we check it for each instance constructed in Lemma~\ref{lem:oneenv},
  and the only overlapping facts in each such instance
  which are not in $I_0$ are the $h(F_i)$,
  which all map to $\bbsim_k$-equivalent $\cov$-images.
  Hence, as $I$ is the disjoint union
  of the instances constructed in Lemma~\ref{lem:oneenv},
  we deduce that the claim holds for~$I$.

  Second, in the proof of the Envelope-Thrifty Completion Proposition,
  we only perform envelope-thrifty chase steps.
  By their definition, whenever we create a new fact $F_\n$ for a fact
class $D$, the only
elements of $F_\n$ that can be part of an overlap between $F_\n$ and an
existing fact are envelope elements, appearing at the one position at which
they appear in $\calE(D)$.
Then, by the last condition in the definition of envelopes
(Definition~\ref{def:envelope}), we deduce that the two overlapping facts
achieve the same fact class, which is what we wanted to show.
\end{proof}

We now want to show that two elements in $J_\f$ having $\bbsim_k$-equivalent
$\cov$ images in $\chase{I_0}{\ids}$
must themselves be $\bbsim_k$-equivalent in $J_\f$. We do it by showing that, in
fact, for any $a \in \dom(J_\f)$, not only do we have
$(I_\f, a) \bsim_k (\chase{I_0}{\ids}, \cov(a))$, but we also have the reverse:
$(\chase{I_0}{\ids}, \cov(a)) \bsim_k (I_\f, a)$.
In other words, intuitively, the facts of the chase must be ``mirrored'' in $I_\f$.

We define the \deft{ancestry} $\calA_F$ of a fact $F$ in $\chase{I_0}{\ids}$ as
$I_0$ plus the facts of the path in the chase forest that leads to $F$
(if $F \in I_0$ then $\calA_F$ is just $I_0$).
The \deft{ancestry} $\calA_a$ of~$a \in
\dom(\chase{I_0}{\ids})$ is that of the fact where $a$ was introduced. 

We now claim the following:

\begin{lemma}
  \label{lem:ancesthom}
  For any $a \in \dom(I_\f)$, there is a homomorphism $h_a$ from $\calA_{\cov(a)}$
  to $I_\f$ such that $h_a(\cov(a)) = a$.
\end{lemma}

\begin{proof}
  We prove that this property holds on $I_\f$, by first showing that it is true of the
  instance constructed in the Sufficiently Envelope-Saturated Solutions Proposition
  (Proposition~\ref{prp:preproc}).
  This is clearly the case because the instances created by
  Lemma~\ref{lem:oneenv} are just truncations of the chase where some elements
  are identified.

  Second, we show that the property is maintained by the construction of the
  Envelope-Thrifty Completion Proposition. We show the stronger claim that it is
  preserved by any thrifty chase step (Definition~\ref{def:thrifty}).
  Consider a thrifty chase step where,
  in a state $J_1 = (I_1, \cov_1)$
  of the construction of our aligned superinstance,
  we apply a $\uid$ $\tau: \ui{R^p}{S^q}$ to a fact $F_\a = R(\mybf{a})$
  to create a fact $F_\n = S(\mybf{b})$
  and obtain the aligned superinstance $J_2 = (I_2, \cov_2)$.
  Consider the chase witness $F_\w = S(\mybf{b}')$. By
  Lemma~\ref{lem:wadef}, $b'_q$ is the exported element
  between $F_\w$ and its parent in $\chase{I_0}{\ids}$.
  So we know that for any $i \neq q$,
  we have $\calA_{b'_i} = \calA_{b'_q} \sqcup \{F_\w\}$.

  We need to show that the property holds for the $b_i$ that are fresh
  (otherwise we already know that the property is satisfied,
  as adding more facts cannot violate the property in~$J_2$
  on an element for which it held in~$J_1$). 
  So, if none of the $b_i$ are fresh, there is nothing to do.
  Otherwise, choose $i$ such that $b_i$ is fresh. 
  By the definition of thrifty chase steps,
  we have set $\cov(b_i) \defeq b'_i$.
  Because $a_p = b_q$ is in $\dom(I_1)$,
  we know that there is a homomorphism $h_{b_q}$
  from~$\calA_{\cov(b_q)} = \calA_{b'_q}$ to~$I_1$
  such that we have $h(b'_q) = b_q$.
  We extend $h_{b_q}$ to the homomorphism $h_{b_i}$
  from~$\calA_{b'_i} = \calA_{b'_q} \sqcup \{F_\w\}$ to~$I_2$
  such that $h_{b_i}(b'_i) = b_i$,
  by setting $h_{b_i}(F_\w) \defeq F_\n$
  and $h_{b_i}(F) \defeq h(F)$ for any other~$F$ of~$\calA_{b'_i}$;
  we can do this because, by definition of the chase,
  $F_\w$ shares no element with the other facts of $\calA_{b'_i}$
  (that is, with $\calA_{b'_q}$),
  except $b'_q$ for which our definition coincides
  with the existing image.
  This proves the claim.
\end{proof}

We claim that this property implies the following:

\begin{corollary}
  For any $a \in \dom(I_\f)$, there is a homomorphism $h_a$ from
  $\chase{I_0}{\ids}$ to $I_\f$ such that $h_a(\cov(a)) = a$. 
\end{corollary}

\begin{proof}
  Choose $a \in \dom(I_\f)$ and let us construct $h_a$.
Let $h'_a$ be the homomorphism
from $\calA_{\cov(a)}$ to $I_\f$ with $h'_a(\cov(a)) = a$ whose existence was
proved in Lemma~\ref{lem:ancesthom}. Now start by setting $h_a \defeq h'_a$, and
extend $h'_a$ to be the desired homomorphism,  fact by fact,
using the property that $I_\f \models \ids$: for any $b \in
\dom(\chase{I_0}{\ids})$ not in the domain of $h'_a$ but which was
introduced in a fact $F$ whose exported element $c$ is in
the current domain of~$h'_a$, let us extend $h'_a$ to the elements of $F$ in the
following way: consider the parent fact $F'$ of~$F$ and its match
by $h'_a$, let $\tau$ be the $\uid$ used to create $F'$ from $F$, and, because
$I_\f \models \tau$, there must be a suitable fact $F''$ to extend $h'_a$ to all
elements of $F$ by setting $h'_a(F) \defeq F''$;
this is consistent with the image of~$c$ previously defined in~$h'_a$.
Performing this process allows us to define the desired homomorphism~$h_a$.
\end{proof}

Clearly this result implies:

\begin{corollary}
  \label{cor:covrev}
  For any $a \in \dom(I_\f)$, we have $(\chase{I_0}{\ids}, \cov(a)) \bsim_k
  (I_\f,
  a)$.
\end{corollary}

\begin{proof}
  Consider the restriction of~$h_a$ to the neighborhood at distance $k$ in the
  Gaifman graph of~$\cov(a)$.
\end{proof}

We are now ready to show our desired claim:

\begin{lemma}
  \label{lem:covbbsim}
  For any $a, b \in \dom(I_\f)$, if $\cov(a)
\bbsim_k \cov(b)$ in $\chase{I_0}{\ids}$, then $a \bbsim_k b$ in $I_\f$.
\end{lemma}

\begin{proof}
  Fix $a, b \in \dom(I_\f)$.
  We have $(I_\f, a) \bsim_k (\chase{I_0}{\ids}, \cov(a))$
  because $\cov$ is a $k$-bounded simulation;
  we have $(\chase{I_0}{\ids}, \cov(a)) \bsim_k (\chase{I_0}{\ids}, \cov(b))$
  because $\cov(a) \bbsim_k \cov(b)$;
  and we have $(\chase{I_0}{\ids}, \cov(b)) \bsim_k (I_\f, b)$
  by Corollary~\ref{cor:covrev}.
  By transitivity, we have $(I_\f, a) \bsim_k (I_\f, b)$.
  The other direction is symmetric,
  so the desired claim follows.
\end{proof}

We prove Lemma~\ref{lem:ksimhomom} immediately from Lemma~\ref{lem:covcautious} and
Lemma~\ref{lem:covbbsim}.

\mysubsection{Proof of the Mixed Product Preservation Lemma
  (Lemma~\ref{lem:mixedprod})}
\label{apx:prf_mixedprod}

\mixedprod*

Write $I_\m \defeq (I, I_0) \mprod^h G$.

If $\tau$ is a $\uid$, the claim is immediate even without the cautiousness
hypothesis. (In fact, the analogous claim could even be proven for the simple
product.) Indeed, for any $a \in
\dom(I)$ and $R^p \in \positions(\sigma)$, if $a \in \pi_{R^p}(I)$ then
$(a, g) \in \pi_{R^p}(I_\m)$ for all $g \in G$; conversely, if
$a \notin \pi_{R^p}(I)$ then
$(a, g) \notin \pi_{R^p}(I_\m)$ for all $g \in G$. Hence,
letting $\tau : \ui{R^p}{S^q}$ be a $\uid$ of~$\ids$, if there is $(a, g) \in
\dom(I_\m)$ such that $(a, g) \in \pi_{R^p}(I_\m)$ but $(a, g) \notin
\pi_{S^q}(I_\m)$ then $a \in \pi_{R^p}(I)$ but $a \notin \pi_{S^q}(I)$. Hence any
violation of~$\tau$ in $I_\m$ implies the existence of a violation of~$\tau$
in $I$, so we conclude because $I \models \tau$.

Assume now that $\tau$ is a $\fd$ $\phi : R^L \rightarrow R^r$.
Assume by contradiction that there are two facts
$F_1 = R(\mybf{a})$ and $F_2 = R(\mybf{b})$ in~$I_\m$
that violate~$\phi$, i.e., we have $a_l = b_l$ for all $l \in L$,
but $a_r \neq b_r$.
Write $a_i = (v_i, f_i)$ and $b_i = (w_i, g_i)$ for all $R^i \in \pos(R)$.
Consider $F_1' \defeq R(\mybf{v})$ and $F_2' \defeq R(\mybf{w})$ 
the facts of~$I$ that are the images of~$F_1$ and~$F_2$
by the homomorphism from~$I_\m$ to~$I$
that projects on the first component.
As $I \models \tau$, $F_1'$ and $F_2'$ cannot violate~$\phi$,
so as $v_l = w_l$ for all $l \in L$, we must have $v_r = w_r$.
Further, we have $\pi_{R^{l_0}}(F_1') = \pi_{R^{l_0}}(F_2')$
for any $l_0 \in L$;
hence, as $I$ is cautious for~$h$,
either $F_1', F_2' \in I_0$ or $h(F_1') = h(F_2')$.

In the first case,
by definition of the mixed product,
there are $f, g \in G$
such that $f_i = f$ and $g_i = g$ for all $R^i \in \pos(R)$.
Thus, taking any $l_0 \in L$, as we have $a_{l_0} = b_{l_0}$,
we have $f_{l_0} = g_{l_0}$, so $f = g$, which implies that $f_r = g_r$.
Hence, as $v_r = w_r$,
we have $(v_r, f_r) = (w_r, g_r)$,
contradicting the fact that $a_r \neq b_r$.

In the second case,
as $h$ is the identity on $I_0$
and maps $I \backslash I_0$ to $I' \backslash I_0$,
$h(F_1') = h(F_2')$ implies that
either $F_1'$ and $F_2'$ are both facts of~$I_0$
or they are both facts of $I \backslash I_0$;
but we have already excluded the former possibility in the first case,
so we assume the latter.
Let $F$ be $h(F_1')$.
By definition of the mixed product,
there are $f, g \in G$ such that
$f_i = f \cdot \l^{h(F)}_i$ and $g_i = g \cdot \l^{h(F)}_i$
for all $R^i \in \pos(R)$.
Picking $l_0 \in L$, from $a_{l_0} = b_{l_0}$,
we deduce that $f \cdot \l^{h(F)}_{l_0} = g \cdot \l^{h(F)}_{l_0}$,
which simplifies to $f = g$. Hence, $f_r = g_r$
and we conclude like in the first case.

\mysubsection{Proof of the Mixed Product Homomorphism Lemma
(Lemma~\ref{lem:mixedhom})}
\label{apx:prf_mixedhom}

\mixedhom*

We use the homomorphism $h: I \to I_1$
to define the homomorphism $h'$ from $I_\m \defeq (I, I_0) \mprod^h G$ to $I_\p
\defeq (I, I_0) \sprod G$
by $h'((a, g)) \defeq (h(a), g)$
for every $(a, g) \in \dom(I) \times G$.

Consider a fact $F = R(\mybf{a})$ of~$I_\m$, with $a_i
= (v_i, g_i)$ for all $R^i \in \pos(R)$. Consider its image $F' = R(\mybf{v})$
by the homomorphism from $I_\m$ to
$I$ obtained by projecting to the first component,
and the image $h(F')$ of $F'$ by the homomorphism~$h$.
As $\restr{h}{I_0}$ is the identity and $\restr{h}{(I \backslash I_0)}$ maps
to $I_1 \backslash I_0$, $h(F')$ is a fact of~$I_0$ iff $F'$ is.
Now by definition of the simple product it is clear that $I_\p$ contains the
fact~$h'(F)$ (it was created in~$I_\p$ from~$h(F')$
for the same choice of $g \in G$).

The fact that $h$ is the identity on $I_0$ also ensures that $h'$ is the
identity on~$I_0 \times G$.

\end{document}